\documentclass[a4paper,onecolumn,11pt,accepted=2024-03-06]{quantumarticle}
\pdfoutput=1

\usepackage[T1]{fontenc}
\usepackage[utf8]{inputenc}
\usepackage{amsmath,amsfonts,amsthm,amssymb}
\usepackage{mathtools}
\usepackage{subeqnarray}
\usepackage{setspace}
\usepackage{graphics,graphicx,color}
\usepackage{url}
\usepackage{enumerate}
\usepackage{float} 
\usepackage{multirow}
\usepackage{hhline}
\usepackage[margin=1.2in]{geometry}
\usepackage{braket}
\usepackage{dsfont} 
\usepackage{multirow}
\usepackage{hhline}
\usepackage[makeroom]{cancel}
\usepackage{ifthen}
\usepackage{comment}
\usepackage{breakcites}
\usepackage[usenames,dvipsnames]{xcolor}
\usepackage[numbers,sort&compress]{natbib}

\definecolor{ddgreen}{rgb}{.05,.4,.05}
\definecolor{damethyst}{rgb}{0.4, 0.2, 0.6}
\usepackage[colorlinks, linkcolor=damethyst,citecolor=ddgreen]{hyperref}

\newtheorem{theorem}{Theorem}
\newtheorem{cor}{Corollary}
\newtheorem{definition}{Definition}

\newtheorem{lem}{Lemma}

\newtheorem{construction}{Construction}
\newtheorem{example}{Example}
\newtheorem{remark}[theorem]{Remark}

\usepackage[scr=boondoxo,scrscaled=1.05]{mathalfa}

\newcommand{\Tr}[1]{\mbox{$\mathrm{Tr}\Big[#1\Big]$}}

\newcommand{\Id}{\mathds{1}}
\newcommand{\beq}{\begin{eqnarray}}
\newcommand{\eeq}{\end{eqnarray}}

\newcommand{\Hmin}{\mathbf{H}_{\min}}
\newcommand{\POVM}{\mathsf{POVM}}\newcommand{\CPTP}{\mathsf{CPTP}}
\newcommand{\QPT}{\mathsf{QPT}}

\newcommand{\SSL}{\mathsf{SSL}}
\newcommand{\Lease}{\mathsf{Lease}}
\newcommand{\Verify}{\mathsf{Verify}}
\newcommand{\Gen}{\mathsf{Gen}}
\newcommand{\Pieq}{\Pi_\theta^{\mathsf{eq}}}

\newcommand{\Eval}{\mathsf{Eval}}

\newcommand{\sk}{\mathsf{sk}}

\newcommand{\poly}{\mathsf{poly}}
\newcommand{\negl}{\mathsf{negl}}
\newcommand{\EPR}{\mathsf{EPR}}
\newcommand{\KeyGen}{\mathsf{KeyGen}}
\newcommand{\Enc}{\mathsf{Enc}}
\newcommand{\Dec}{\mathsf{Dec}}
\newcommand{\SKQES}{\mathsf{SKQES}}
\newcommand{\QECM}{\ensuremath{\mathsf{QECM}}}

\usepackage{tikz}
\def\firstcircle{(210:1.75cm) circle (2.5cm)}
\def\secondcircle{(330:1.75cm) circle (2.5cm)}

\definecolor{dred}{rgb}{.7,.2,.2}

\definecolor{amethyst}{rgb}{0.6, 0.4, 0.8}

\definecolor{amethyst}{rgb}{0.6, 0.4, 0.8}

\title{Quantum copy-protection of compute-and-compare programs in the quantum random oracle model}

\author{Andrea Coladangelo}
\affiliation{University of California, Berkeley, USA.}
\email{andrea.coladangelo@gmail.com}

\author{Christian Majenz}
\affiliation{QuSoft and Centrum Wiskunde \& Informatica, The Netherlands.}
\affiliation{Department of Applied Mathematics and Computer Science, Technical University of Denmark, Denmark.}
\email{christian.majenz@cwi.nl}

\author{Alexander Poremba}
\affiliation{Computing and Mathematical Sciences, Caltech, USA.}
\email{aporemba@caltech.edu}

\begin{document}

\maketitle

\abstract{Copy-protection allows a software distributor to encode a program in such a way that it can be evaluated on any input, yet it cannot be ``pirated'' --  a notion that is impossible to achieve in a classical setting. Aaronson (CCC 2009) initiated the formal study of quantum copy-protection schemes, and speculated that quantum cryptography could offer a solution to the problem thanks to the quantum no-cloning theorem.

In this work, we introduce a quantum copy-protection scheme for a large class of evasive functions known as ``compute-and-compare programs'' -- a more expressive generalization of point functions. A compute-and-compare program $\mathsf{CC}[f,y]$ is specified by a function $f$ and a string $y$ within its range: on input $x$, $\mathsf{CC}[f,y]$ outputs $1$, if $f(x) = y$, and $0$ otherwise. We prove that our scheme achieves non-trivial security against fully malicious adversaries in the quantum random oracle model (QROM), which makes it the first copy-protection scheme to enjoy any level of provable security in a standard cryptographic model. 
As a complementary result, we show that the same scheme fulfills a weaker notion of software protection, called ``secure software leasing'', introduced very recently by Ananth and La Placa (eprint 2020), with a standard security bound in the QROM, i.e. guaranteeing negligible adversarial advantage. Finally, as a third contribution, we elucidate the relationship between
unclonable encryption and copy-protection for multi-bit output point
functions.
}

\newpage
\tableofcontents

\section{Introduction}

\subsection{Copy-protection}

Copy-protection captures the following cryptographic task. A vendor wishes to encode a program in such a way that a user who receives the encoded program is able to run it on arbitrary inputs. However, the recipient should not be able to create functionally equivalent ``pirated'' copies of the original program. More concretely, no user should be able to process the encoded program so as to split it into two parts, each of which allows for the evaluation of the function implemented by the original program. 
Copy-protection of any kind is trivially impossible to achieve classically. This is because any classical information that the user receives can simply be copied. In the quantum realm, however, the no-cloning theorem prevents any naive copying strategy from working unconditionally, and copy-protection seems, at least in principle, possible. The key question then becomes:
\emph{Is it possible to encode the functionality of a program in a quantum state while at the same time preserving the no-cloning property?}

To be precise, we are not satisfied with preventing an adversary from copying the state that encodes the program (this is certainly a necessary condition), but we also require that there is no other way for a (computationally bounded) user to process the state into two parts (not necessarily a copy of the original) so as to allow each half to recover the input-output behaviour of the encoded program.

Quantum copy-protection was first formally considered by Aaronson \cite{aaronson2009quantum}. One of the first observations there is that families of \emph{learnable} functions cannot be copy-protected: access to a copy-protected program, and hence its input-output behaviour, allows one to recover a classical description of the program itself, which can be copied. In \cite{aaronson2009quantum}, Aaronson provides some formal definitions and constructions of copy-protection schemes. More precisely, Aaronson describes:
\begin{itemize}
\item A provably secure scheme to copy-protect any family of efficiently computable functions which is not quantumly learnable, assuming a \textit{quantum} oracle implementing a certain family of unitaries. 
\item Two candidate schemes to copy-protect point functions in the plain model, although neither of the two has a proof of security.
\end{itemize}
In recent work \cite{aaronson2020quantum}, Aaronson, Liu and Zhang provide a scheme to copy-protect any family of efficiently computable functions which is not quantumly learnable, assuming access to a \textit{classical} oracle, i.e. an oracle (which can be queried in superposition) that implements a classical function. We emphasize, however, that this classical function is dependent on the function that one wishes to copy-protect. In particular, none of the oracles that these schemes rely upon have candidate realizations in the plain model, and some, in particular, are impossible to  realize. For example, the classical oracle used in the scheme from \cite{aaronson2020quantum} can be used to construct an ideal obfuscator for the function $f$ that is being copy-protected. Such an ideal obfuscator is impossible to realize in general---with the exception of \emph{learnable programs}~\cite{barak2012possibility} and highly restricted functionalities~\cite{10.5555/1760749.1760766}.
Some of the questions left open in \cite{aaronson2009quantum} and \cite{aaronson2020quantum} are:
\begin{itemize}
    \item[(i)] Does there exist a scheme to copy-protect any non-trivial family of functions (the simplest example being point functions) which we can prove secure in the plain model under some standard assumption? What about larger classes of programs?
    \item[(ii)] Can such a scheme exist that does not involve multi-qubit entanglement?
\end{itemize}

On the negative side, aside from the impossibility of copy protecting families of learnable functions, it has remained an open question to determine whether a more general impossibility result applies. In a recent result, Ananth and La Placa \cite{ananth2020secure} prove that a universal copy-protection scheme cannot exist, assuming the quantum hardness of the learning with errors problem \cite{Regev_LWE} and the existence of quantum fully homomorphic encryption~\cite{DBLP:conf/focs/Mahadev18a,10.1007/978-3-319-96878-0_3}.

\subsection{Our contributions}
In this work, we approach copy-protection from the positive side. Our main result is a copy-protection scheme for compute-and-compare programs, for which we prove non-trivial security in the quantum random oracle model. By non-trivial security we mean, informally, that a query-bounded adversary fails at pirating with at least some constant probability (which is approximately $10^{-4}$ for our scheme).

A desirable feature of our scheme is that the copy-protected program does not involve multi-qubit entanglement -- in fact it only involves BB84 states and computational and Hadamard basis measurements. This is in contrast to previous candidate schemes for point functions in \cite{aaronson2009quantum}, whose security is only conjectured, and which employ highly entangled states. 
The simple structure of the copy-protected program is advantageous for, e.g., error-corrected storage of the copy-protected program. We point out, however, that in a practical implementation of our scheme, where the oracle is replaced by a hash function, evaluation of the copy-protected program on an input requires \emph{coherently} computing the hash function in an auxiliary register. This operation requires universal quantum computation.

Our scheme is not in the plain model, and hence does not fully resolve questions (i) and (ii). The (quantum) random oracle model, however, enjoys widespread acceptance and popularity in (post-quantum/quantum) cryptography, and many schemes designed for, and deployed in, practical applications enjoy provable security in that model only. Our security definition is essentially analogous to the original definition in \cite{aaronson2009quantum} but differs more significantly from the more recent definition in \cite{aaronson2020quantum}, which is weaker.

Our techniques and construction are inspired by recent work on \textit{unclonable encryption} by Broadbent and Lord \cite{broadbent2019uncloneable}. The main technical ingredient on which their construction relies are \emph{monogamy of entanglement games}, introduced and studied extensively in \cite{tomamichel2013monogamy}, which they combine with an adaption of the one-way-to-hiding (O2H) lemma of \cite{Unruh15} for a security analysis in the quantum random oracle model. In a nutshell, (a special case of) the latter lemma allows one to upper bound the probability that an algorithm outputs $H(x)$, where $H$ is a random oracle and $x$ is any string in the domain, in terms of the probability that the algorithm ``queries'' at $x$ at some point during its execution. The adaption of \cite{broadbent2019uncloneable} extends the applicability of the O2H lemma to a setting that involves \emph{two players}, and upper bounds the probability that the two (possibly entangled) players \emph{simultaneously} guess $H(x)$ by the probability that they both query at $x$ at some point during the execution of their respective strategies.

Our main technical contribution in this work is that we augment the analysis of this ``simultaneous one-way-to-hiding'' lemma by a search-to-decision reduction. This allows us to overcome an important hurdle that is inherently present in the security analysis of copy-protection schemes, namely that security is based on a distinguishing game (the ``freeloaders'' who receive pirated copies have to return a single bit, namely the evaluation of the original program on a ``challenge'' input), rather than a guessing game (freeloaders who receive pirated copies have to guess some long string). Our reduction is ``lossy'': we show that the two players fail at least with some \emph{constant} probability at simultaneously evaluating correctly on their respective inputs. We give an informal description of this contribution at the end of the next subsection. We believe that obtaining a reduction in which the advantage of the two players at simultaneously deciding correctly is \emph{negligible} requires further novel techniques. Along the way, we generalize the analysis of Broadbent and Lord to certain random oracles with non-uniform distributions.

\subsubsection{A sketch of our copy-protection scheme}
We start by describing a scheme to copy-protect point functions, which is the crux of this work. Subsequently, we describe how to extend this to compute-and-compare programs. Our scheme is inspired by Broadbent and Lord's unclonable encryption scheme \cite{broadbent2019uncloneable}, which is itself rooted in Wiesner's conjugate coding scheme \cite{Wiesner83}. 

Let $\lambda \in \mathbb{N}$. The main idea is that it is possible to ``hide'' a string $v \in \{0,1\}^{\lambda}$ by making $\lambda$ uniformly random choices of basis (either computational or Hadamard) which we denote by $\theta \in \{0,1\}^{\lambda}$, and then encoding each bit of $v$ in either the computational or Hadamard basis, according to $\theta$. Formally, this amounts to preparing the following quantum state on $\lambda$ qubits:
$$\ket{\Psi} = \bigotimes_{i=1}^{\lambda} \ket{v_i^{\theta_i}},$$
where $\ket{b^{s}} = H^s \ket{b}$, for $b,s \in \{0,1\}$. Given the string of basis choices $\theta$ and the state $\ket{\Psi}$, one is able to ``decrypt'' and recover the string $v$ by measuring each qubit of $\ket{\Psi}$ in the basis specified by $\theta$.
We can bootstrap this idea to copy-protect point functions as follows. Let $P_y$ be a point function with marked input $y \in \{0,1\}^\lambda$, i.e.
$$ P_y (x) = \begin{cases} 1   & \text{if } x = y\,,\\
    0 &\text{if } x \neq y \,.   \end{cases}  $$
Our scheme rests on the following simple idea: we interpret the string $y$ as the basis choice that ``encrypts'' a uniformly random string $v$. The copy-protected version of $P_y$ then consists of the state $\ket{\Psi}$
together with some classical information that enables an evaluator to ``recognize'' $v$. One can take the latter information to be $H(v)$, for some hash function $H$ (or a uniformly random function $H$, if one works in the random oracle model). Then, to evaluate the program on some input $x$, the evaluator attempts to ``decrypt'' using $x$ as the basis choice, i.e. applies Hadamards $H^{x}=H^{x_1} \otimes \dots \otimes H^{x_\lambda}$ to $\ket{\Psi}$, followed by a measurement in the computational basis. Let $v' \in \{0,1\}^{\lambda}$ be the outcome of this measurement. The evaluator checks that $H(v') = H(v)$, and outputs $1$ if so, $0$ otherwise.\footnote{To ensure that the quantumly copy-protected program can be reused, $v'$ is not actually measured, rather, the check $H(v') = H(v)$ is performed coherently (only the result of the check is measured). For details, see Section \ref{sec: main}.} Except for a minor modification which we highlight in the next paragraph, this is our scheme (described in detail in Construction \ref{cons:cp}). We will now informally discuss the correctness and the copy-protection property of this scheme.

\paragraph{Correctness.} Informally, the scheme is correct since ``decrypting'' using $y$ will result in the correct string $v$ with certainty, whereas if one tries to ``decrypt'' using $x \neq y$, the outcome will most likely be a string $v' \neq v$ with $H(v') \neq H(v)$ (provided $H$ has a large enough range). There is a slight issue with this approach, namely that an $x$ which is, for instance, equal to $y$ everywhere except for a single bit, will result in an honest evaluator outputting $1$ (i.e the incorrect output) with probability $\frac12$. To circumvent this, instead of having $y$ itself be the choice of basis for the encoding, we have the latter be $G(y)$, where $G$ is hash function whose range is sufficiently larger than the domain\footnote{Such a hash function can, e.g., be obtained using the sponge construction as in SHA3, but by extending the so-called squeezing phase.}. For the rest of the discussion in this section we will omit $G$ for ease of exposition.

\paragraph{Copy-protection.} The copy-protection property crucially leverages the following property: it is impossible for any pirate who has $\ket{\Psi} = \bigotimes_{i=1}^{\lambda} \ket{v_i^{y_i}}$ but does not know $y$, to produce a state on two registers $\textsf{AB}$ such that two ``freeloaders'' Alice and Bob with access to registers $\textsf{A}$ and $\textsf{B}$ respectively, \emph{as well as} access to $y$, can simultaneously recover $v$. Note that the latter property crucially holds even when both Alice and Bob are \emph{simultaneously} receiving $y$. This property is essentially a consequence of the \emph{monogamy of entanglement}. At a high level, one can consider a purification of the state received by the pirate (when averaging over the choice of $v$ and $y$). Let \textsf{C} be the purifying register, which we can think of as being held by the vendor. Since the pirate does not have access to the register containing the choice of basis $y$, it is possible to argue that the only way for the state of register $\textsf{A}$ to allow for recovery of $v$ (with high probability over the basis choice) is if $\textsf{A}$ is (close to) maximally entangled with $\textsf{C}$. The same argument applies to $\textsf{B}$. Hence, the monogamy of entanglement prevents Alice and Bob from recovering $v$ simultaneously.
This property is captured formally in \cite{tomamichel2013monogamy}, via the study of \emph{monogamy of entanglement games} (Section \ref{sec: monogamy}). In particular, a rephrasing of the results of \cite{tomamichel2013monogamy} is that, for any (unbounded) strategy of the pirate, and Alice and Bob, the probability that both Alice and Bob are able to output $v$ is exponentially small. Unfortunately, our proof of security does not immediately follow from the above observations, mainly due to the fact that the encoded program also consists of the classical string $H(v)$, which further complicates the matter.

Before we expand on the technical hurdles we encounter when proving the copy-protection property of our scheme, let us first define security in a bit more detail. Informally, a quantum copy-protection scheme is \textit{secure} for a family of circuits $\mathcal{C}$ (as well as a distribution $D$ over $\mathcal{C}$) if no adversary---consisting of a triple of quantum polynomial time algorithms $(\mathcal{P},\mathcal{F}_1,\mathcal{F}_2)$, a ``pirate'' $\mathcal{P}$ and two ``freeloaders'' $\mathcal{F}_1$ and $\mathcal{F}_2$---can
succeed with sufficiently high (i.e., non-trivial) probability at the following game:
\begin{itemize}
    \item The pirate $\mathcal{P}$ receives a copy-protected program $\rho_C$ from the challenger (where the program $C \in \mathcal{C}$ is sampled from $D$). $\mathcal{P}$ then creates a bipartite state on registers $\textsf{A}$ and $\textsf{B}$, and sends $\textsf A$ to $\mathcal{F}_1$ and $\textsf{B}$ to $\mathcal{F}_2$.
    \item The challenger samples a pair $(x_1,x_2)$ of inputs to $C$ from a suitable distribution (which is allowed to depend on $C$), and sends $x_1$ to $\mathcal{F}_1$ and $x_2$ to $\mathcal{F}_2$.
    \item $\mathcal{F}_1$ and $\mathcal{F}_2$, who are not allowed to communicate, return bits $b_1$ and $b_2$ respectively. 
    \item $(\mathcal{P}, \mathcal{F}_1, \mathcal{F}_2)$ win if $b_1  = C(x_1)$ and $b_2  = C(x_2)$.
\end{itemize}
The crux in proving that our scheme satisfies this security definition is in arguing that a strategy that performs well enough in the security game must be such that $\mathcal{F}_1$ and $\mathcal{F}_2$ are simultaneously querying the oracle $H$ at $v$ with significant probability (at some point during their executions). This would allow to construct a strategy that simultaneously ``extracts'' $v$, and thus breaks the monogamy of entanglement property.

Unruh's one-way-to-hiding (O2H) lemma \cite{Unruh15} is the standard tool to argue the above. One variant of the O2H lemma provides an upper bound on the probability that an adversary distinguishes $H(v)$ from a uniformly random string in the co-domain of $H$, in terms of the probability that such an adversary queries the oracle at $v$. However, in our security proof, this analysis needs to be augmented: we need to account for possibly \emph{entangled} strategies that attempt to distinguish $H(v)$ from a uniformly random string.
\begin{itemize}
    \item The main technical contribution of \cite{broadbent2019uncloneable} is an important step in this direction. There, the authors bound the probability of entangled parties simultaneously guessing $H(v)$ in terms of the probability that the two parties simultaneously query the oracle at $v$.
    \item The above is not entirely sufficient for our purpose: in our security game, the freeloaders are not asked to guess $H(v)$, rather they are only required to return a single bit. Note that an adversary who wins our security game (with high probability) \emph{must} be able to distinguish $H(v)$ from a uniformly random string. This is because an adversary who receives a uniformly random string instead of $H(v)$ cannot do better than random guessing in the security game. Thus, we need to provide a \textit{search-to-decision} reduction: an adversary such that the freeloaders are both able to distinguish  $H(v)$ from a uniformly random string (sufficiently well) can be used to construct an adversary such that the freeloaders simultaneously extract $v$. This is our main technical contribution (captured by Lemmas \ref{lem: reduction to Gr}, \ref{lem: main bound} and \ref{lem: Gr to monogamy}).
\end{itemize} 

Naively, the difficulties encountered when attempting simultaneous extraction might seem somewhat surprising given the fact that the problem is solved by a straightforward union bound in the classical setting. Indeed, if two algorithms both distinguish $H(v)$ from a uniformly random string with probability at least $3/4+\epsilon$, for some $\epsilon >0$, the probability that the respective query transcripts contain $v$ is at least $1/2+2\varepsilon$ for each of them. This guarantees that both transcripts contain $v$ simultaneously with probability at least $4\varepsilon$. For quantum queries, however, there exists no transcript, and measuring a query for extraction disturbs the run of the algorithms. But the extraction probability is too small for a union bound even in the classical case, as extraction requires choosing a query at random, which suppresses the success probability by a factor of $O(q^\alpha)$, where $q$ is the total number of oracle queries, and $\alpha=-1$ in the classical and $\alpha=-2$ in the quantum case.

The dependence of the simple classical technique on query transcripts, and the lack of such in the quantum case, in some sense capture the essence of a major difficulty we encounter. Let us therefore elaborate a bit further by noting that the classical technique mentioned above crucially depends on \emph{separate} query transcripts for the two players. This implies that it depends on information that can be elusive in the quantum case, in a very strong sense. Suppose e.g. the two players share an entangled state $(\ket 0\ket 1+\ket 1\ket 0)/\sqrt{2}$. They proceed by making a single query each, on input $x_i$ controlled on their part of the entangled state being in state $i$, for two arbitrary inputs $x_0, x_1$. The final state is then
$$
\ket\psi=\frac 1{\sqrt{2}}\left(\ket 0\ket{x_0}\ket{H(x_0)}\ket 1\ket{x_1}\ket{H(x_1)}+\ket 1\ket{x_1}\ket{H(x_1)}\ket 0\ket{x_0}\ket{H(x_0)}\right),
$$
where the first, and the last, three registers are held by the first, and the second, player, respectively. 
At this point, the inputs $x_0$ and $x_1$ have both been queried with certainty, but the information of who has queried which of the two is, in fact, \emph{distributed quantum} information in the hand of the players, precluding any third-party knowledge about it due to the no-cloning theorem. The difference between global and individual query transcripts also becomes evident in the recently developed superposition oracle framework \cite{Zhandry-how}. In the described example, the fact that both $x_0$ and $x_1$ have been queried by somebody can be recovered using the superposition oracle framework. More generally, the superposition oracle framework can be used as a replacement for query transcripts, sometimes in a quite straight-forward manner (see e.g. \cite{Zhandry-how,bindel2019,Alagic20,czajkowski2019quantum,Hosoyamada19}). However, the described example illustrates that recording \emph{which input} was queried \emph{at which interface} is incompatible with the correctness of any quantum-accessible random oracle simulation.

\subsubsection{Extension to compute-and-compare programs.}\label{subsubsec:CC} The copy-protection scheme for point functions we described in the previous section can be straightforwardly extended to the more general class of compute-and-compare programs~\cite{Wichs_Zirdelis_Compute_and_Compare, goyal2017lockable}. A compute-and-compare program $\mathsf{CC}[f,y]$ is specified by an efficiently computable function $f: \{0,1\}^n \rightarrow \{0,1\}^m$ and a string $y \in \{0,1\}^m$ in its range, where
\begin{equation*}
    \mathsf{CC}[f,y](x) = \begin{cases} 1    & \text{if } f(x) = y\,,\\
    0 &\text{if } f(x) \neq y \,.   \end{cases}
\end{equation*}
Point functions are a special case of compute-and-compare programs where the function $f$ is the identity map.

In Section \ref{sec:CnC}, we show how to copy-protect $\mathsf{CC}[f,y]$ in the following simple way: the copy-protected program consists of (a description of a circuit computing) $f$ in the clear, together with a copy-protected version of the point function with marked input $y$.
The intuition is that it is enough to copy-protect the marked input $y$ in order to render $\mathsf{CC}[f,y]$ unclonable. At first, it might seem surprising that one can give $f$ in the clear while preserving unclonability, as the encoded program now leaks significantly more information than its input/output behavior alone. 

However, at a second thought, it is in fact quite natural that one can render a functionality ``unclonable'' by just making some sufficiently important component of it unclonable. 
In fact, it is straightforward to show that copy-protection security of the extended construction reduces to security of the original point function scheme.

\paragraph{Quantum copy-protection of multi-bit point functions in the QROM.} 

The copy-protection scheme we described earlier only considers the case of single-point functions. Notice, however, that a copy-protection scheme for multi-bit point-functions is somewhat \textit{weaker} (and possibly significantly easier to build) than a copy-protection scheme for regular point-functions (i.e. with single-bit output). This is because in the multi-bit case, the freeloaders need to find the output string in addition to the marked input.

Our second contribution is to give a quantum copy-protection scheme for multi-bit point functions with security in the QROM. Here, a multi-bit point function $P_{y,m}$, for strings $y,m \in \{0,1\}^\lambda$, is given by:
$$ P_{y,m} (x) = \begin{cases} m   & \text{if } x = y\,,\\
    0^\lambda &\text{if } x \neq y \,.   \end{cases}  $$
    
Our construction is inspired by recent work on \textit{unclonable encryption} by Broadbent and Lord \cite{broadbent2019uncloneable}. Informally, an unclonable encryption scheme is an encryption scheme in which the ciphertext is a quantum state. In addition to standard notions of security, an unclonable encryption scheme satisfies the following security guarantee. It is not possible for a pirate who receives a valid ciphertext to make two copies of it such that two separate parties, each with one copy, are simultaneously able to decrypt, even when given the (classical) secret key. Our simple observation is that an \textit{unclonable encryption} scheme can be turned into a copy-protection scheme for multi-bit point functions as follows. To copy-protect $P_{y,m}$ encrypt message $m$ with secret key $y$. Then, provided there exists a mechanism for \textit{wrong-key detection}, this already achieves our goal: to evaluate at point $x$, attempt to decrypt using $x$; if decryption succeeds output the decrypted message, if decryption fails, output $0^\lambda$. We then observe that any unclonable encryption scheme can be easily upgraded to achieve wrong-key detection in the QROM, thus yielding the desired copy-protection scheme.  

\subsubsection{Secure software leasing}
On top of proving the impossibility of a general copy-protection scheme for all unlearnable functions, Ananth and La Placa introduce in \cite{ananth2020secure} a weaker notion of copy-protection, which they call ``secure software leasing''
$(\SSL)$. The sense in which the latter is weaker than copy-protection is that one assumes that the freeloaders $\mathcal{F}_1$ and $\mathcal{F}_2$ (now a single adversary) are limited to performing the honest evaluation procedure only. Rather than emphasizing the impossibility of simultaneous evaluation on inputs chosen by a challenger, $\SSL$ captures the essence of quantum copy-protection in the following scenario. An authority (the lessor) wishes to lease a copy $\rho_C$ of a classical circuit $C \in \mathcal{C}$ to a user (the lessee) who is supposed to return back $\rho_C$ at a later point in time, as specified by the lease agreement. Once the supposed copy is returned and verified by the lessor, the security property requires that the adversary can no longer compute $C$. More formally, no adversary should be able to produce a (possibly entangled) quantum state such that:
\begin{itemize}
    \item One half of the state is deemed valid by the lessor, once it is returned.
    \item The other half can be used to honestly evaluate $C$ on every input of the adversary's choosing.
\end{itemize}
Surprisingly, Ananth and La Placa  were able to show in \cite{ananth2020secure} that a general $\SSL$ scheme is also impossible, despite having weaker security requirements compared to copy-protection.
On the positive side, the authors describe an $\SSL$ scheme for a \emph{searchable class} of circuits (i.e., a large class of general evasive circuits) assuming the existence of subspace-hiding obfuscators~\cite{Zhandry-Quantum_Lightning} and the quantum hardness of the learning with errors problem~\cite{Regev_LWE}. Because subspace-hiding obfuscators are currently only known to exist under indistinguishability obfuscation~\cite{Zhandry-Quantum_Lightning, Sahai_Waters_iO}, the same applies to the security of the scheme proposed in \cite{ananth2020secure}.

\paragraph{Our work: $\SSL$ revisited.}
As we mentioned earlier, the original definition of secure software leasing in \cite{ananth2020secure} is a weaker version of copy-protection in the following two ways:
\begin{itemize}
\item The lessor performs a prescribed verification procedure on a register returned by the lessee.
    \item The lessee is required to perform the honest evaluation procedure with respect to any post-verification registers in the lessee's possession.
\end{itemize}

We revisit the notion of secure software leasing from a similar perspective as in our copy-protection definition. Our main contributions are the following. First, we introduce a new and intuitive $\SSL$ definition (Section~\ref{sec:SSL}) by means of a cryptographic security game which does not limit the adversary to performing the honest evaluation on any post-verification registers.\footnote{The $\SSL$ definition in \cite{ananth2020secure} is not ``operational'' and cannot be directly phrased as a security game.} Our definition remains faithful to the idea of $\SSL$, while at the same time offering a stronger security guarantee. Second, we show that our definition of security is achievable with a standard negligible security bound in the quantum random oracle model for the class of compute-and-compare programs (Section~\ref{sec:SSL}). Our $\SSL$ scheme (Construction \ref{cons: pf to cc - SSL}) is virtually equivalent to our copy-protection scheme, but is adapted to the syntax of $\SSL$. Informally, any $\SSL$ scheme $(\SSL.\Gen, \SSL.\Lease,\SSL.\Eval,\SSL.\Verify)$ according to our definition should satisfy the following property. After receiving a leased copy of a classical circuit $C$, denoted by $\rho_C$ (and generated using $\SSL.\Lease$), and a circuit for $\SSL.\Eval$, no adversary should be able to produce a (possibly entangled) quantum state $\sigma$ on two registers $\mathsf{R}_1$ and $\mathsf{R}_2$ such that:
\begin{itemize}
    \item $\SSL.\Verify$ deems the contents of register $\mathsf{R}_1$ of  $\sigma_{\mathsf{R_1R_2}}$ to be valid, and
    \item the adversary can predict the output of circuit $C$ (on challenge inputs chosen by the lessor) using an arbitrary measurement of the post-verification state in register $\mathsf{R}_2$.
\end{itemize}
In Section~\ref{sec:SSL} we show the following key property about our $\SSL$ scheme in Construction \ref{cons: pf to cc - SSL}: once a leased copy is successfully returned to the lessor, no adversary can distinguish the marked input of a compute-and-compare program from a random (non-marked) input with probability better than $1/2$, except for a negligible advantage (in the security parameter). Our scheme can thus be thought of as having perfect security for a natural choice of ``input challenge distribution''. 
The result follows from a standard application of the O2H lemma and a custom ``uncertainty relation'' variant of the monogamy of entanglement property which appeared in a work of Unruh~\cite{Unruh15}. The latter appears in similar contexts in the quantum key-distribution literature. Note that the technical complications arising in the proof of security of our original copy-protection scheme do not appear in the $\SSL$ security proof. Crucially, this is because we can leverage the fact that the lessor is performing a prescribed verification procedure.

\subsection{Related notions}

\paragraph{Obfuscation.} Copy-protection is related to program obfuscation \cite{barak2012possibility,AlagicFefferman16}, although the extent to which they relate to each other is still unclear. It seems plausible that some degree of ``obfuscation'' is necessary in order to prevent a pirate from copying a program. A natural question is whether there exist schemes that satisfy both notions simultaneously.

We give an affirmative answer to this question by showing that our quantum copy-protection scheme for point functions (Construction \ref{cons:cp}) also satisfies the notion of quantum virtual-black box $(\mathsf{VBB})$ obfuscation \cite{AlagicFefferman16}. This results in the first provably secure scheme which is simultaneously a quantum copy-protection scheme as well as a quantum obfuscation scheme, in the quantum random oracle model. 
The main idea is the following: any computationally bounded adversary cannot distinguish between $\left((\bigotimes_{i} \ket{v_i^{\theta_i}}) , H(v) \right)$ and $\left((\bigotimes_{i} \ket{v_i^{\theta_i}}) , z \right)$, where $z$ is a uniformly random string of the same size as $H(v)$, unless the adversary queries the oracle at $v$, which can only happen with negligible probability. Notice that, when averaging over $H, v, \theta$ and $z$, the state $\left((\bigotimes_{i} \ket{v_i^{\theta_i}}) , z \right)$ is maximally mixed. Thus, the copy-protected program is computationally indistinguishable from a maximally mixed state. We give a detailed proof in Section \ref{subsec:VBB}.

Conversely, there exist, of course, program obfuscation schemes which are not copy-protection schemes (since obfuscators for certain classes of programs are possible classically, while copy-protection schemes are impossible). Moreover, Aaronson
\cite{aaronson2009quantum} previously observed that any family of \textit{learnable} programs cannot be copy-protected: access to a copy-protected program (and hence its input-output behaviour) allows the user to recover the classical description of the program itself, which can be copied. But for precisely the same reason, \textit{learnable} programs can be $\mathsf{VBB}$ obfuscated by definition.

Finally, we also provide a missing piece that allows us to separate the notions of quantum copy-protection and quantum $\mathsf{VBB}$ obfuscation (Figure \ref{fig:separation}). Namely, we show that there exists a provably secure scheme (Construction \ref{cons: pf to cc}) to copy-protect compute-and-compare programs (in the quantum random oracle model) that does not, in general, satisfy the notion of quantum $\mathsf{VBB}$ obfuscation (this follows easily in the case when $f$ comes from an ensemble that is not $\mathsf{VBB}$ obfuscatable). Our construction suggests that it is sometimes safe to publish parts of the program that is being copy-protected in the clear, as long as some crucial components are still ``obfuscated''.

In conclusion, we find that obfuscation and copy-protection are, in general, incomparable functionalities. We point out, however, that obfuscators have so far been employed in all proposed constructions of quantum copy-protection schemes, such as in \cite{aaronson2009quantum,aaronson2020quantum}, as well as in our scheme of Construction $\ref{cons:cp}$ (note that publishing $H(v)$ for a random oracle $H$ is an ideal obfuscator for the point function with marked input $v$).

\begin{figure}
\begin{center}
    \begin{tikzpicture}[scale=0.9]
      \draw \firstcircle 
      node at (-2.5,-3.75) {Quantum copy-protection};
      \draw node at (-2.1,-2.05) {$\bullet$};
      \draw node at (-2.1,-2.35) {Constr. \ref{cons: pf to cc}};
      \draw node at (0,-0.25) {$\bullet$};
      \draw node at (0,-0.55) {Constr. \ref{cons:cp}};
      \draw \secondcircle 
     node at (2.45,-3.75) {Quantum obfuscation};
     \draw[dotted] (2.3,-2.2) ellipse (1.2cm and 0.7cm);
     \draw node at (2.4,-2.08) {\small{learnable}};
     \draw node at (2.35,-2.48) {\small{programs}};
    \end{tikzpicture}
\end{center}
 \caption{\textbf{Separation between quantum copy-protection and quantum $\mathsf{VBB}$ obfuscation}. Construction \ref{cons:cp} features a quantum copy-protection scheme for point functions which satisfies both notions, while our second scheme for compute-and-compare programs in Construction \ref{cons: pf to cc} does not satisfy the notion of quantum $\mathsf{VBB}$ obfuscation.
 }
 \label{fig:separation}
\end{figure}
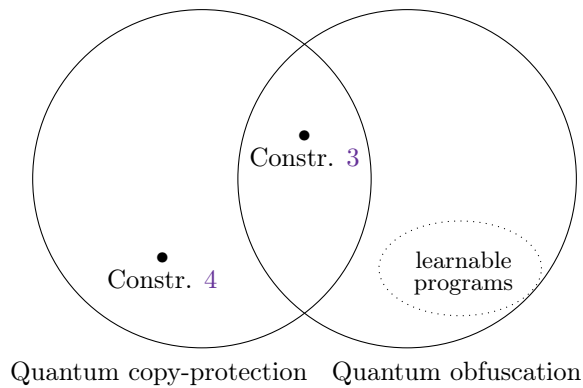

\paragraph{Unclonable encryption.} This cryptographic functionality was formalized recently by Broadbent and Lord \cite{broadbent2019uncloneable}, and informally introduced earlier by Gottesman \cite{gottesman2003uncloneable}. In an unclonable encryption scheme, one encrypts a \emph{classical} message in a \emph{quantum} ciphertext, in such a way that the latter cannot be processed and \emph{split} into two parts, each of which, together with a classical secret key, enables decryption. The setting of unclonable encryption is very similar to that of copy-protection, the main difference being that there is no ``functionality'' associated to the quantum ciphertext, other than it being used for recovering the encrypted message. As we mentioned earlier, our copy-protection scheme is inspired by the unclonable encryption scheme in \cite{broadbent2019uncloneable}, and our analysis extends some of the techniques developed there.

\paragraph{Revocable quantum timed-release encryption.} Timed-release encryption (also known as time-lock puzzles) is an encryption scheme that allows a recipient to decrypt only after a specified amount time, say $T$, has passed. Unruh \cite{Unruh15} gave the first quantum timed-release encryption scheme that is ``revocable'' in the sense that a user can return the timed-release encryption before time $T$, thereby losing all access to the data. It is easy to see that this notion is impossible to achieve classically for precisely the same reason copy-protection is impossible: any adversary can simply generate copies of the classical ciphertext or source code, respectively. From a technical point of view, revocable quantum timed-release encryption shares many similarities with the notion of ``secure software leasing'' in \cite{ananth2020secure}. Besides the fact that the former encodes a \textit{plaintext} and the latter encodes a \textit{program}, the security property essentially remains the same: once a quantum state is returned and successfully verified, the user is supposed to lose all relevant information.
Our proof of security for the $\SSL$ scheme in Construction \ref{const: ssl} is inspired by Unruh's proof for revocable one-way timed-release quantum encryption in \cite{Unruh15}.

\subsection{Open questions}
Our work is the first to construct a copy-protection scheme in a standard cryptographic model (the QROM) with non-trivial security against malicious adversaries. It leaves several questions open. The most pressing ones are the following. 
\begin{itemize}
    \item First of all, our security guarantee is very weak: it only ensures that no adversary can win with probability more than $1-\delta$, for a very small constant $\delta$ (approximately $10^{-4}$). Is it possible to boost the security to ensure negligible advantage? (The main technical hurdle seems to be the factor of $9$ in the bound of Lemma \ref{lem: main bound}.)
    \item Another open problem concerns upgrading security in a different way, by providing the pirate with $k$ independent copy-protected copies of a program and asking them to satisfy $k+1$ freeloaders. We believe that our scheme can achieve such a notion, but with a security that becomes worse as $k$ grows. The setting of $k \mapsto k+1$ unclonable security has been explored in subsequent works for constructions based on subspace coset states~\cite{10.1007/978-3-031-22318-1_11,cakan2023unclonable}. 
    \item Finally, is it possible to remove the requirement of a random oracle, and to achieve a scheme with non-trivial security against malicious adversary in the plain model? We think that this would require fundamentally different techniques.
\end{itemize}

\subsection{Subsequent work}

We remark that a series of subsequent works have meanwhile improved on some of our results on copy-protection and secure software leasing. 

Majenz et al.~\cite{https://doi.org/10.48550/arxiv.2103.14510} gave new \emph{negative results} for unclonable encryption schemes whose security proofs rely on the simultaneous one-way-to-hiding (O2H), such as the scheme of Broadbent and Lord \cite{broadbent2019uncloneable} and, by implication, our QCP scheme for point functions in Construction \ref{cons:cp}. The authors give an explicit
cloning-indistinguishable attack whose advantage is characterized by the largest eigenvalue of the ciphertext. This suggests that the security of our QCP scheme for point functions cannot be significantly improved.

Broadbent et al.~\cite{Broadbent_2021} showed how to construct an \emph{information-theoretic} SSL scheme for compute-and-compare programs from any quantum message authentication code. This improves on our $\SSL$ scheme in Construction \ref{const: ssl} which relies on the quantum random oracle heuristic.

Using the powerful notion of \emph{subspace coset states} rather than BB84 states,
Coladangelo et al.~\cite{10.1007/978-3-030-84242-0_20} construct copy-protection for pseudorandom functions from post-quantum indistinguishability obfuscation (iO). Subsequently,
Ananth et al.~\cite{https://doi.org/10.48550/arxiv.2207.06589} used subspace coset states to obtain a copy-protection for single-bit point functions which achieves the standard notion of \emph{negligible adversarial security} in the QROM.

Also using subspace coset states and iO, Liu et al.~\cite{10.1007/978-3-031-22318-1_11}
construct (bounded) collusion-resistant 
copy-protection for ``watermarkable functionalities'', such as pseudorandom functions and digital signatures. Moreover, they construct copy-protection secure against $k \mapsto k+1$ attacks. Later, a follow-up work by Cakan and Goyal~\cite{cakan2023unclonable} achieved unbounded-collusion resistance.

\subsection{Acknowledgements}
The authors thank Anne Broadbent, Yfke Dulek, Jiahui Liu, Fermi Ma, Christian Schaffner, Umesh Vazirani and Thomas Vidick for helpful discussions. In particular, the authors thank Thomas Vidick for valuable feedback on an earlier draft of this manuscript, and Umesh Vazirani for suggesting the extension of our scheme to compute-and-compare programs. The authors are grateful for the support of the Simons Institute for the Theory of Computing, where part of this research was carried out. A.C was supported by the Simons Institute for the Theory of Computing, and Vannevar Bush faculty fellowship N00014-17-1-3025. C.M. was funded by a NWO VENI grant (Project No. VI.Veni.192.159). A.P. is partially supported by AFOSR YIP award number FA9550-
16-1-0495, the Institute for Quantum Information
and Matter (an NSF Physics Frontiers Center; NSF
Grant PHY-1733907) and the Kortschak Scholars program.

\section{Preliminaries}

\subsection{Notation and background}

\textit{Basic notation.} For $n,m \in \mathbb{N}$, we denote the set of all functions $f: \{0,1\}^{n} \to\{0,1\}^m$,  as $\mathsf{Bool}(n,m)$. The notation $x \leftarrow \{0,1\}^n$ denotes sampling an element $x$ uniformly at random in $\{0,1\}^n$, whereas
$x \leftarrow D$ denotes sampling of an element $x$ from a distribution $D$.
We denote the expectation value of a random variable $X$ by $\mathbb{E}[X] = \sum_{x} x \Pr[ X = x] $. 
We call a non-negative real-valued function $\mu : \mathbb{N} \rightarrow \mathbb{R}^+$ negligible if $\mu(n) = o(1/p(n))$, for every polynomial $p(n)$.
The min-entropy of a random variable $X$ is defined as $\Hmin(X) = - \log(\max_{x} \Pr[X = x])$. The conditional min-entropy of a random variable $X$ conditioned on a correlated random variable $Y$ is defined as
$\Hmin(X|Y) = - \log\big(\mathbb{E}_{y \leftarrow Y} \big[\max_x \Pr[X = x | Y = y] \big] \big).$\ \\
\ \\
\textit{Quantum computation.}
A comprehensive introduction to quantum computation and quantum information can be found in \cite{NC} and \cite{WildeQIT}. We denote a finite-dimensional complex Hilbert space by $\mathcal{H}$, and we use subscripts to distinguish registers. For example $\mathcal{H}_{\mathsf{R}}$ is the Hilbert space of register $\mathsf{R}$. The Euclidean norm over a finite-dimensional complex Hilbert space $\mathcal{H}$ is denoted as $\|\cdot \|$. A quantum register over the Hilbert space $\mathcal{H} = \mathbb{C}^2$ is called a \emph{qubit}. For $n \in \mathbb{N}$, we refer to quantum registers over the Hilbert space $\mathcal{H} = \big(\mathbb{C}^2\big)^{\otimes n}$ as $n$-qubit states. We use the word ``quantum state'' to refer to both pure states (unit vectors $\ket{\psi} \in \mathcal{H}$) and density matrices (positive semidefinite matrices $\rho$ of unit trace in the space of density matrices $\mathcal{D}(\mathcal{H)}$).
When $n$ is clear from the context, we denote by $\ket{\phi^+}_{\mathsf{AB}}$ the maximally entangled $n$-qubit Einstein-Podolsky-Rosen $(\EPR)$ \cite{PhysRev.47.777} state on the Hilbert space $\mathcal{H}_{\mathsf{A}} \otimes \mathcal{H}_{\mathsf{B}}$:
\begin{align*}
\ket{\phi^+}_{\mathsf{AB}} = \frac{1}{\sqrt{2^n}}\sum_{v \in \{0,1\}^n} \ket{v}_\mathsf{A} \otimes \ket{v}_\mathsf{B}.
\end{align*}
A $\POVM$ is a finite set $\{M_i\}$ of positive semidefinite matrices with the property that $\sum_{i} M_i = \Id$, where $\Id$ is the identity matrix. The trace distance between two states $\rho$ and $\sigma$ is defined by $\mathsf{TD}(\rho,\sigma) = \frac{1}{2}\| \rho- \sigma\|_1$.
A quantum channel is a linear map $\Phi: \mathcal{H}_{\mathsf{A}} \rightarrow \mathcal{H}_{\mathsf{B}}$ between two Hilbert spaces $\mathcal{H}_{\mathsf{A}}$ and $\mathcal{H}_{\mathsf{B}}$. We say that a channel $\Phi$ is completely positive if, for any ancilla register of dimension $n$, the induced map $ \Id_n \otimes \Phi$ is positive, and we call it trace-preserving if it holds that $\mathrm{Tr}[\Phi(\rho)] = \mathrm{Tr}[\rho]$, for all $\rho \in \mathcal{D}(\mathcal{H}_\mathsf{A})$.
A quantum channel that is both completely positive and trace-preserving is called a quantum $\CPTP$ channel.
Let $\rho \in \mathcal{D}(\mathcal{H}_{\mathsf{A}} \otimes \mathcal{H}_{\mathsf{B}})$ be a bipartite state. Then, the min-entropy of $\mathsf{A}$ conditioned on system $\mathsf{B}$ is given by:
\begin{align*}
  \Hmin(\mathsf{A}|\mathsf{B})_\rho = - \inf_{\sigma \in \mathcal{D}(\mathcal{H}_\mathsf{B})} D_{\max} \big(\rho_\mathsf{AB} \, || \, \Id_\mathsf{A} \otimes \sigma_\mathsf{B} \big),
\end{align*}
where $D_{\max} \big(\rho \, || \, \sigma \big) = \inf_\lambda \{\lambda: \rho \leq 2^\lambda \sigma\}$ is the max-relative entropy.
A \textit{classical-quantum} state $\rho \in \mathcal{D}(\mathcal{H}_{\mathsf{X}} \otimes \mathcal{H}_{\mathsf{B}})$ is a quantum state on registers $\mathsf{X}$ and $\mathsf{B}$ that depends on a classical variable $X$ in system $\mathsf{X}$. If the variable $X$ is distributed according to $P_X$, the state $\rho$ can be expressed as
\begin{align*}
\rho_\mathsf{XB} = \sum_x P_X(x) \ket{x}\bra{x}_\mathsf{X} \otimes \rho_\mathsf{B}^x.
\end{align*}
Let $p_g(\mathsf{X}|\mathsf{B}) = \sum_x P_X(x) \mathrm{Tr}[\Lambda_x \rho_\mathsf{B}^x]$ be the optimal guessing probability for the variable $X$ when using an optimal measurement strategy $\{ \Lambda_x\}$ that maximizes the expression.
The quantity can then be expressed in terms of the min-entropy via $p_g(\mathsf{X}|\mathsf{B}) = 2^{-\Hmin(\mathsf{X}|\mathsf{B})}$ \cite{Koenig-Renner-Schaffner09}.
By a polynomial-time \textit{quantum algorithm} (or $\QPT$ algorithm) we mean a polynomial-time uniform family of quantum circuits, where each circuit in the circuit family is described by a sequence of unitary gates and measurements. A quantum algorithm may, in general, receive (mixed) quantum states as inputs and produce (mixed) quantum states as outputs. We oftentimes restrict $\QPT$ algorithms implicitly; for example, if we write $\Pr[\mathcal{A}(1^{\lambda}) = 1]$ for a $\QPT$ algorithm $\mathcal{A}$, it is implicit that $\mathcal{A}$ is a $\QPT$ algorithm that output a single classical bit.

\begin{definition}[Indistinguishability of ensembles of quantum states, \cite{Watrous06zero-knowledgeagainst}]
\label{def: indistinguishability}
Let $p: \mathbb{N} \rightarrow \mathbb{N}$ be a polynomially bounded function,
and let $\rho_\lambda$ and $\sigma_\lambda$
be $p(\lambda)$-qubit quantum states. We say that $\{\rho_{\lambda}\}_{\lambda \in \mathbb{N}}$ and $\{\sigma_\lambda\}_{\lambda \in \mathbb{N}}$ are quantum computationally indistinguishable ensembles of quantum states, denoted by $\rho_{\lambda} \approx_c \sigma_{\lambda}\,,$
if, for any $\QPT$ distinguisher $\mathcal{D}$ with single-bit output, any polynomially bounded $q: \mathbb{N} \rightarrow \mathbb{N}$, any family of $q(\lambda)$-qubit auxiliary states $\{\nu_{\lambda}\}_{\lambda \in \mathbb{N}}$, and every $\lambda \in \mathbb{N}$,
$$ \big| \Pr[\mathcal{D}(\rho_{\lambda} \otimes \nu_{\lambda})=1] - \Pr[\mathcal{D}(\sigma_{\lambda} \otimes \nu_{\lambda})=1] \big| \leq \negl(\lambda) \,.$$
\end{definition}

We also make use of the following standard results:

\begin{lem}[Ricochet property]\label{lem:ricochet}
Let $n \in \mathbb{N}$ and $M \in \mathbb{C}^{2^n \times 2^n}$ be any matrix, and let $\ket{\phi^+}_{\mathsf{AB}}$ an $\EPR$ state on registers $\mathsf{A}$ and $\mathsf{B}$ on $n$ qubits. Then,
$$
(M_\mathsf{A} \otimes \Id_\mathsf{B}) \ket{\phi^+}_{\mathsf{AB}} = (\Id_\mathsf{A} \otimes M^T_\mathsf{B}) \ket{\phi^+}_{\mathsf{AB}}.
$$
\end{lem}

\begin{lem}
[Gentle Measurement Lemma, \cite{Winter99,Aaronson_2005}]\label{lem:gentle} Let $\rho$ be a density matrix, let $\Lambda$, where $0 \leq \Lambda \leq \Id$, be an element of a $\POVM$ such that $\mathrm{Tr}[\Lambda \rho] \geq 1 - \epsilon$, for some $\epsilon > 0$ and let $\tilde{\rho}= \Lambda \rho \Lambda$. Then:
$$
\mathsf{TD}(\rho,\tilde \rho)  \, \leq \, \sqrt{\epsilon}.
$$
\end{lem}
And a slight variation of Lemma \ref{lem:gentle}.
\begin{lem}[Closeness to ideal states, \cite{Unruh15}]\label{lem:closeness_ideal}
Let $\rho$ be a density matrix, let $\Pi$ be a projector, and let $\epsilon = \mathrm{Tr}[\left(\mathds 1-\Pi\right) \rho]$. Then, there exists an ideal state $\rho^{\textsf{id}}$ such that:
\begin{itemize}
    \item $\mathsf{TD}(\rho,\rho^{\textsf{id}}) \leq \sqrt{\epsilon}$
    \item $\rho^{\textsf{id}}$ is a mixture in the image of $\Pi$, i.e. $\rho^{\textsf{id}} = \sum_i p_i \ket{\psi_i} \bra{\psi_i}$ is a normalized state with $\ket{\psi_i} \in \mathsf{im}(\Pi)$, $\sum_i p_i =1$ and $p_i \geq 0$, for all $i$.
\end{itemize}
\end{lem}

\subsection{Monogamy of entanglement games}
\label{sec: monogamy}
For a detailed introduction to monogamy-of-entanglement games, we refer the reader to the seminal paper on the topic \cite{tomamichel2013monogamy}, where they were introduced and studied extensively. In this section, we limit ourselves to introducing a version of a monogamy-of-entanglement game that suffices for our purpose. Let $\lambda \in \mathbb{N}$. The game is between a challenger and an adversary, specified by a triple of interactive quantum machines $(\mathcal{A}, \mathcal{B}, \mathcal{C})$ (for a formal definition of interactive quantum machine we refer the reader to \cite{unruh2012quantum}). For brevity, we use the notation $\ket{x^{\theta}} = H^{\theta} \ket{x}$, where $H^{\theta} = H^{\theta_1} \otimes \ldots \otimes H^{\theta_{\lambda}}$ and $\theta,x \in \{0,1\}^{\lambda}$.
The game takes place as follows:
\begin{itemize}
    \item The challenger samples $x,\theta \leftarrow \{0,1\}^{\lambda}$ and sends the state $\ket{x^{\theta}}$ to $\mathcal{A}$.
    \item $\mathcal{A}$ sends a quantum register to $\mathcal{B}$ and one to $\mathcal{C}$.
    \item The challenger sends $\theta$ to both $\mathcal{B}$ and $\mathcal{C}$.
    \item $\mathcal{B}$ and $\mathcal{C}$ return strings $x'$ and $x''$ to the challenger.
\end{itemize}
$\mathcal{A}$, $\mathcal{B}$ and $\mathcal{C}$ are not allowed to communicate other than where specified by the game. $\mathcal{A},\mathcal{B},\mathcal{C}$ win if $x = x' = x''$.

The following lemma, from \cite{tomamichel2013monogamy}, upper bounds the winning probability of an adversary in the game. As stated in the form below, this lemma appears in \cite{broadbent2019uncloneable}.

\begin{lem}
\label{lem: monogamy}
Let $\lambda \in \mathbb{N}$. For any Hilbert spaces $\mathcal{H}_\mathsf{B}$ and $\mathcal{H}_\mathsf{C}$, any families of $\mathsf{POVMs}$ on these Hilbert spaces respectively,
$$ \Big\{\Big\{B_{\theta}^x \Big\}_{x \in \{0,1\}^{\lambda}} \Big\}_{\theta \in \{0,1\}^{\lambda}} \, \text{ and }\, \,\, \Big\{\Big\{C_{\theta}^x \Big\}_{x \in \{0,1\}^{\lambda}} \Big\}_{\theta \in \{0,1\}^{\lambda}}\,,$$
and any $\CPTP$ map $\Phi: \mathcal{D}\left((\mathbb{C}^2)^{\otimes \lambda}\right) \rightarrow \mathcal{D}(\mathcal{H}_\mathsf{B} \otimes \mathcal{H}_\mathsf{C})$, we have:
\begin{equation}
\mathbb{E}_{\theta \in \{0,1\}^{\lambda}}\mathbb{E}_{x \in \{0,1\}^{\lambda}} \mathrm{Tr} \Big[B_{\theta}^x \otimes C_{\theta}^x \,\Phi\left(\ket{x^{\theta}}\bra{x^{\theta}}\right) \Big]\leq \left(\frac12 + \frac{1}{2\sqrt{2}}\right)^{\lambda}
\end{equation}

\end{lem}

It is natural to wonder whether a similar guarantee holds also when $x$ is sampled from a distribution with high min-entropy, which is not necessarily uniform. This follows as a corollary of Lemma \ref{lem: monogamy}. 

\begin{cor}
\label{cor: monogamy entropy}
Let $\lambda \in \mathbb{N}$.  Let $X$ be a random variable over $\{0,1\}^{\lambda}$ with min-entropy $\Hmin(X)\geq h$. For any Hilbert spaces $\mathcal{H}_\mathsf{B}$ and $\mathcal{H}_\mathsf{C}$, any family of $\mathsf{POVMs}$ on these Hilbert spaces respectively,
$$ \Big\{\Big\{B_{\theta}^x \Big\}_{x \in \{0,1\}^{\lambda}} \Big\}_{\theta \in \{0,1\}^{\lambda}} \, \text{ and }\, \,\, \Big\{\Big\{C_{\theta}^x \Big\}_{x \in \{0,1\}^{\lambda}} \Big\}_{\theta \in \{0,1\}^{\lambda}}\,,$$
and any $\CPTP$ map $\Phi: \mathcal{D}\left((\mathbb{C}^2)^{\otimes \lambda}\right) \rightarrow \mathcal{D}(\mathcal{H}_\mathsf{B} \otimes \mathcal{H}_\mathsf{C})$, we have:
\begin{equation}
\label{eq: entropy}
\sum_{x \in \{0,1\}^{\lambda}} \Pr[X=x] \,\mathbb{E}_{\theta}\Tr{B_{\theta}^x \otimes C_{\theta}^x \,\Phi\left(\ket{x^{\theta}}\bra{x^{\theta}}\right)} \leq 2^{-h} \left(1 + \frac{1}{\sqrt{2}}\right)^{\lambda} \,.
\end{equation}
\end{cor}

\begin{proof} This fact is easily verified as follows:
\begin{align}
&\sum_{x \in \{0,1\}^{\lambda}} \Pr[X=x] \,\mathbb{E}_{\theta}\Tr{B_{\theta}^x \otimes C_{\theta}^x \,\Phi\left(\ket{x^{\theta}}\bra{x^{\theta}}\right)} \nonumber\\
&\leq 2^{-h} 2^{\lambda} \mathbb{E}_x \mathbb{E}_{\theta}\Tr{B_{\theta}^x \otimes C_{\theta}^x \,\Phi\left(\ket{x^{\theta}}\bra{x^{\theta}}\right)} \nonumber\\ 
&\leq 2^{-h+\lambda} \, \left(\frac12 + \frac{1}{2\sqrt{2}}\right)^{\lambda}
= 2^{-h} \left(1 + \frac{1}{\sqrt{2}}\right)^{\lambda}.
\end{align}
Note that the second inequality follows from Lemma \ref{lem: monogamy}.
\end{proof}

In particular, notice that when $h > \frac{4}{5}\lambda$, the RHS of \eqref{eq: entropy} is less than $0.981^{\lambda}$, and thus negligible.

\subsection{The quantum random oracle model}

Oracles with quantum access have been studied extensively, for example in \cite{BBBV97, BDF11}. We say that a quantum algorithm $\mathcal{A}$ has oracle access to a classical function $H: \{0,1 \}^{\lambda} \rightarrow \{0,1 \}^m$, denoted by $\mathcal{A}^H$, if $\mathcal{A}$ is allowed to use a unitary gate $O^H$ at unit cost in time. The unitary $O^H$ acts as follows on the computational basis states of a Hilbert space $\mathcal{H}_\mathsf{A} \otimes \mathcal{H}_\mathsf{B}$ of $\lambda+m$ qubits:
$$
O^H: \quad
\ket{x}_\mathsf{A} \otimes \ket{y}_\mathsf{B} \longrightarrow \ket{x}_\mathsf{A} \otimes \ket{y \oplus H(x)}_\mathsf{B},
$$
where the operation $\oplus$ denotes bit-wise addition modulo $2$. In general, we can model the interaction of a quantum algorithm that makes $q$ queries to an oracle $H$ as $(UO^H)^q$, i.e. alternating unitary computations and queries to the oracle $H$, where $U$ is some operator acting on $\mathcal{H}_\mathsf{A} \otimes \mathcal{H}_\mathsf{B} \otimes \mathcal{H}_\mathsf{C}$,
where $\mathcal{H}_\mathsf{C}$ is some auxiliary Hilbert space \cite{BBBV97, BDF11, Unruh15}\footnote{We can chose the algorithm's unitaries between oracle calls to be all the same by introducing a ``clock register'' that keeps track of the number of oracle calls made so far.}.
We call a (possibly super-polynomial-time) quantum algorithm $\mathcal{A}$ with access to an oracle $O$ \textit{query-bounded} if $\mathcal{A}$ makes at most polynomially many (in the size of its input) queries to $O$.
The \textit{random oracle} model refers to a setting in which the function $H: \{0,1 \}^{\lambda} \rightarrow \{0,1 \}^m$ is sampled uniformly at random. Random oracles play an important role in cryptography as models for cryptographic hash functions in the so-called random oracle model (ROM) \cite{BR93}. For post-quantum and quantum cryptography, random oracles modelling hash functions need to be \emph{quantum} accessible (i.e. accessible as a unitary gate, and thus in superposition), resulting in what is known as the quantum random oracle model (QROM) \cite{BDF11}. Despite being uninstantiable in principle \cite{CGH04,Song20}, modeling hash functions in the (Q)ROM is considered a standard assumption in cryptography.

\subsubsection{Some technical lemmas}
Recall that we denote by $\textnormal{Bool}(\lambda,m)$ the set of functions from $\{0,1\}^{\lambda}$ to $\{0,1\}^m$.
\begin{lem}
\label{lem: basic}
Let $f: \textnormal{Bool}(\lambda,m) \rightarrow \mathbb{R}$, and $x \in \{0,1\}^{\lambda}$. For $H \in \textnormal{Bool}(\lambda,m)$ and $y \in \{0,1\}^m$, let  $H_{x,y} \in \textnormal{Bool}(\lambda,m)$ be such that

$$ H_{x,y} (s) = \begin{cases} H(s)    & \text{if } s\neq x\,,\\
    y &\text{if } s = x \,.   \end{cases}$$
    
Then, 
$$ \mathbb{E}_H f(H) = \mathbb{E}_H \mathbb{E}_y f(H_{x,y}) \,.$$ 
\end{lem}
\begin{proof}
The proof is straightforward, and can be found in Lemma 19 of \cite{broadbent2019uncloneable}.
\end{proof}

The following is a technical lemma about a quantum adversary not being able to distinguish between samples from $H(U_{\lambda})$ and from $U_m$, even when given oracle access to $H$, where the function $H: \{0,1\}^{\lambda} \rightarrow \{0,1\}^m$ is sampled uniformly at random. 

Consider the following game between a challenger and a quantum adversary $\mathcal{A}$, specified by $\lambda, m \in \mathbb{N}$, and a distribution $X$ over $\{0,1\}^{\lambda}$,  
\begin{itemize}
    \item The challenger samples a uniformly random function $H: \{0,1\}^{\lambda} \rightarrow \{0,1\}^m$ and $b \leftarrow \{0,1\}$.
    \item If $b = 0$: the challenger samples $x \leftarrow X$, sends $H(x)$ to $\mathcal{A}$.\\ If $b = 1$: the challenger samples uniformly $z \leftarrow \{0,1\}^m$, sends $z$ to $\mathcal{A}$.
    \item $\mathcal{A}$ additionally gets oracle access to $H$. $\mathcal{A}$ returns a bit $b'$ to the challenger.
\end{itemize}
$\mathcal{A}$ wins if $b=b'$. Let $\textsf{Distinguish}(\mathcal{A}, \lambda , m, X)$ be a random variable for the outcome of the game. 

\begin{lem}
\label{lem: qrom distinguishing}
For any adversary $\mathcal{A}$ making $q$ oracle queries, any family of distributions ${\{X_{\lambda}: \lambda \in \mathbb{N}\}}$ where for all $\lambda$, $X_{\lambda}$ is a distribution over $\{0,1\}^{\lambda}$, for any polynomially bounded function $m: \mathbb{N} \rightarrow \mathbb{N}$, there exists a negligible function $\mu$ such that, for any $\lambda \in \mathbb{N}$, the following holds:
$$ \Pr[\textsf{Distinguish}(\mathcal{A}, \lambda, m(\lambda), X_{\lambda}) = 1] \leq \frac12 + (3q + 2) q M + \mu(\lambda) \,,$$
where $M$ is a quantity that is negligible in $\lambda$ if $2^{-\Hmin(X_{\lambda})/2}$ is negligible in $\lambda$.
\end{lem}

\begin{cor}
\label{cor: dist entropy}
For any query-bounded adversary $\mathcal{A}$, any $\epsilon > 0$, any family of distributions $\{X_{\lambda}: \lambda \in \mathbb{N}\}$, where $X_{\lambda}$ is a distribution over $\{0,1\}^{\lambda}$ with $\Hmin(X_{\lambda}) > \lambda^{\epsilon}$ for every $\lambda$, for any polynomially bounded function $m: \mathbb{N} \rightarrow \mathbb{N}$, there exists a negligible function $\mu$ such that, for any $\lambda \in \mathbb{N}$, the following holds:
$$ \Pr[\textsf{Distinguish}(\mathcal{A}, \lambda, m(\lambda), X_{\lambda}) = 1] \leq \frac12 + \mu(\lambda)  \,.$$
\end{cor}

The key step in the proof of Lemma \ref{lem: qrom distinguishing} is captured by the one-way-to-hiding lemma \cite{Unruh15, ambainis2019quantum}\footnote{While additional improved variants of the one-way to hiding lemma were developed \cite{bindel2019,kuchta2020}, any of them suffices for our asymptotic analysis.}. We restate it here as Lemma \ref{lem: qrom technical step} following our notation (and provide a proof in the Appendix \ref{apx: lem 7 proof} for completeness). Informally, the lemma gives an upper bound on an adversary's advantage (when given access to a uniformly random function $H: \{0,1\}^n \rightarrow \{0,1\}^m$) at distinguishing between a sample drawn from $H(U_n)$ and a sample drawn from $U_m$. The upper bound is in terms of the probability that the adversary queries the oracle at the pre-image of the sample at some point during its execution. Equivalently, given two oracles that are identical except on a single input (or more generally on a subset of the inputs), the advantage of an adversary at distinguishing the two oracles is bounded above in terms of the probability that the adversary queries at the differing point (or at a point in the subset where they differ) at some point during its execution. We prove the following lemma in Appendix \ref{apx: lem 7 proof}.

\begin{lem}
\label{lem: qrom technical step}
Let $\lambda, m \in \mathbb{N}$. For any $q \in \mathbb{N}$, any unitaries $U$, any family of states $\{\ket{\psi_x}\}_{x \in \mathcal{X}}$, any complete pair of orthogonal projectors $(\Pi^0, \Pi^1)$ and any distribution $X$ on $\{0,1\}^{\lambda}$, it holds that:
\begin{align}
\label{eq: 43}
\frac12 \mathbb{E}_{H}\mathbb{E}_{x \leftarrow X}\| \Pi^0  (U O^H)^q\left( \ket{H(x)} \otimes \ket{\psi_x} \right)\|^2 + \frac12 \mathbb{E}_{H}\mathbb{E}_{z \leftarrow \{0,1\}^m}\| \Pi^1  (U O^H)^q\left( \ket{z} \otimes \ket{\psi_x} \right)\|^2 \nonumber \\
    \leq \frac12 + (3q+2)q M \,, \quad\quad\quad\quad
\end{align}
where $O^H$ is the oracle unitary for $H: \{0,1\}^{\lambda} \rightarrow \{0,1\}^m$, and $M$ is given by
\begin{align}
    M = &\frac12 \,\mathbb{E}_{H}\mathbb{E}_{x \leftarrow X} \mathbb{E}_{z \leftarrow \{0,1\}^m} \mathbb{E}_k \|  \ket{x}\bra{x} (U O^{{H_{x,z}}})^k \ket{z} \otimes \ket{\psi_x}\| \nonumber\\
+& \frac12 \,\mathbb{E}_{H}\mathbb{E}_{x \leftarrow X} \mathbb{E}_{z \leftarrow \{0,1\}^m} \mathbb{E}_k \|  \ket{x}\bra{x} (U O^H)^k \ket{z} \otimes \ket{\psi_x}\|  \,. \label{eq: 44}
\end{align}
Moreover, $M$ is negligible if and only if the second term in $M$ is negligible.
\end{lem}
The lemma holds also when the states $\ket{\psi_x}$ are not necessarily pure (but we write them as pure states for ease of notation).

\begin{remark}
Previously, Chung, Guo, Liu and Qian~\cite{9317914} gave an improved bound of the form $\frac{1}{2}+q M$ in the special case that $X$ is uniform.   
\end{remark}

\begin{proof}[Proof of Lemma \ref{lem: qrom distinguishing}]
Without loss of generality, let $\mathcal{A}$ be specified by a unitary $U$, the oracle unitary $O^H$ and a measurement given by projectors $\Pi_0$ and $\Pi_1=\mathds 1-\Pi_0$, so that the unitary part of $\mathcal{A}$'s algorithm is $(U O^H)^q$, where $q$ the number of oracle queries made by $\mathcal{A}$. Then, $\mathcal{A}$'s winning probability is precisely  given by,
\begin{align}
    &\Pr[\textsf{Distinguish}(\mathcal{A}, \lambda, m, X_{\lambda}) = 1] \nonumber \\
    &=\frac12 \mathbb{E}_{H}\mathbb{E}_{x \leftarrow X_{\lambda}}\| \Pi^0  (U O^H)^q  \ket{H(x)} \|^2 + \frac12 \mathbb{E}_{H}\mathbb{E}_{z \leftarrow \{0,1\}^m}\| \Pi^1  (U O^H)^q \ket{z} \|^2 \,,
\end{align} 
where we omit writing ancilla qubits initialized in the zero state that $(UO^H)^q$ might be acting on.

Then, by Lemma \ref{lem: qrom technical step}, we have \begin{equation}
    \Pr[\textsf{Distinguish}(\mathcal{A}, \lambda, m, X_{\lambda}) = 1] \leq \frac12 + (3q+2)q \, M
\end{equation}
where $M$ is the quantity given by
\begin{align}
    M &= \frac12  \,\mathbb{E}_{H}\mathbb{E}_{x \leftarrow X_{\lambda}} \mathbb{E}_{z \leftarrow \{0,1\}^m} \mathbb{E}_k \|  \ket{x}\bra{x} (U O^{{H_{x,z}}})^k \ket{z} \| \nonumber \\
&+ \frac12 \,\mathbb{E}_{H}\mathbb{E}_{x \leftarrow X_{\lambda}} \mathbb{E}_{z \leftarrow \{0,1\}^m} \mathbb{E}_k \|  \ket{x}\bra{x} (U O^H)^k \ket{z} \| \,.  \label{eq: 12}
\end{align}

Moreover, by Lemma \ref{lem: qrom technical step}, $M$ is negligible if and only if the second term,
$$\mathbb{E}_{H}\mathbb{E}_{x \leftarrow X_{\lambda}} \mathbb{E}_{z \leftarrow \{0,1\}^m} \mathbb{E}_k \|  \ket{x}\bra{x} (U O^H)^k \ket{z} \|,$$ 
is negligible. Hence, it suffices to bound the above term.
Notice that for any fixed $H$, $z$ and $k$, 
\begin{align}
    \mathbb{E}_{x \leftarrow X_{\lambda}} & \|  \ket{x}\bra{x} (U O^H)^k \ket{z} \| \nonumber \\
    & \leq \sqrt{\mathbb{E}_{x \leftarrow X_{\lambda}} \| \ket{x}\bra{x} (U O^H)^k \ket{z} \|^2} \nonumber\\
    & \leq 2^{-\Hmin(X_{\lambda})/2} \,.
\end{align} 
where the first inequality follows from Jensen's inequality (for concave functions), and the second inequality uses the fact that the state $(U O^H)^k \ket{z}$ does not depend on $x$, and hence the quantity under the square root is bounded above by the optimal probability of correctly predicting a sample from $X$, which is, by definition, $2^{-\Hmin(X)}$. 
Therefore, $M$ is negligible so long as $2^{-\Hmin(X_{\lambda})}$ is negligible.
\end{proof}

Finally, we define the notion of indistinguishability of ensembles of quantum states \emph{in the QROM}. This is similar to Definition \ref{def: indistinguishability}.

\begin{definition}[Indistinguishability of ensembles of quantum states in the QROM]
\label{def: indistinguishability qrom}
Let $m: \mathbb{N} \rightarrow \mathbb{N}$ and $p: \mathbb{N} \rightarrow \mathbb{N}$ be polynomially bounded functions, and let $\rho^H_{\lambda}$ and $\sigma^H_{\lambda}$ be $p(\lambda)$-qubit states, for $H \in \mathsf{Bool}(\lambda, m(\lambda))$. We say that $\{\rho_{\lambda}^H\}_{\lambda \in \mathbb{N}, H \in \mathsf{Bool}(\lambda, m(\lambda))}$ and $\{\sigma^H_{\lambda}\}_{\lambda \in \mathbb{N}, H \in \mathsf{Bool}(\lambda, m(\lambda))}$ are quantum computationally indistinguishable ensembles of quantum states, denoted by $\rho_{\lambda}^H \approx_c \sigma_{\lambda}^H\,,$
if, for any $\QPT$ distinguisher $\mathcal{D}^H$ with single-bit output, any polynomially bounded $q: \mathbb{N} \rightarrow \mathbb{N}$, any family of $q(\lambda)$-qubit auxiliary states $\{\nu_{\lambda}\}_{\lambda \in \mathbb{N}}$, and every $\lambda \in \mathbb{N}$,
$$ \mathbb{E}_H \big| \Pr[\mathcal{D}^H(\rho_{\lambda}^H \otimes \nu_{\lambda})=1] - \Pr[\mathcal{D}^H(\sigma_{\lambda}^H \otimes \nu_{\lambda})=1] \big|\leq \negl(\lambda) \,.$$
\end{definition}

\subsection{Unclonable quantum encryption}

A quantum encryption scheme is a procedure that encodes a plaintext in the form of a quantum state in terms of a quantum \emph{ciphertext}.

\begin{definition}[$\SKQES$]\label{def:QECM}\ \\
Let $\lambda \in \mathbb{N}$ be the security parameter. A secret-key quantum encryption scheme $(\SKQES)$ is a triplet $\Pi = (\KeyGen,\Enc,\Dec)$ which consists of the following $\QPT$ algorithms:
\begin{itemize}
    \item (key generation) $\KeyGen: 1^\lambda \mapsto k \in \mathcal{K}$.
    \item (encryption) $\Enc: \mathcal{K} \times \mathcal{D}(\mathcal{H}_M) \rightarrow \mathcal{D}(\mathcal{H}_C)$.
    \item (decryption) $\Dec: \mathcal{K} \times \mathcal{D}(\mathcal{H}_C) \rightarrow \mathcal{D}(\mathcal{H}_M)$.
\end{itemize}
such that $\|\Dec_k \circ \Enc_k - \Id_M\|_\diamond \leq \negl(\lambda)$, for all $k \in \boldsymbol{supp}\, \KeyGen(1^\lambda)$, i.e., belonging to the set of all possible outcomes.
\end{definition}

In a similar fashion, we also define quantum encryption of classical messages $(\QECM)$ schemes $\Pi = (\KeyGen,\Enc,\Dec)$ for which the input to $\Enc$ is in the form of a classical key together with a classical message in a quantum register, for example given by $\ket{m}\bra{m} \in \mathcal{D}(\mathcal{H}_M)$ for a plaintext $m$. The following definition is due to Broadbent and Lord \cite{broadbent2019uncloneable}.

\begin{definition}[Unclonable Security]\label{def:unclonable_security}\ \\
A $\QECM$ scheme $\Pi = (\KeyGen,\Enc,\Dec)$ with message space $\{0,1\}^\lambda$ is $t(\lambda)$-unclonable secure if for all $\QPT$ (in terms of the parameter $\lambda$) cloning attacks $\mathcal{A} = (\mathcal{P},\mathcal{D}_{1},\mathcal{D}_{2})$ consisting of a triplet
\begin{itemize}
    \item (pirate) $\mathcal{P}: \mathcal{D}(\mathcal{H}_C) \rightarrow \mathcal{D}(\mathcal{H}_A \otimes \mathcal{H}_B)$
    \item (1st decoder) $\mathcal{D}_{1}: \mathcal{D}(\mathcal{H}_K \otimes \mathcal{H}_A) \rightarrow \mathcal{D}(\mathcal{H}_M)$
    \item (2nd decoder) $\mathcal{D}_{2}: \mathcal{D}(\mathcal{H}_K \otimes \mathcal{H}_B) \rightarrow \mathcal{D}(\mathcal{H}_M)$
\end{itemize}
against $\Pi$, there exists a negligible function $\mu(\lambda)$ such that:
\begin{align*}
\underset{m}{\mathbb{E}} \, \underset{k \leftarrow \mathcal
{K}}{\mathbb{E}} \,\, \Tr{(\ket{m}\bra{m} \otimes \ket{m}\bra{m})(\mathcal{D}_{1,k} \otimes \mathcal{D}_{2,k}) \circ \mathcal{P} \circ \Enc_{k}(\ket{m}\bra{m}) } \leq 2^{- \lambda + t(\lambda)} + \mu(\lambda).
\end{align*}
\end{definition}

\section{Quantum copy-protection}\label{sec:cp}
Except for some slight differences, our definition of a secure copy-protection scheme is identical to the notion in \cite{aaronson2009quantum}. We elaborate on the differences in Section \ref{sec: comparison}.

\subsection{Definition}

\begin{definition}[Quantum copy-protection scheme]\label{def:copy-protection-scheme}
Let $ \mathcal{C}$ be a family of classical circuits with a single bit output. A quantum copy-protection $(\textsf{CP})$ scheme for $\mathcal{C}$ is a pair of $\QPT$ algorithms $(\mathsf{CP.Protect}, \mathsf{CP.Eval})$ with the following properties:
\begin{itemize}
    \item $\mathsf{CP.Protect}$ takes as input a security parameter $\lambda \in \mathbb{N}$ and a classical circuit $C \in \mathcal{C}$, and outputs a quantum state $\rho$.
    \item $\mathsf{CP.Eval}$ takes as input a quantum state $\rho$ and a string $x$, and outputs a single bit.
\end{itemize}
We say that the scheme is correct if, for any $\lambda \in \mathbb{N}$, $C \in \mathcal{C}$, and any input string $x$ to $C$:
 $$\Pr[\mathsf{CP.Eval}( \rho, x) = C(x) : \rho \leftarrow \mathsf{CP.Protect}(1^{\lambda}, C)] \geq 1-\negl(\lambda).$$
(note that we think of $\mathsf{CP.Protect}$ as a deterministic procedure which outputs a mixed state; so the probability in the above expression comes from $\mathsf{CP.Eval}$.)
\end{definition}

Two remarks are in order about Definition \ref{def:copy-protection-scheme}. \begin{itemize}
    \item For ease of exposition, we define copy-protection for circuits outputting a single bit. It is straightforward to generalize the above definition to circuits with multi-bit output. 
    \item While our definition only requires the ability to compute the circuit $C$ on a single input using the copy-protection state $\rho$, polynomially many evaluations are possible by delaying any measurements in $\mathsf{CP.Eval}$, and uncomputing after copying the output bit to recover the approximate original state $\rho$.
\end{itemize}

We define security in terms of a game between a challenger and an adversary consisting of a triple of $\QPT$ algorithms $\mathcal{A} = (\mathcal{P},\mathcal{F}_1, \mathcal{F}_2)$ - a ``pirate'' $\mathcal{P}$ and two ``freeloaders'' $\mathcal{F}_1$ and $\mathcal{F}_2$. The game is specified by a security parameter $\lambda$, a distribution $D_{\lambda}$ over circuits in $\mathcal{C}$, an ensemble $\{D_C\}_{C \in \mathcal{C}}$ where $D_C$ is a distribution over pairs of inputs to $C \in \mathcal{C}$.
We refer to $ \{D_{\lambda}\}_{\lambda \in \mathbb{N}}$ as the \textit{program ensemble}, and to $\{D_{C}\}_{C \in \mathcal{C}}$ as the \textit{input challenge ensemble}. The security game proceeds as follows.
\begin{itemize}
    \item The challenger samples $C \leftarrow D_{\lambda}$ and sends $\rho \leftarrow \mathsf{CP.Protect}(1^\lambda, C)$ to $\mathcal{P}$.
    \item $\mathcal{P}$ creates a state on registers \textsf{A} and \textsf{B}, and sends \textsf{A} to $\mathcal{F}_1$ and \textsf{B} to $\mathcal{F}_2$.
    \item (\emph{input challenge phase:}) The challenger samples $(x_1, x_2) \leftarrow D_C$ and sends $x_1$ to $\mathcal{F}_1$ and $x_2$ to $\mathcal{F}_2$. ($\mathcal{F}_1$ and $\mathcal{F}_2$ are not allowed to communicate).
    \item $\mathcal{F}_1$ and $\mathcal{F}_2$ each return bits $b_1$ and $b_2$ to the challenger.
\end{itemize}
$\mathcal{A} = (\mathcal{P}, \mathcal{F}_1, \mathcal{F}_2)$ win if $b_1 = C(x_1)$ and $b_2 = C(x_2)$. We use the random variable $\mathsf{PiratingGame}(\lambda, \mathcal{P}, \mathcal{F}_1, \mathcal{F}_2, D_{\lambda}, \{D_C\})$ to denote whether the game is won.

Before defining security, we define $p^{\mathrm{triv}}_{D_\lambda,\{D_C\}_{C \in \mathcal{C}}}$ to be the winning probability that is trivially possible due to correctness: the pirate forwards the copy-protected program to one of the freeloaders, and leaves the other one with guessing as his best option. Formally, let $\hat D_C$ be the induced distribution of winning answer pairs, and let $\hat D_{C,i}$, for $i \in \{1,2\}$ be its marginals. We define the optimal guessing probability of any of the two freeloaders, 
$$
p^{\mathrm{triv}}_{D_\lambda,\{D_C\}_{C \in \mathcal{C}}}=\max_{i\in\{1,2\}}\max_{b\in\{0,1\}}\mathbb E_{C\leftarrow D_\lambda}\hat D_{C,i}(b).
$$ 

We define security as follows:

\begin{definition}[Security]
\label{def: security copy protection}
A quantum copy-protection scheme for a family of circuits $ \mathcal{C}$ is said to be $\delta$-secure with respect to the ensemble $\{D_{\lambda}\}_{\lambda \in \mathbb{N}}$ of distributions over circuits in $\mathcal{C}$, and with respect to the ensemble $\{D_{C}\}_{C \in \mathcal C}$, where $D_{C}$ is a distribution over pairs of inputs to program $C \in \mathcal{C}$, if for any $\QPT$ adversary $(\mathcal{P}, \mathcal{F}_1, \mathcal{F}_2)$, any $\lambda \in \mathbb{N}$,
$$ \Pr[\textsf{PiratingGame}(\lambda, \mathcal{P}, \mathcal{F}_1, \mathcal{F}_2, D_{\lambda}, \{D_{C}\}) =1 ] \leq 1-\delta(\lambda)+\negl(\lambda)\,.$$
If $\delta(\lambda)=1-p^{\mathrm{triv}}_{D_\lambda,\{D_C\}_{C \in \mathcal{C}}}$, we simply say that the copy-protection scheme is secure.
\end{definition}
Two remarks are in order about the above definition.
\begin{itemize}
    \item The definition can be generalized by quantifying over all challenge distributions. The acceptable adversarial winning probability then needs to be related to the optimal guessing probability for challenges drawn from the distribution. We refrain from such a generalization for ease of exposition.  
    \item We follow Aaronson \cite{Aaronson_2005} in that the parameter $\delta$ quantifies \emph{security}, not \emph{insecurity}. We decided in favor of this convention to maintain coherence with the previous literature on quantum copy protection, despite the fact that in cryptography, the $\epsilon$ in ``$\epsilon$-secure'' traditionally quantifies adversarial advantage which is a measure of \emph{in}security.
\end{itemize}

\subsection{Comparison with existing definitions of copy-protection}
\label{sec: comparison}

Our definition is very similar to the original security definition first proposed by Aaronson \cite{aaronson2009quantum}. The only difference is the following. In \cite{aaronson2009quantum}, a scheme is $\delta$-secure if for any bounded adversary who tries to create $k+1$ programs upon receiving $k$ copy-protected copies the average number of input challenges answered correctly is $k(1+\delta)$. In contrast, in our definition we say that the scheme is secure if no adversary can succeed with non-negligible advantage beyond the trivial guessing probability. In our work, we exclusively focus on the case of $k=1$.


\section{Quantum copy-protection from unclonable encryption}

In this section, we make a conceptual connection between unclonable quantum encryption and quantum copy-protection. Our main result is a quantum copy-protection scheme for multi-bit point functions which we obtain from any unclonable encryption scheme with a so-called ``wrong-key detection mechanism''. Canetti et al. \cite{tcc-2010-18737} previously introduced a similar property for classical encryption in the context of point function obfuscation. 

\subsection{Quantum encryption with wrong-key detection}

Let us first formalize the ``wrong-key detection mechanism'' for quantum encryption schemes.

\begin{definition}[Wrong-Key Detection] Let $(\KeyGen,\Enc,\Dec)$ be a $\SKQES$. We say that the scheme satisfies the wrong-key detection (WKD) property if, for every $k'\neq k \leftarrow \KeyGen(1^\lambda)$:
$$
\| \Dec_{k'} \circ \Enc_k - \langle {\ket{\bot}\bra{\bot} \rangle} \|_\diamond \leq \negl(\lambda).
$$
Here,  $\langle {\ket{\bot}\bra{\bot} \rangle}(\cdot) = \ket{\bot}\bra{\bot} \mathrm{Tr}[\cdot]$.
\end{definition}

Next, we give a simple transformation that achieves WKD in the QROM.

\begin{construction}[Generic Transformation for WKD in the QROM]\label{cons:key-detection}\ \\
Let $\Pi=(\KeyGen,\Enc,\Dec)$ be a $\QECM$, let $\lambda$ be the security parameter and fix a function $H: \{0,1\}^\lambda \rightarrow \{0,1\}^\ell$. The $\QECM$ $\Pi_H=(\KeyGen',\Enc',\Dec')$ scheme is defined by the following $\QPT$ algorithms:
\begin{itemize}
    \item  $\KeyGen'$: on input $1^\lambda$, run $\KeyGen( 1^\lambda)$ to output $k\in \mathcal{K}$.
    \item $\Enc'$: on input $m$, run $\Enc: \mathcal{K} \times \mathcal{D}(\mathcal{H}_M) \rightarrow \mathcal{D}(\mathcal{H}_C)$ and output $(\Enc_k(\ket{m}\bra{m}),H(k))$.
    \item $\Dec'$: on input $(\rho,c)$, first check whether $H(k)=c$. Output $\ket{\bot}\bra{\bot}$, if false. Otherwise, run $\Dec: \mathcal{K} \times \mathcal{D}(\mathcal{H}_C) \rightarrow \mathcal{D}(\mathcal{H}_M)$ and output $\Dec_k(\rho)$.
\end{itemize}
\end{construction}

\begin{lem}\label{lem:generic_transf}
Let $\Pi$ be any $t(\lambda)$-unclonable secure $\QECM$ and let $H: \{0,1\}^\lambda \rightarrow \{0,1\}^\ell$ be a hash function, for $ \ell = 2 \lambda$. Then, Constr. \ref{cons:key-detection} yields an $t(\lambda)$-unclonable secure $\QECM$ scheme $\Pi_H$ with WKD in the QROM.
\end{lem}

\begin{proof}
Correctness is clearly preserved. Let us first verify the WKD property of the scheme $\Pi_H =(\KeyGen',\Enc',\Dec')$ in the QROM. It is not hard to see that the WKD property depends on the collision probability for the event that $H(k) = H(k')$, for some $k' \in \{0,1\}^\lambda \setminus \{k\}$. In fact, we can express the quantum channel $\Dec'_{k'} \circ \Enc'_k$ as follows:
$$
\Dec'_{k'} \circ \Enc'_k =  \Pr[\textsc{Coll}] \Dec_{k'} \circ \Enc_k + (1-\Pr[\textsc{Coll}]) \langle{\ket{\bot}\bra{\bot} \rangle}.
$$
Moreover, by the birthday bound, we have
\begin{align*}
\Pr[\textsc{Coll}]&=\Pr_H[\exists k' \in \{0,1\}^\lambda \setminus \{k\} \, : \, H(k) = H(k')]\\
&\leq \sum_{k' \in \{0,1\}^\lambda \setminus \{k\}}\Pr_H[H(k) = H(k')] = \frac{2^\lambda -1}{2^{2\lambda}}.
\end{align*}
Hence, we can readily verify the WKD property as follows:
\begin{align*}
&\| \Dec'_{k'} \circ \Enc'_k - \langle {\ket{\bot}\bra{\bot} \rangle} \|_\diamond \\ &= \underset{\rho_{MM'}}{\max} \| (\Dec'_{k'} \circ \Enc'_k \otimes \Id_{M'}) (\rho_{MM'}) - (\langle{\ket{\bot}\bra{\bot} \rangle} \otimes \Id_{M'}) (\rho_{MM'})  \|_1\\
&\leq   \underset{\rho_{MM'}}{\max} \left\| \Pr[\textsc{Coll}]  \Dec_{k'} \circ \Enc_k(\rho_{MM'}) -  \Pr[\textsc{Coll}](\langle{\ket{\bot}\bra{\bot} \rangle} \otimes \Id_{M'}) (\rho_{MM'}) \right\|_1 \\
&\leq \Pr[\textsc{Coll}]\underset{\rho_{MM'}}{\max} \Big(\|  \rho_{MM'} \|_1 + \| (\langle{\ket{\bot}\bra{\bot} \rangle} \otimes \Id_{M'}) (\rho_{MM'}) \|_1 \Big) \leq \frac{2^\lambda -1}{2^{2\lambda -1}} \, = \, \negl(\lambda).
\end{align*}
For security, let $\Pi=(\KeyGen,\Enc,\Dec)$ and recall that $\Enc'_k(\ket{m}\bra{m}) = (\Enc_k(\ket{m}\bra{m}),H(k))$ according to Constr. \ref{cons:key-detection}. Let $\mathcal{A} = (\mathcal{P},\mathcal{D}_{1},\mathcal{D}_{2})$ be an adversary against the unclonable security game with respect to $\Pi_H$. We give a reduction from the unclonable security of $\Pi$. Suppose that $\mathcal{P}$ receives access to a re-programmed oracle $H_{k,z}$, where $H(k)=z$, and that $\mathcal{P}$ makes at most $q=\poly(\lambda)$ queries in total.  Without loss of generality, we assume that $\mathcal{P}$ is specified by $(U O^H)^q$, for some unitary $U$. It is sufficient to argue that the following two global states are negligibly close in trace distance: 
\begin{align}& \mathbb{E}_H \mathbb{E}_k \mathbb{E}_{z} \ket{H}\bra{H} \otimes \ket{k}\bra{k} \otimes \left((U O^{H_{k,z}})^{q} \Enc_k(\ket{m}\bra{m}) \otimes \ket{z}\bra{z} \left((U O^{H_{k,z}})^{q}\right)^{\dagger}\right)\nonumber \\ \approx \,& \mathbb{E}_H \mathbb{E}_k \mathbb{E}_{z } \ket{H}\!\bra{H}\!\otimes\!\ket{k}\!\bra{k} \!\otimes\!\left((U O^{H})^{q} \Enc_k(\ket{m}\!\bra{m}) \!\otimes\!\ket{z}\!\bra{z}\left((U O^{H})^{q}\right)^{\dagger} \right)\,. \label{eq: states almost equal 2_QECM}
\end{align}
We use the one-way-to-hiding (Lemma \ref{lem: qrom technical step}) to deduce that the above distance is negligible, so long as the following quantity is negligible:
\begin{align}
\mathbb{E}_{H}\mathbb{E}_{k} \mathbb{E}_{z} \mathbb{E}_\nu \Tr{\ket{k}\!\bra{k} (U O^H)^\nu  \Enc_k(\ket{m}\!\bra{m}) \otimes \ket{z}\!\bra{z} \big((U O^H)^\nu\big)^\dag }\!.
\end{align}
Suppose for the sake of contradiction that the latter is non-negligible. Then, we can construct an adversary that wins at the unlconable security game against $\Pi$. The reduction is straightforward: the adversary for the unlconable security game runs $\mathcal{P}$ (by simulating the random oracle $H$) to extract $k$. The adversary then decrypts the challenge ciphertext using $k$, and forwards the appropriate message $m$ to the decoders $\mathcal{D}_{1}$ and $\mathcal{D}_{2}$. By the assumption that $\mathcal{P}$ succeeds with non-negligible probability, so does the adversary against the unclonable security of $\Pi$, yielding a contradiction.

Using Eq.~\ref{eq: states almost equal 2_QECM}, Lemma \ref{lem:TD_inequalities} and that $\Pi$ is $t(\lambda)$-unclonable secure it follows that, for all $\QPT$ cloning attacks $\mathcal{A} = (\mathcal{P},\mathcal{D}_{1},\mathcal{D}_{2})$
against $\Pi_H$, there exists a negligible function $\mu(\lambda)$ such that:
\begin{align*}
\underset{m}{\mathbb{E}} \, \underset{k \leftarrow \mathcal
{K}}{\mathbb{E}} \!\Tr{(\ket{m}\!\bra{m} \otimes \ket{m}\!\bra{m})(\mathcal{D}_{1,k} \otimes \mathcal{D}_{2,k}) \circ \mathcal{P} \circ \Enc_{k}(\ket{m}\!\bra{m}) } \leq 2^{- \lambda + t(\lambda)} + \mu(\lambda).
\end{align*}
We conclude that $\Pi_H$ is $t(\lambda)$-unclonable secure.
\qed\end{proof}

\subsection{Quantum copy-protection of multi-bit point functions from unclonable encryption schemes with wrong-key detection}\label{sec: multibit}

We are now ready to state our quantum copy-protection scheme for multi-bit point functions which we obtain from any unclonable encryption scheme with the aforementioned ``wrong-key detection mechanism''. Here, we consider multi-bit point functions $P_{y,m}$ of the form
$$ P_{y,m} (x) = \begin{cases} m   & \text{if } x = y\,,\\
    0^\lambda &\text{if } x \neq y \,,   \end{cases}  $$ 
where $y,m \in \{0,1\}^\lambda$. Our construction is the following

\begin{construction}[Quantum CP scheme for multi-bit point functions]\label{cons:cp_UQE}\ \\ To construct a CP scheme for multi-bit point functions with input and output sizes $\lambda$, respectively, let $\Pi = (\KeyGen,\Enc,\Dec)$ be a $\QECM$ with WKD, with security parameter and message length equal to $\lambda$. We define the CP scheme $(\mathsf{CP.Protect},\mathsf{CP.Eval})$ as follows:

\begin{itemize}
\item $\mathsf{CP.Protect}(1^{\lambda}, P_{y,m})$: Takes as input a security parameter $\lambda$ and a multi-bit point function $P_{y,m}$, succinctly specified by the marked input $y$ (of size $\lambda$) and message $m$ (of size $\lambda$), and  outputs the quantum ciphertext given by $\Enc_y(m)$.
\item $\mathsf{CP.Eval}(1^{\lambda}, \rho ; x)$: Takes as input a security parameter $\lambda$, an alleged copy-protected program $\rho$, and a string $x \in \{0,1\}^\lambda$ (the input on which the program is to be evaluated). Appends an ancillary qubit in the $\ket{0}$ state. Then, coherently\footnote{If $\Dec$ is not unitary, performing this measurement coherently requires purifying, or dilating, it} performs a two-outcome measurement to check whether $\Dec_x(\rho)$ is in the state $\ket{\bot}\bra{\bot}$, or not, and stores the resulting bit in the ancilla. If true, output $0^\lambda$. Otherwise, rewind the procedure and measure in the standard basis to obtain a message $m'$.
\end{itemize}
\end{construction}

Before stating our main theorem on the security of Constr. \ref{cons:cp_UQE}, we define the following two classes of distributions with respect to multi-bit point functions $P_{y,m}$ of the form
$$ P_{y,m} (x) = \begin{cases} m   & \text{if } x = y\,,\\
    0^\lambda &\text{if } x \neq y \,.   \end{cases}  $$
First, we let $D = \{D_{\lambda}\}$ be an ensemble of distributions $D_{\lambda}$ over multi-bit point functions $P_{y,m}$ over $\{0,1\}^{\lambda}$ that sample a marked input $y$ as well an output message $m$ uniformly at random with respect to $\{0,1\}^{\lambda}$. Further, by $D' = \{D_y\}$ we denote an arbitrary ensemble of challenge distributions, where each $D_y$ is a distribution of challenge input pairs to the program.
We now prove the security of Construction~\ref{cons:cp_UQE} with respect to the aforementioned distributions $D,D'$.

\begin{theorem}\label{thm:point-function-CP-from-UCE}
Let $\Pi = (\KeyGen,\Enc,\Dec)$ be any $t(\lambda)$-unclonable secure $\QECM$ with WKD such that $
\lambda - t = \omega(\log \lambda)$. Then, Construction \ref{cons:cp_UQE} yields a secure quantum copy-protection for multi-bit point functions with respect to the pair of ensembles $(D, D')$, against  computationally-bounded adversaries.
\end{theorem}

\begin{proof} The correctness of the CP scheme follows directly from the WKD property of the $\QECM$. Let $\mathcal{A} = (\mathcal{P}_{\mathcal{A}},\mathcal{F}_{1},\mathcal{F}_{2})$ denote the adversary for $\mathsf{PiratingGame}$. We consider two cases, namely when $p^{\mathrm{triv}}_{D_\lambda,D_y}=1$ and when $p^{\mathrm{triv}}_{D_\lambda,D_y} <1$ (depending on the challenge distribution $D' = \{D_y\}$). In the former case, the scheme is trivially secure by definition and we are done. Hence, we will assume that $p^{\mathrm{triv}}_{D_\lambda,D_y}<1$ for the remainder of the proof. Note that, in this case, the distribution $D_y$ has non-zero weight on the marked input $y$.

Let $(x_1,x_2) \leftarrow D_y$ denote the inputs received by the freeloaders $\mathcal{F}_{1}$ and $\mathcal{F}_{2}$ during the challenge phase. 
We can express the probability that $\mathcal{A}$ succeeds at $\mathsf{PiratingGame}$ as follows:
\begin{align}
&\Pr[ \mathcal{A} \text{ wins}] \nonumber\\
&\! =\! \Pr[ \mathcal{A} \text{ wins} \,|\, x_1 \!\neq\! y \!\neq\! x_2] \!\cdot\!\Pr[x_1 \!\neq\! y \!\neq\! x_2]\!+\!\Pr[ \mathcal{A} \text{ wins} \,|\, x_1 \!=\! y \!\neq\! x_2] \!\cdot\!\Pr[x_1 \!=\! y \!\neq\! x_2]\nonumber\\
& \!+\!\Pr[ \mathcal{A} \text{ wins} \,|\, x_1 \!\neq\! y \!=\! x_2] \!\cdot\!\Pr[x_1 \!\neq\! y \!=\! x_2]\!+\!\Pr[ \mathcal{A} \text{ wins} \,|\, x_1 \!=\! y \!=\! x_2] \!\cdot\!\Pr[x_1 \!=\! y \!=\! x_2].
\end{align}
Without loss of generality, we assume that $\Pr[x_1 = y \!\neq\! x_2] \leq \Pr[x_1 \!\neq\! y = x_2]$. Hence,
\begin{align}
&\Pr[ \mathcal{A} \text{ wins}] \leq  \Pr[ \mathcal{A} \text{ wins} \,|\, x_1 \!\neq\! y \!\neq\! x_2] \cdot \Pr[x_1 \!\neq\! y \!\neq\! x_2]\nonumber\\
& +\big(\Pr[ \mathcal{A} \text{ wins} \,|\, x_1 = y \!\neq\! x_2]+\Pr[ \mathcal{A} \text{ wins} \,|\, x_1 \!\neq\! y = x_2]\big) \cdot \Pr[x_1 \!\neq\! y = x_2]\nonumber\\
& +\Pr[ \mathcal{A} \text{ wins} \,|\, x_1 = y = x_2] \cdot \Pr[x_1 = y = x_2]. \label{eq:A_wins}
\end{align}
Let us now state the following simple inequality. By first applying the union bound and then using that $\mathcal{F}_{1}$ and $\mathcal{F}_{2}$ are non-signalling, we find that:
\begin{align}
&\Pr[ \mathcal{A} \text{ wins} \,|\, x_1 = y =  x_2]\nonumber\\
&= \Pr[\mathcal{F}_{1} \text{ succeeds} \land \mathcal{F}_{2} \text{ succeeds} \,|\, x_1 = y = x_2]\nonumber\\
&\geq \Pr[\mathcal{F}_{1} \text{ succeeds} \,|\, x_1 = y = x_2] + \Pr[ \mathcal{F}_{2} \text{ succeeds} \,|\, x_1 = y = x_2] -1\nonumber\\
&= \Pr[\mathcal{F}_{1} \text{ succeeds} \,|\, x_1 = y \neq x_2] + \Pr[ \mathcal{F}_{2} \text{ succeeds} \,|\, x_1 \neq y = x_2] -1\nonumber\\
&\geq \Pr[\mathcal{A} \text{ wins} \,|\, x_1 = y \neq x_2] + \Pr[ \mathcal{A} \text{ wins} \,|\, x_1 \neq y = x_2] -1.
\end{align}
Plugging this into Eq.~\eqref{eq:A_wins}, we obtain the following upper bound:
\begin{align}
\Pr[ \mathcal{A} \text{ wins}] & \leq  \Pr[ \mathcal{A} \text{ wins} \,|\, x_1 \neq y \neq x_2] \cdot \Pr[x_1 \neq y \neq x_2]\nonumber\\
& +\big(1 + \Pr[ \mathcal{A} \text{ wins} \,|\, x_1 = y =  x_2]\big) \cdot \Pr[x_1 \neq y = x_2]\nonumber\\
& +\Pr[ \mathcal{A} \text{ wins} \,|\, x_1 = y = x_2] \cdot \Pr[x_1 = y = x_2]\nonumber\\
& \leq  \Pr[x_1 \neq y \neq x_2] + \Pr[x_1 \neq y = x_2]+2\Pr[ \mathcal{A} \text{ wins} \,|\, x_1 = y = x_2]\nonumber\\
&= p^{\mathrm{triv}}_{D_\lambda,D_y} + 2\Pr[ \mathcal{A} \text{ wins} \,|\, x_1 = y = x_2].\label{eq:A_wins3}
\end{align}
In the last line, we used the assumption that $\Pr[x_1 = y \neq x_2] \leq \Pr[x_1 \neq y = x_2]$ together with the following simple identity for the trivial guessing probability:
$$
p^{\mathrm{triv}}_{D_\lambda,D_y} = \Pr[x_1 \neq y \neq x_2] + \max\big\{\Pr[x_1 \neq y = x_2],\Pr[x_1 = y \neq x_2]\big\}.
$$
We complete the proof by showing that $\Pr[ \mathcal{A} \text{ wins} \,|\, x_1 = y = x_2] \leq \negl(\lambda)$. This implies that
\begin{align}
\Pr[ \mathcal{A} \text{ wins}] \leq p^{\mathrm{triv}}_{D_\lambda,D_y} + \negl(\lambda).\label{eq:CP_security_bound}
\end{align}
Suppose that $\mathcal{A} = (\mathcal{P}_{\mathcal{A}},\mathcal{F}_{1},\mathcal{F}_{2})$ succeeds with probability $p = \Pr[ \mathcal{A} \text{ wins} \,|\, x_1 = y = x_2]$ on the challenge pair consisting of $x_1 =y$ and $x_2=y$.
We will use $\mathcal{A}$ to construct an adversary against the unclonable security of the $\QECM$ scheme $\Pi$. Consider the $\QPT$ adversary $\mathcal{B} = (\mathcal{P}_{\mathcal{B}},\mathcal{D}_{1},\mathcal{D}_{2})$ against $\Pi$, which we define as follows:
\begin{itemize}
    \item $\mathcal{P}_{\mathcal{B}}$ receives the state $\rho = \Enc_y(\ket{m}\bra{m})$ and runs the pirate $\mathcal{P}_{\mathcal{A}}$ on $\rho$. Next, $\mathcal{P}_{\mathcal{B}}$ passes the two registers output by $\mathcal{P}_{\mathcal{A}}$ to the decoders $\mathcal{D}_{1}$ and $\mathcal{D}_{2}$.
    \item The decoders $\mathcal{D}_{1}$ and $\mathcal{D}_{2}$ each receive the marked input $y$ and then run the freeloaders $\mathcal{F}_{1}$ and $\mathcal{F}_{2}$, respectively, on the two registers prepared by $\mathcal{P}_{\mathcal{A}}$. Finally, the decoders output the outcomes obtained from running the freeloaders.
\end{itemize}
Since $\Pi$ is $t(\lambda)$-unclonable secure, there exists a negligible $\mu(\lambda)$ such that:
\begin{align*}
\underset{m}{\mathbb{E}} \, \underset{y}{\mathbb{E}} \,\, \Tr{(\ket{m}\bra{m} \otimes \ket{m}\bra{m})(\mathcal{D}_{1,y} \otimes \mathcal{D}_{2,y}) \circ \mathcal{P}_{\mathcal{B}} \circ \Enc_{k}(\ket{m}\bra{m}) } \leq 2^{- \lambda + t(\lambda)} + \mu(\lambda).
\end{align*}
Since $
\lambda - t = \omega(\log \lambda)$, we conclude that $\mathcal{A}$ succeeds with probability $p \leq \negl(\lambda)$.
This completes the proof of Eq.~\eqref{eq:CP_security_bound}, and thus the proof of Thm. \ref{thm:point-function-CP-from-UCE}.
\end{proof}

Finally, when applying the WKD transformation from Construction \ref{cons:key-detection} to the $\log_2(9)$-unclonable encryption scheme in~\cite{broadbent2019uncloneable}, we obtain the following theorem  with respect to the aforementioned distributions $D,D'$.

\begin{theorem}\label{thm:point-function-CP-from-UCE-general}
There exists a $\log_2(9)$-unclonable secure $\QECM$ scheme with WKD for which Construction \ref{cons:cp_UQE} yields a secure quantum copy-protection scheme for multi-bit point functions with respect to the pair of ensembles $(D, D')$, against  query-bounded (computationally bounded) adversaries in the QROM.
\end{theorem}


\section{Quantum copy-protection of point functions}
\label{sec: main}
In this section, we present a quantum copy-protection scheme for point functions, and prove its security in the quantum random oracle model (QROM).

In what follows, we consider a copy-protection scheme for the class of point functions $P_y$ with marked input $y \in \{0,1\}^n$. Note that, for simplicity, we hand $y$ to the copy-protection algorithm as an input, rather than as a circuit for the point function $P_y$ itself. 

\begin{remark}\label{remark:lambda-n}
In the following description, the size of the marked input $n$ and the security parameter $\lambda$ are distinct variables (as they should be in principle). However, the security guarantee that we will prove is with respect to ensembles of programs for which $n=\lambda$. When copy-protecting an ensemble of programs, the level of security cannot be independent of the size of the inputs to the programs, since a pirate with access to the copy-protected program can always determine the whole truth table of the program in time $2^n$.
\end{remark}

Recall that for $v, \theta \in \{0,1\}^{\lambda}$, we use the notation $\ket{v^\theta} = H^{\theta} \ket{v}$, where $H^{\theta} = H^{\theta_1} \otimes \ldots \otimes H^{\theta_{\lambda}}$ is the Hadamard operator on $\lambda$ many qubits. Our construction is the following:

\begin{construction}[Copy-protection scheme for point functions]\label{cons:cp} Let $\lambda$ be the security parameter, and let $G: \{0,1\}^{n} \rightarrow \{0,1\}^{m(\lambda)}$ and $H: \{0,1\}^{m(\lambda)} \rightarrow  \{0,1\}^{\lambda}$ be hash functions, where $m(\lambda) > \lambda$, and $n$ is the input size of the point function. We define $(\mathsf{CP.Protect},\mathsf{CP.Eval})$ as follows:

\begin{itemize}
\item $\mathsf{CP.Protect}(1^{\lambda}, y)$: Takes as input a security parameter $\lambda$ and a point function $P_y$, succinctly specified by the marked input $y$ (of size $n$).
\begin{itemize}
    \item Set $\theta = G(y)$.
    \item Sample $v \leftarrow \{0,1\}^{m(\lambda)}$ uniformly at random. Let $z = H(v)$.
        \item Output $(\ket{v^{\theta}}, z)$.
\end{itemize}
\item $\mathsf{CP.Eval}(1^{\lambda}, (\rho, z) ; x)$: Takes as input a security parameter $\lambda$, an alleged copy-protected program $(\rho, z)$, and a string $x \in \{0,1\}^n$ (the input on which the program is to be evaluated).
\begin{itemize}
    \item Set $\theta' = G(x)$.
    \item Apply the Hadamard operator $H^{\theta'}$ to $\rho$. Append $n+1$ ancillary qubits, all in state $\ket 0$, and compute the hash function $H$ with input $\rho$ into the first $n$ of them (possibly making use of additional ancillary qubits). Then, coherently measure whether the first $n$ ancilla qubits are in state $\ket z$, recording the result in the last ancilla qubit, uncompute the hash function $H$ and undo the Hadamards $H^{\theta'}$. Finally, measure the last ancilla qubit to obtain a bit $b$ as output.
\end{itemize}
\end{itemize}
\end{construction}

In what follows, we will model both $G$ and $H$ in Construction \ref{cons:cp} as random oracles on the the appropriate domain and co-domains, i.e. we operate in the quantum random oracle model (QROM).
Before stating our main theorem about the security of Construction \ref{cons:cp}, we define the following class of distributions over programs.

\begin{itemize}
    \item $\mathcal{D}_{\mathsf{PF}\mbox{-}\mathsf{UNP}}$. The class of \textit{unpredictable point function distributions} $\mathcal{D}_{\mathsf{PF}\mbox{-}\mathsf{UNP}}$ consists of ensembles $D = \{D_{\lambda}\}$ where $D_{\lambda}$ is a distribution over point functions on $\{0,1\}^{\lambda}$ such that $P_y \leftarrow D_{\lambda}$ satisfies $\Hmin(y) \geq \lambda^\epsilon$ for some $\epsilon>0$.
\end{itemize}
We also define the following class of distributions over input challenges.
\begin{itemize}
\item $\mathcal{D}_{\mathsf{PF}\mbox{-}\mathsf{Chall}}$. An ensemble $D = \{D_y\}$, where each $D_y$ is a distribution over $\{0,1\}^{|y|} \times \{0,1\}^{|y|}$, belongs to the class $\mathcal{D}_{\mathsf{PF}\mbox{-}\mathsf{Chall}}$ if there exists an efficiently sampleable family $\{X_{\lambda}\}$  of distributions over $\{0,1\}^{\lambda}$ with $\Hmin(X_{\lambda}) \geq \lambda^\epsilon$, for some $\epsilon >0$, such that $D_y$ is the following distribution (where $\lambda = |y|$):
    \begin{itemize}
        \item with probability $1/3$, sample $x \leftarrow X_{\lambda}$, and output $(x,y)$. 
    \item with probability $1/3$, sample $x \leftarrow X_{\lambda}$, and output $(y,x)$. 
    \item with probability $1/3$, sample $x,x' \leftarrow X_{\lambda}$, and output $(x,x')$.  
    \end{itemize}
We say the ensemble $D$ is \emph{specified} by the ensemble $X_{\lambda}$.
\end{itemize}

We finally define two classes of distributions over pairs of programs and challenges.
\begin{itemize}
\item $\mathcal{D}_{\mathsf{PF}\mbox{-}\mathsf{pairs}\mbox{-}\mathsf{stat}}.$ This consists of pairs of ensembles $\left(D = \{D_{\lambda}\}, D' = \{D'_{y}\} \right)$ such that:
\begin{itemize}
    \item $D \in \mathcal{D}_{\mathsf{PF}\mbox{-}\mathsf{UNP}}$ and $D' \in \mathcal{D}_{\mathsf{PF}\mbox{-}\mathsf{Chall}}$.
    \item Let $D'$ be specified by the family  $\{X_{\lambda}\}$, and denote by  $\mathsf{MarkedInput}(D_{\lambda})$ the following distribution 
    over $\{0,1\}^{\lambda}$: sample $P_y \leftarrow D_{\lambda}$, and ouput $y$. We require the families $\{X_{\lambda}\}$ and $\{\mathsf{MarkedInput}(D_{\lambda})\}$ to be statistically indistinguishable. 
\end{itemize}  
\item $\mathcal{D}_{\mathsf{PF}\mbox{-}\mathsf{pairs}\mbox{-}\mathsf{comp}}.$ This is defined in the same way as $\mathcal{D}_{\mathsf{PF}\mbox{-}\mathsf{pairs}\mbox{-}\mathsf{stat}}$, except that we only require $\{X_{\lambda}\}$ and $\{\mathsf{MarkedInput}(D_{\lambda})\}$ to be \emph{computationally} indistinguishable. 
\end{itemize}

We are ready to state our main theorem about security of Construction \ref{cons:cp}.
\begin{theorem}
\label{thm: security}
There exists a constant $\delta^*>0$ such that, the scheme of Construction \ref{cons:cp}, with $m(\lambda) > 5\lambda$, is a $\delta^*$-secure quantum copy-protection scheme for point functions with respect to any pair of ensembles $(D, D') \in \mathcal{D}_{\mathsf{PF}\mbox{-}\mathsf{pairs}\mbox{-}\mathsf{stat}}$ ($\in \mathcal{D}_{\mathsf{PF}\mbox{-}\mathsf{pairs}\mbox{-}\mathsf{comp}}$), against query-bounded (computationally bounded) adversaries in the quantum random oracle model (assuming the existence of quantum-secure one-way functions).
\end{theorem}
We emphasize that quantum-secure one-way functions are only needed for the computational version of Theorem \ref{thm: security}.

Correctness of Construction \ref{cons:cp} is immediate to verify, and makes use the fact that $G$ is a random oracle with a range that is sufficiently larger than its domain. The next section is devoted to proving security. 

\subsection{Proof of security}
\label{sec: security}

We assume that $(D, D') \in \mathcal{D}_{\mathsf{PF}\mbox{-}\mathsf{pairs}\mbox{-}\mathsf{stat}}$ (the case of $(D, D') \in \mathcal{D}_{\mathsf{PF}\mbox{-}\mathsf{pairs}\mbox{-}\mathsf{comp}}$ is analogous except for a slight difference in the proof of Lemma \ref{lem: new hybrids 1}, which we will point out). Moreover, following the notation for ensembles in $(D,D') \in \mathcal{D}_{\mathsf{PF}\mbox{-}\mathsf{pairs}\mbox{-}\mathsf{stat}}$, we let $D'$ be specified by an efficiently sampleable family $\{X_{\lambda}\}$  of distributions over $\{0,1\}^{\lambda}$. As discussed in Remark \ref{remark:lambda-n}, we assume that $n=\lambda$ in Construction \ref{cons:cp}.

We will prove Theorem \ref{thm: security} through a sequence of hybrids. 
\vspace{2mm}

\noindent $H_0$: This is the security game $\mathsf{PiratingGame}$ for the copy-protection scheme of Construction \ref{cons:cp}.
\begin{itemize}
    \item The challenger samples a point function $P_y \leftarrow D_{\lambda}$ with $y \in \{0,1\}^\lambda$, and sends the state $(\ket{v^\theta},z) \leftarrow \mathsf{CP.Protect}(1^\lambda, y)$, where $\theta=G(y)$ and $z = H(v)$, to the pirate $\mathcal{P}$.
    \item $\mathcal{P}$ creates a state on registers \textsf{A} and \textsf{B}, and sends \textsf{A} to $\mathcal{F}_1$ and \textsf{B} to $\mathcal{F}_2$.
    \item (\emph{input challenge phase:}) The challenger samples $(x_1, x_2) \leftarrow D_y'$ and sends $x_1$ to $\mathcal{F}_1$ and $x_2$ to $\mathcal{F}_2$. ($\mathcal{F}_1$ and $\mathcal{F}_2$ are not allowed to communicate).
    \item $\mathcal{F}_1$ and $\mathcal{F}_2$ return $b_1$ and $b_2$, respectively, and win if $b_1 = P_y(x_1)$ and $b_2 = P_y(x_2)$.
\end{itemize}
\vspace{2mm}

\noindent $H_1$: Same as $H_0$, except that in the input challenge phase the challenger samples $(x_1, x_2) \leftarrow D'_y$. Then, it sends $G(x_1)$ and $G(x_2)$ to $\mathcal{F}_1$ and $\mathcal{F}_2$ respectively (instead of sending $x_1$ and $x_2$ directly).
\vspace{2mm}

\noindent $H_2$: Same as $H_1$, except for the following. During the input challenge phase, the challenger samples $(x_1, x_2) \leftarrow D'_y$. Then, for $i \in \{1,2\}$, if $x_i \neq y$, the challenger samples $\theta_i' \leftarrow \{0,1\}^{m(\lambda)}$, and sends $\theta_i'$ to $\mathcal{F}_i$ instead of $G(x_i)$. 
\vspace{2mm}

\noindent $H_3$: Same as $H_2$, except that in the first step of the security game, the challenger samples $\theta \leftarrow \{0,1\}^{m(\lambda)}$ (as opposed to sampling $P_y\leftarrow D_{\lambda}$ and setting $\theta = G(y)$). Then, in the input challenge phase, if $x_i = y$, the challenger sends $\theta_i := \theta$ to $\mathcal{F}_i$, and sends a uniformly random $\theta_i'$ otherwise.
\vspace{2mm}

\noindent $H_4$: Same as $H_3$, except that the challenger samples $z \leftarrow \{0,1\}^{\lambda}$ instead of choosing $z=H(v)$. 
\vspace{2mm}

\noindent $H_5$: Same as $H_4$, except the pirate gets the challenge inputs $\theta_1$ and $\theta_2$ together with the copy-protected program.
\vspace{2mm}

In the rest of the section, we say that $x$ is a $0$-input ($1$-input) to a boolean function $f$, if $f(x)=0$ ($f(x) = 1$). We prove the following two lemmas, which together give Theorem \ref{thm: security}.

\begin{lem}
\label{lem: first}
For any adversary $\mathcal{A}$, 
$$ \Pr[\mathcal{A} \text{ wins } H_5] \leq \frac13 \,.$$
\end{lem}

In the rest of the section, for any function $g: [0,1] \rightarrow [0,1]$, we use the notation 
$$ p(H_i) > g (p(H_j)) - \negl(\lambda)$$
as a shorthand for the following: for any adversary $\mathcal A$ for $H_j$, there exists an adversary $\mathcal{A}'$ for $H_i$ and a negligible function $\mu$ such that 
$$\Pr[\mathcal{A}' \text{ wins } H_i] > g(\Pr[\mathcal A \text{ wins } H_j]) - \mu(\lambda)\,.$$ 

Informally, one can think of $p(H_i)$ the optimal winning probability in hybrid $H_i$ (up to negligible functions in the security parameter).
\begin{lem}
\label{lem: h3 and h1}
There exists a constant $\delta^*>0$ such that \begin{equation}
    p(H_5) >  p(H_0) - 2/3 + \delta^* - \negl(\lambda) \,. \label{eq: lem 7}
\end{equation}
\end{lem}

Lemmas \ref{lem: first} and \ref{lem: h3 and h1} immediately imply that, for any adversary $\mathcal{A}$ for $H_0$,
$$\Pr[\mathcal{A} \text{ wins } H_0] < 1-\delta^* + \negl(\lambda)\,,$$ which gives Theorem \ref{thm: security}.

\begin{proof}[Proof of Lemma \ref{lem: first}]
Denote by $x_1$ and $x_2$ the inputs sampled by the challenger. There are three cases: $x_1$ and $x_2$ are both 0-inputs; $x_1$ is a 0-input and $x_2$ is a 1-input; $x_1$ is a 1-input and $x_2$ is a 0-input.
We argue that the density matrix $\rho$ that is handed to the pirate in all three cases is a maximally mixed state. More precisely,
when $P_y(x_1) = 0, P_y(x_2) = 0$, the state $\rho$ that the pirate receives is the following, which is completely independent of the oracle $H$: 
\begin{align}
&\mathbb{E}_{v, \theta, \theta', \theta'',z} \ket{v^{\theta}}\bra{v^{\theta}} \otimes \ket{\theta'}\bra{\theta'} \otimes \ket{\theta''}\bra{\theta''} \otimes \ket{z}\bra{z} \nonumber\\
&= \frac{\mathds{1}}{2^{3m+{\lambda}}}  \,, \label{eq: 15}
\end{align}

When $P_y(x_1) = 0, P_y(x_2) = 1$, the state $\rho$ that the pirate receives is again completely independent of the oracle, and is the following:
\begin{align}
&\mathbb{E}_{v, \theta, \theta',z} \ket{v^{\theta}}\bra{v^{\theta}} \otimes \ket{\theta'}\bra{\theta'} \otimes \ket{\theta}\bra{\theta}  \otimes \ket{z}\bra{z} \nonumber \\
&= \frac{\mathds{1}}{2^{3m+\lambda}} \,, \label{eq: 17}
\end{align}
where crucially that the state is still maximally mixed.

The third case is analogous to the second. Thus it is impossible, even for an unbounded pirate to distinguish the three cases with any advantage over random guessing.
\end{proof}

We prove Lemma \ref{lem: h3 and h1} by keeping track of how the optimal winning probability changes across hybrids. We break down the proof into the following lemmas.

\begin{lem}
\label{lem: new hybrids 1}
$p(H_1) \geq p(H_0) - \negl(\lambda)$.
\end{lem}

\begin{lem}
\label{lem: new hybrids 2}
$|p(H_1) - p(H_2)| = \negl(\lambda)$.
\end{lem}

\begin{lem}
\label{lem: new hybrids 3}
$|p(H_2) - p(H_3)| = \negl(\lambda)$.
\end{lem}

\begin{lem}
\label{lem: 4}
$p(H_4) > p(H_3) - 2/3 + \delta^* - \negl(\lambda)$, for some constant $\delta^*>0$.
\end{lem}

\begin{lem}
\label{lem: 5}
$p(H_5) \geq p(H_4)$.
\end{lem}

Lemma \ref{lem: h3 and h1} follows immediately from Lemmas \ref{lem: new hybrids 1}-\ref{lem: 5}. The crux is Lemma \ref{lem: 4}.

\begin{proof}[Proof of Lemma \ref{lem: new hybrids 1}]
Let $\mathcal A = (\mathcal{P}, \mathcal{F}_1, \mathcal{F}_2)$ be a query-bounded adversary that wins with probability $p$ in $H_0$. We will construct an adversary $\mathcal{A}'$ that wins with probability at least $p-\negl(\lambda)$ in hybrid $H_1$. The adversary $\mathcal{A}' = (\mathcal{P}', \mathcal{F}_1', \mathcal{F}_2')$ is the following:
\begin{itemize}
    \item Upon receiving $(\ket{\Psi}, z)$ from the challenger, where $\ket{\Psi} = \ket{v^{G(y)}}$ for some string $y \in \{0,1\}^{\lambda}$, $\mathcal{P}'$ samples a uniformly random function $\hat{G}: \{0,1\}^{\lambda} \rightarrow \{0,1\}^{m(\lambda)}$ (in the computational version of the argument, $\mathcal{P}'$ samples instead a function $\hat{G}$ from a quantum-secure pseudo-random function ($\mathsf{PRF}$) family. Note that a quantum-secure $\mathsf{PRF}$ family can be constructed from any quantum-secure one-way function \cite{zhandry2012construct}). Subsequently, $\mathcal{P}'$ runs $\mathcal{P}$ on input $(\ket{\Psi}, z)$, using $\hat{G}$ to respond to queries to $G$. $\mathcal{P}'$ forwards the two registers $\textsf{A}$ and $\textsf{B}$ output by $\mathcal{P}$ to $\mathcal{F}_1'$ and $\mathcal{F}_2'$, respectively, together with a description of $\hat{G}$.
    \item $\mathcal{F}_1'$ and $\mathcal{F}_2'$ receive $w_1 = G(x_1)$ and $w_2 = G(x_2)$ respectively from the challenger, for some $x_1$ and $x_2$. $\mathcal{F}_1'$ and $\mathcal{F}_2'$ sample $x_1', x_2' \leftarrow X_{\lambda}$. $\mathcal{F}_1'$ runs $\mathcal{F}_1$ on input $(x_1', \textsf{A})$ and responds to oracle queries to $G$ using $\hat{G}_{x_1', w_1}$, where 
    $$
    \hat{G}_{x_1', w_1}(x) =  \begin{cases} \hat{G}(x)    & \text{if } x \neq x'\,,\\
    w_1 &\text{if } x=x_1'\,.   \end{cases}
    $$
    $\mathcal{F}_1'$ returns to the challenger the output of $\mathcal{F}_1$.
    $\mathcal{F}_2'$ proceeds analogously.
\end{itemize}

We claim that the success probability of $(\mathcal{P}',\mathcal{F}_1',\mathcal{F}_2')$ is at least $p - \negl(\lambda)$. We leave the full details of the proof to Appendix \ref{app: 1}, as these are lengthy but not particularly enlightening. Here, why provide some intuition about the reduction. The crucial observation is that there is nothing special about the the marked input $y$ in $H_0$ other than the fact that the oracle $G$ maps it to the correct basis choice used by the challenger. What $\mathcal{P}'$ does (without ever seeing $y$) is come up with $x_1'$ and $x_2'$, and then reprogram the oracle so that $x_1'$ and $x_2'$ map to the (correct or incorrect) basis choices he received as part of the $H_1$ game. Crucially, $x_1$ and $x_2$ are not sampled uniformly at random, but from the distribution $X_{\lambda}$ which is statistically indistinguishable from the distribution of marked inputs from which $y$ is sampled (by definition of the challenge distribution $\mathcal{D}_{\mathsf{PF}\mbox{-}\mathsf{pairs}\mbox{-}\mathsf{stat}}$). This makes the distribution of the game played by the invoked adversary $(\mathcal{P}, \mathcal{F}_1, \mathcal{F}_2)$ essentially the same as in the game $H_0$. The only difference between the two games is that they involve slightly different oracles, respectively $\hat{G}$ and $G$. Replacing one with the other can be done without affecting the winning probability by more than a negligible amount as long as $P$ does not query at $x_1'$ and $x_2'$ with non-negligible probability, which is the case.

We point out that the proof of this lemma is the only step in the proof of Theorem \ref{thm: security} in which one uses the indistinguishability condition in the definition of  $\mathcal{D}_{\mathsf{PF}\mbox{-}\mathsf{pairs}\mbox{-}\mathsf{stat}}$ (and of $\mathcal{D}_{\mathsf{PF}\mbox{-}\mathsf{pairs}\mbox{-}\mathsf{comp}}$ for the computational version of the result). This completes the proof of Lemma \ref{lem: new hybrids 1}. We point out that the first step in the proof of the computational version of Lemma \ref{lem: new hybrids 1} is to argue that using a pseudorandom function ($\mathsf{PRF}$) instead of a uniformly random function $\hat{G}$ does not change the success probability of $(\mathcal{P}', \mathcal{F}_1', \mathcal{F}_2')$ more than negligibly, which is straightforward. The rest of the proof for the version of Theorem \ref{thm: security} with $\mathcal{D}_{\mathsf{PF}\mbox{-}\mathsf{pairs}\mbox{-}\mathsf{comp}}$ is analogous to the proof for the version with  $\mathcal{D}_{\mathsf{PF}\mbox{-}\mathsf{pairs}\mbox{-}\mathsf{stat}}$.
\vspace{1mm}
\end{proof}

\begin{proof}[Proof of Lemma \ref{lem: new hybrids 2}]
Suppose for a contradiction that the lemma is false. Let $\mathcal{A} = (\mathcal{P}, \mathcal{F}_1, \mathcal{F}_2)$ be a query-bounded adversary that wins $H_1$ with probability noticeably higher than $H_2$ (the reverse case being similar). We construct a query-bounded adversary $\mathcal{A}'$ that distinguishes samples from the distributions $G(X_{\lambda})$ and $U_{m(\lambda)}$, which cannot exist by Corollary \ref{cor: dist entropy}. We assume here that $\mathcal A'$ has access to two samples from either $G(X_{\lambda})$ or $U_{m(\lambda)}$ (this assumption is justified by the equivalence of single-sample and polynomially many-sample distinguishing tasks for efficiently sampleable distributions). $\mathcal{A}'$ receives as input two samples $\theta_1, \theta_2 \in \{0,1\}^{m(\lambda)}$. $\mathcal{A}'$ then simulates the challenger in a copy-protection game of hybrid $H_1$ with $\mathcal{A}$, in the following way: in the input challenge phase, it samples $(x_1,x_2) \leftarrow D'_y$. For each $i \in \{1,2\}$, if $x_i \neq y$, $\mathcal A'$ sends $\theta_i$ to $\mathcal{F}_i$ (otherwise sends $G(y)$). If $\mathcal A$ wins the game, $\mathcal{A}'$ guesses that the sample was from $G(X_{\lambda})$, otherwise that it was from $U_{m(\lambda)}$. 
\end{proof}

\begin{proof}[Proof of Lemma \ref{lem: new hybrids 3}]
The proof is similar to the previous lemma. Suppose for a contradiction that the lemma is false. Let $\mathcal{A} = (\mathcal{P}, \mathcal{F}_1, \mathcal{F}_2)$ be a query-bounded adversary that wins $H_2$ with probability noticeably higher than $H_3$ (the reverse case being similar). We construct a query-bounded adversary $\mathcal{A'}$ that distinguishes samples from the distributions $G(X_{\lambda})$ and $U_{m(\lambda)}$, which cannot exist by Corollary \ref{cor: dist entropy}. $\mathcal{A}'$ receives as input a challenge $\theta \in \{0,1\}^{m(\lambda)}$. $\mathcal{A}'$ then simulates the challenger in a copy-protection game of hybrid $H_2$ with $\mathcal{A}$ as follows. In the input challenge phase, it samples $(x_1,x_2) \leftarrow D_y'$. For each $i \in \{1,2\}$, if $x_i = y$, $\mathcal A'$ sends $\theta$ to $\mathcal{F}_i$. If $\mathcal A$ wins the game, $\mathcal{A}'$ guesses that the sample was from $G(X_{\lambda})$, otherwise that it was from $U_{m(\lambda)}$. 
\end{proof}

\begin{proof}[Proof of Lemma \ref{lem: 5}] Given an adversary $\mathcal{A}$ that wins with probability $p$ in $H_4$, the adversary $\mathcal{A}'$ for $H_5$ which acts identically to $\mathcal{A}$ (i.e. $\mathcal{P}$ ignores the additional inputs $\theta_1$ and $\theta_2$ received at the start, and simply forwards them to $\mathcal{F}_1$ and $\mathcal{F}_2$ respectively) clearly wins with probability $p$.

\end{proof}

We move to the crux of the proof, Lemma \ref{lem: 4}. We break down the proof of Lemma \ref{lem: 4} into a few technical lemmas.

\begin{lem}
\begin{align}
\Big|\Pr \Big[(\mathcal{P}, \mathcal{F}_1, \mathcal{F}_2)  &\text{ win } H_4 \,|\, x_1 \text{ is a 0-input and } x_2 \text{ is a 0-input} \Big] \nonumber\\&- \Pr\Big[(\mathcal{P}, \mathcal{F}_1, \mathcal{F}_2) \text{ win } H_3 \,|\, x_1 \text{ is a 0-input and } x_2 \text{ is a 0-input}\Big] \Big| = \negl(\lambda)\label{eq: lem 3}
\end{align}

\end{lem}
\begin{proof}
Notice that the task of distinguishing $H_3$ and $H_4$ is precisely amenable to the one-way-to-hiding lemma (in its form of Lemma \ref{lem: qrom technical step}). By an application Lemma \ref{lem: qrom technical step}, it is sufficient to show that the probability that an adversary, with access to all of the information received in $H_4$ (including the input phase challenges, which in this case are independent of the correct basis choice), queries the oracle at $v$ is negligible. The latter is true, since an adversary violating this straightforwardly gives an adversary which contradicts Lemma \ref{lem: monogamy}. 
\end{proof}
The next lemma is the crucial step in the proof of Lemma \ref{lem: 4}.

\begin{lem}
\label{lem: 8}
For any bounded strategy $(\mathcal{P}, \mathcal{F}_1, \mathcal{F}_2)$, there exists a negligible function $\mu$ such that, for all $\lambda$, there exists an $i \in \{1,2\}$ such that: 
\begin{align}\Big|\Pr&\big[\mathcal{F}_i \text{ returns 1 in } H_4 \,|\, x_i \text{ is a 1-input } \big] \nonumber \\
& -\Pr\big[\mathcal{F}_i \text{ returns 1 in } H_3 \,|\, x_i \text{ is a 1-input } \big] \Big| < 1-\epsilon^*  + \mu(\lambda)\,, \label{eq: lem 44}
\end{align}
where we can take $\epsilon^* = 10^{-4}$.

\end{lem}
The proof of Lemma \ref{lem: 8} is fairly involved, and constitutes the main technical work in this paper. We break the proof down into two parts: 
\begin{itemize}
    \item First, we show that an adversary violating Lemma \ref{lem: 8} can be used to construct an adversary that wins at a certain $r$-fold parallel repetition of an appropriate distinguishing game (where one should think of $r$ as a small constant), which we call $\mathsf{G}_r$ (Lemma \ref{lem: reduction to Gr}).
    \item Next, we prove a technical lemma (Lemma \ref{lem: main bound}, a generalization of a similar lemma in \cite{broadbent2019uncloneable}), which allows us to upper bound the probability of a query-bounded adversary winnning in $\mathsf{G}_r$ (Lemma \ref{lem: Gr to monogamy}). Such an upper bound relies on Lemma \ref{lem: main bound} as well as on properties of monogamy of entanglement games.
\end{itemize}
We discuss informally below, after the description of the game $\mathsf{G}_r$, the reason why we consider the $r$-fold parallel repetition of a natural distinguishing game.

Let $\epsilon >0$. Suppose for a contradiction that there exists a strategy $(\mathcal{P},\mathcal{F}_1, \mathcal{F}_2)$ such that, for all negligible functions $\mu$, there exists some $\lambda$ such that, for all $i \in \{1,2\}$,
\begin{align}\Big|\Pr&[\mathcal{F}_i \text{ returns 1 in } H_4 | x_i \text{ is a 1-input } ] \nonumber \\
& -\Pr[\mathcal{F}_i \text{ returns 1 in } H_3 | x_i \text{ is a 1-input } ] \Big| \geq 1-\epsilon + \mu(\lambda) \,. \label{eq: lem -7}
\end{align}

Notice, in particular, that the above straightforwardly implies that, for all negligible functions $\mu$, there exist \emph{infinitely many} $\lambda$'s such that, for all $i \in \{1,2\}$,
\begin{align}\Big|\Pr&[\mathcal{F}_i \text{ returns 1 in } H_4 | x_i \text{ is a 1-input } ] \nonumber \\
& -\Pr[\mathcal{F}_i \text{ returns 1 in } H_3 | x_i \text{ is a 1-input } ] \Big| \geq 1-\epsilon + \mu(\lambda) \,. \label{eq: lem -3}
\end{align}

Let $r \in \mathbb{N}$. Given such $(\mathcal{P},\mathcal{F}_1,\mathcal{F}_2)$, we will construct an adversary $(\mathcal{P}', \mathcal{F}_1', \mathcal{F}_2')$ that succeeds at the following game $\mathsf{G}_r$ between a challenger and an adversary $(\mathcal{P}', \mathcal{F}_1', \mathcal{F}_2')$ with probability at least $1 - g(\epsilon)$, for some continuous, monotonically increasing function $g(\epsilon)$ such that $g(0)=0$. Given a function $H$, we denote by $H_{x,y}$ the function such that $H_{x,y}(x') = H (x')$ for all $x' \neq x$, and $H_{x,y}(x) = y$. We describe game $\mathsf{G}_r$ as follows:

\begin{itemize}
    \item[(i)] The challenger samples $\theta, v \leftarrow \{0,1\}^{m(\lambda)}$, $z' \leftarrow \{0,1\}^{\lambda}$, and $w \leftarrow \{0,1\}^r$ as well as a random oracle $H: \{0,1\}^{m(\lambda)} \rightarrow \{0,1\}^{\lambda}$. The challenger then sends $(\ket{v^{\theta}}, z)$ with $z=H(v)$ to $\mathcal{P}'$ who gets oracle access to $H_1,..,H_r$, where
    $$H_i = \begin{cases} H    & \text{if } w_i=0\,,\\
    H_{v,z'} &\text{if } w_i=1 \,.   \end{cases}
    $$ 
    \item[(ii)] $\mathcal{P}'$ creates a bipartite state $\sigma$, and sends its two subsystems to $\mathcal{F}'_1$ and $\mathcal{F}'_2$, respectively.
    \item[(iii)] The challenger sends $\theta$ to both $\mathcal{F}'_1$ and $\mathcal{F}'_2$.
    \item[(iv)] $\mathcal{F}'_1$ and $\mathcal{F}'_2$ (who cannot communicate) return $w'$ and $w''$ in $\{0,1\}^r$ respectively to the challenger.
\end{itemize}

$(\mathcal{P}', \mathcal{F}'_1, \mathcal{F}'_2)$ win if $w' = w'' = w$.
\vspace{2mm}

It might seem mysterious why we consider such $r$-fold parallel repetition of the more natural distinguishing game with $r=1$. The reason will become more apparent later in the proof. For now, informally the reason is the following: in order to complete the proof of security, we will need to employ a decision version of a one-way-to-hiding lemma for entangled parties. This does not straightforwardly follow from the search version in \cite{broadbent2019uncloneable}, because of a factor of $9$ loss in the security compared to the non-entangled version the lemma. The parallel repetition allows to overcome the security loss, and to obtain a non-trivial security statement.

\vspace{2mm}
The reduction we use to construct $(\mathcal{P}', \mathcal{F}_1', \mathcal{F}_2')$ winning at game $\mathsf{G}_r$ uses the so-called ``Gentle Measurement Lemma'' or ``Almost As Good As New Lemma'' \cite{Winter99,Aaronson_2005}: if a measurement succeeds with high probability, then the post-measurement state conditioned on success, 
is close to the initial state. Specifically, we need the following lemma, which is a consequence of the aforementioned lemmas.

\begin{lem}\label{lem:multi-gentle}
Let $r \in \mathbb{N}$, let $\ket\psi$ to be a unit vector, and let $\{\Pi_i\}_{i \in [r]}$ be such that for every $i \in [r]$: $$\|\Pi_i\ket\psi\|_2^2\ge1-\epsilon.$$ 
Then, it follows that:
\begin{equation*}
    \|\Pi_r\Pi_{r-1}...\Pi_1\ket\psi\|_2^2\ge 1- 2r\sqrt{\epsilon}.
\end{equation*}
\end{lem}
\begin{proof}
Let $\ket{\delta_i}=(\mathds 1-\Pi_i)\ket \psi.$
We prove by induction that
\begin{equation*}
    \Pi_{\ell}\Pi_{\ell-1}...\Pi_1\ket\psi=\ket\psi+\ket{\eta_\ell},
\end{equation*}
for some vector $\ket{\eta_\ell}$ such that $\|\ket{\eta_\ell}\|_2\le \ell\sqrt\epsilon$. The statement clearly holds for $\ell=0$. Assuming it holds up to $\ell-1$, we get
\begin{align*}
    \Pi_{\ell}\Pi_{\ell-1}...\Pi_1\ket\psi&=\Pi_\ell\left(\ket\psi+\ket{\eta_{\ell-1}}\right)\\
    &=\ket\psi-\ket{\delta_\ell}+\Pi_\ell\ket{\eta_{\ell-1}}\\
    &\eqqcolon \ket \psi+\ket{\eta_{\ell}},
\end{align*}
where $\|\ket{\eta_\ell}\|_2\le \ell\sqrt\epsilon$ by the triangle inequality, the fact that projectors have unit operator norm, and the inductive hypothesis. One more application of the triangle inequality shows that:
\begin{align*}
    \|\Pi_r\Pi_{r-1}...\Pi_1\ket\psi\|_2^2&= \|\ket \psi+\ket{\eta_{r}}\|_2^2\\
    &\ge \left(1-r\sqrt\epsilon\right)^2\\
    &\ge 1-2r\sqrt \varepsilon.
\end{align*}
\end{proof}

We apply the above lemma to construct a strategy that wins at $\mathsf{G}_r$.

\begin{lem}
\label{lem: reduction to Gr}
Suppose $(\mathcal{P}, \mathcal{F}_1, \mathcal{F}_2)$ satisfies \eqref{eq: lem -3} for some $\epsilon>0$. Then, there exists $(\mathcal{P}', \mathcal{F}_1', \mathcal{F}_2')$ such that, for any negligible function $\nu$, there exist infinitely many $\lambda$, such that $(\mathcal{P}', \mathcal{F}_1', \mathcal{F}_2')$ wins game $\mathsf{G}_r$ with security parameter $\lambda$ with probability at least $1-4r\sqrt\epsilon +\nu(\lambda)$. 
\end{lem}

\begin{proof}
Note that we can equivalently recast the game in hybrid $H_3$ as a game in which $\mathcal{P}$ receives a uniformly random string $z$ (instead of $H(v)$), but then the oracle is reprogrammed at $v$, so that $\mathcal{P}$ has oracle access to $H_{v,z}$ instead of $H$, where $H_{v,z}(x) = H(x)$ if $x \neq v$, and $H(v) = z$.

Let $q$ be the number of queries $\mathcal{P}$ makes. Without loss of generality, $\mathcal{P}$'s strategy is specified by a unitary $U$ and takes the form $(UO^H)^q$.

Notice that 
\begin{equation}
\label{eq: lemma 8}
    \mathbb{E}_H \mathbb{E}_v \mathbb{E}_{\theta} \mathbb{E}_{z} \mathbb{E}_{k \in [q]} \left\| \ket{v}\bra{v} (U O^{H_{v,z}})^k \ket{v^{\theta}} \otimes \ket{z} \right\|^2  = \negl(\lambda)\,.
\end{equation}

Suppose for a contradiction that the latter was not the case, then the pirate could recover the classical string $v$ with non-negligible probability and send $v$ to both $\mathcal{F}_1$ and $\mathcal{F}_2$. This would give a strategy that wins the monogamy game (more precisely the variant of Lemma \ref{lem: monogamy}).

By an application of a suitable variation of the one-way-to-hiding lemma (Lemma \ref{lem: qrom technical step}), the global state after $\mathcal{P}$'s action, including the challenger's registers, is negligibly close in hybrids $H_3$ and $H_4$, i.e. there exists a negligible function $\nu'$, such that 
\begin{align}& \Big\| \mathbb{E}_H \mathbb{E}_v \mathbb{E}_{\theta} \mathbb{E}_z \ket{H}\bra{H} \otimes \ket{v}\bra{v} \otimes \ket{\theta}\bra{\theta} \otimes \left((U O^{H_{v,z}})^q \ket{v^{\theta}} \bra{v^{\theta}} \otimes \ket{z}\bra{z}\left((U O^{H_{v,z}})^q\right)^{\dagger} \right)\nonumber \\ \, &-\mathbb{E}_H \mathbb{E}_v \mathbb{E}_{\theta} \mathbb{E}_{z} \ket{H}\bra{H} \otimes \ket{v}\bra{v} \otimes \ket{\theta}\bra{\theta} \otimes \left((U O^{H})^q \ket{v^{\theta}} \bra{v^{\theta}} \otimes \ket{z}\bra{z}\left((U O^{H})^q\right)^{\dagger}\right) \Big\|_1 \nonumber\\
&= \nu'(\lambda)\,. \label{eq: 29 crucial}
\end{align}

Now, by hypothesis, for any negligible function $\mu$, for all $i \in \{1,2\}$ there exist infinitely many $\lambda$ such that $(\mathcal{P}, \mathcal{F}_1, \mathcal{F}_2)$ satisfies expression \eqref{eq: lem -3}. In particular, we will use this hypothesis for a choice of $\mu$ sufficiently large (with respect to the negligible functions $\nu$ and $\nu'$).

We will also assume that, in \eqref{eq: lem -3}, it is 
$$ \Pr[\mathcal{F}_i \text{ returns 1 in } H_4 | x_i \text{ is a 1-input } ] > \Pr[\mathcal{F}_i \text{ returns 1 in } H_3 | x_i \text{ is a 1-input } ] \,,
$$
for infinitely many such $\lambda$'s (the reverse case being similar). Let $\Lambda_{\mathsf{good}}$ be the corresponding set of such $\lambda$'s.

We will construct an adversary $(\mathcal{P}', \mathcal{F}_1', \mathcal{F}_2')$ for the game $\mathsf{G}_r$ with security parameter $\lambda \in \Lambda_{\mathsf{good}}$, which wins with probability at least $1-4r\sqrt{\epsilon} + \nu(\lambda)$. The adversary $(\mathcal{P}', \mathcal{F}_1', \mathcal{F}_2')$ is as follows:
\begin{itemize}
    \item $\mathcal{P}'$ receives $\rho$ and $z$ from the challenger and runs $\mathcal{P}$ on input $\rho$ and $z$ and using oracle $H_1$ to answer any query by $\mathcal{P}$ (although any other $H_i$ would be fine). Let $\tilde{\rho}$ (a bipartite state) be the output. $\mathcal{P}'$ sends the first subsystem to $\mathcal{F}_1'$ and the second subsystem to $\mathcal{F}_2'$.
    \item $\mathcal{F}_1'$ and $\mathcal{F}_2'$ receive $\theta$ from the challenger. For $i \in [r]$, they do the following:
    \begin{itemize}
        \item Let $b \in \{1,2\}$, $\mathcal{F}_b'$ runs $\mathcal{F}_b$ on input the system received from $\mathcal{P}'$ and $\theta$, using oracle $H_i$. Let $w_i^{(b)}$ be the outcome. $\mathcal{F}_b'$ runs the inverse of the computation run by $\mathcal{F}_b$ before the measurement.
    \end{itemize}
    \item $\mathcal{F}_b'$ returns the string $(w^{(b)}_1, \ldots, w^{(b)}_r)$ to the challenger.
\end{itemize}
For $i \in [r]$, let $U_i^b$ be the unitary implemented by $\mathcal{F}_b$ when run with oracle $H_i$, and let $\Pi^b_i$ be the projector corresponding to the correct outcome in the final measurement. Define the projectors $\widetilde{\Pi}^b_i=\left(U_i^b\right)^\dagger \Pi^b_iU_i^b$. Crucially, Equation \eqref{eq: 29 crucial}, together with the hypothesis of Lemma \ref{lem: reduction to Gr} (used with a sufficiently large negligible function $\mu$), imply the following: 
\begin{itemize}
    \item If $\ket{\psi}$ is a purification of the state $\tilde{\rho}$ returned by $\mathcal{P}$ to $\mathcal{P}'$, then, for all $b \in \{1,2\}$, $i \in [r]$:
    \begin{equation}
        \| \widetilde{\Pi}^b_i \ket{\psi}\|^2 \geq 1 - \epsilon + \mu'(\lambda) \,
    \end{equation}
\end{itemize}
for a sufficiently large negligible function $\mu'$, and for infinitely many $\lambda$.

Applying the variant of the ``Gentle Measurement Lemma'' from Lemma \ref{lem:multi-gentle} implies that $\mathcal{F}_b'$ succeeds with probability at least $1-2r \sqrt\epsilon + \frac12 \nu(\lambda)$. By a union bound, we conclude that  $(\mathcal{P}', \mathcal{F}_1', \mathcal{F}_2')$ win with probability at least $1-4r\sqrt\epsilon + \nu(\lambda)$, for infinitely many $\lambda$, as desired.
\end{proof}

Next, our goal is to upper bound the probability that a query-bounded adversary $(\mathcal{P},\mathcal{F}_1,\mathcal{F}_2)$ wins $\mathsf{G}_r$ by the probability that $\mathcal{F}_1$ and $\mathcal{F}_2$ simultaneously query the oracle at point $v$. This is captured by the following technical lemma, which we specialize to our exact case in the subsequent corollary. 
Before stating the lemma, we introduce some notation. Let $F: \mathcal{X} \rightarrow \mathcal{Y}$. Let $\underline{v} \in \mathcal{X}^r$ be such that all $\underline{v}_i$ are distinct, for $i \in [r]$, and let $\underline{z} \in \mathcal{Y}^r$. We denote by $F_{\underline{v}, \underline{z}}$ the function,
$$F_{\underline{v}, \underline{z}}(x)= \begin{cases} F(x)    & \text{if } x \neq \underline{v}_i \text{ for all } i \,,\\
    \underline{z}_i &\text{if } x=\underline{v}_i \,.   \end{cases}$$

\begin{lem}
\label{lem: main bound}
Let $\mathcal{X}, \mathcal Y, \mathcal R$ be sets of binary strings. Let $r \in \mathbb{N}$. Let $F$ be a function-valued random variable, where each function is from $\mathcal{X}$ to $\mathcal{Y}$. Let $\mathcal Z \subseteq \mathcal Y^r$.
For any $\underline{v} \in \mathcal{X}^r$, any injective function $g:\mathcal Z\to\mathcal R$, any distribution $Z$ on $\mathcal{Z}$, any unitaries $U_\mathsf{B}, U_\mathsf{C}$, any bipartite state $\ket{\psi}$, any $q_\mathsf{B}, q_\mathsf{C} \in \mathbb{N}$, the following holds:
\begin{align}
    \mathbb{E}_F \mathbb{E}_{\underline{z} \leftarrow Z} \Big\| \Pi^{g(\underline{z})} \left(U_\mathsf{B} O_\mathsf{B}^{F_{\underline{v}, \underline{z}}}\right)^{q_\mathsf{B}} \otimes &\left(U_\mathsf{C} O_\mathsf{C}^{F_{\underline{v},\underline{z}}}\right)^{q_\mathsf{C}} \ket{\psi}\Big\|^2 \\ &\leq 9p_{\max{}} + \poly(q_\mathsf{B}, q_\mathsf{C}) \sqrt{M}\,,
\end{align} 
where $\Pi^{w} = \Pi_\mathsf{B}^w \otimes \Pi_\mathsf{C}^w$,
$
\,\,p_{\max{}}=\max_{r\in\mathcal R}\Pr[g(\underline{z})=r: \underline{z} \leftarrow Z]
$, and 
$$M = \mathbb{E}_k \mathbb{E}_l \mathbb{E}_F \mathbb{E}_{\underline{z} \leftarrow \mathcal{Z}} \Big\| ( P_{\underline{v}} \otimes P_{\underline{v}}) \left(U_\mathsf{B} O^{F_{\underline{v},\underline{z}}}_\mathsf{B})^{k} \otimes (U_\mathsf{C} O^{F_{\underline{v},\underline{z}}}_\mathsf{C})^{l}\right) \ket{\psi}\Big\|^2 \,,$$
where $P_{\underline{v}} = \sum_{i=1}^r \ket{\underline{v}_i}\bra{\underline{v}_i}$.
\end{lem}

\begin{proof}
The following proof is similar to the proof of Lemma 21 in \cite{broadbent2019uncloneable}, but extended to a slightly more general setting.

For $\mathsf{L} \in \{\mathsf{B},\mathsf{C}\}$, let 
$V_\mathsf{L}^{F} = \left(U_\mathsf{L} O_\mathsf{L}^{F}(\Id-P_{\underline{v}})\right)^{q_\mathsf{L}}$ and $W_\mathsf{L}^{F} = (U_\mathsf{L} O^{F}_\mathsf{L})^{q_\mathsf{L}} - V_\mathsf{L}^{F}$.

Fix $\underline{v} \in \mathcal{X}^r$, a set $\mathcal{Z} \subseteq \mathcal{Y}^r$ and $\underline{z} \in \mathcal{Z}$. We have the following.
\begin{align}
    \| \Pi^{g(\underline{z}))} &\left(\left(U_\mathsf{B} O_\mathsf{B}^{F_{\underline{v}, \underline{z}}}\right)^{q_\mathsf{B}} \otimes \left(U_\mathsf{C} O_\mathsf{C}^{F_{\underline{v}, \underline{z}}}\right)^{q_\mathsf{C}} \right) \ket{\psi}\|^2  \nonumber \\
    &= \| \Pi^{g(\underline{z})} \left(\left(U_\mathsf{B} O_\mathsf{B}^{F_{\underline{v}, \underline{z}}}\right)^{q_\mathsf{B}} \otimes V_\mathsf{C}^{F_{\underline{v}, \underline{z}}} + V_\mathsf{B}^{F_{\underline{v}, \underline{z}}} \otimes W_\mathsf{C}^{F_{\underline{v}, \underline{z}}} + W_\mathsf{B}^{F_{\underline{v}, \underline{z}}} \otimes W_\mathsf{C}^{F_{\underline{v}, \underline{z}}} \right) \ket{\psi}\|^2 \\
    &\leq \| \Pi^{g(\underline{z})} \left(\left(U_\mathsf{B} O_\mathsf{B}^{F_{\underline{v}, \underline{z}}}\right)^{q_\mathsf{B}} \otimes  V_\mathsf{C}^{F_{\underline{v}, \underline{z}}} + V_\mathsf{B}^{F_{\underline{v}, \underline{z}}} \otimes W_\mathsf{C}^{F_{\underline{v}, \underline{z}}}\right) \ket{\psi}\|^2 \nonumber \\
    &+ (3q_\mathsf{B} q_\mathsf{C} +2)q_\mathsf{B} q_\mathsf{C} \,\mathbb{E}_k \mathbb{E}_l \left\| P_{\underline{v}} \otimes P_{\underline{v}} \left(\left( U_\mathsf{B} O_\mathsf{B}^{F_{\underline{v}, \underline{z}}} \right)^k \otimes \left( U_\mathsf{C} O_\mathsf{C}^{F_{\underline{v}, \underline{z}}}\right)^l \right) \ket{\psi} \right\| \,, \label{eq: 9}
\end{align} 
where the inequality follows from a triangle inequality together with a bound used in \cite{broadbent2019uncloneable} (more precisely, Lemma 18 in \cite{broadbent2019uncloneable}). Applying Jensen's inequality and the definition of $M$, we obtain
\begin{align}
\mathbb{E}_F \mathbb{E}_{\underline{z} \leftarrow \mathcal{Z}}\| &\Pi^{g(\underline{z})} \left(\left(U_\mathsf{B} O_\mathsf{B}^{F_{\underline{v}, \underline{z}}}\right)^{q_\mathsf{B}} \otimes \left(U_\mathsf{C} O_\mathsf{C}^{F_{\underline{v}, \underline{z}}}\right)^{q_\mathsf{C}} \right) \ket{\psi}\|^2 \nonumber\\
 &\leq \mathbb{E}_{F} \mathbb{E}_{\underline{z} \leftarrow \mathcal{Z}}\| \Pi^{g(\underline{z})} \left(\left(U_\mathsf{B} O_\mathsf{B}^{F_{\underline{v}, \underline{z}}}\right)^{q_\mathsf{B}} \otimes  V_\mathsf{C}^{F_{\underline{v}, \underline{z}}} + V_\mathsf{B}^{F_{\underline{v}, \underline{z}}} \otimes W_\mathsf{C}^{F_{\underline{v}, \underline{z}}}\right) \ket{\psi}\|^2\\
 &\quad \quad+ (3q_\mathsf{B} q_\mathsf{C} + 2) q_\mathsf{B} q_\mathsf{C} \sqrt{M}  
\end{align}

We will show that, for any fixed $F$,
$$\mathbb{E}_{\underline{z} \leftarrow \mathcal{Z}} \| \Pi^{g(\underline{z})} \left(\left(U_\mathsf{B} O_\mathsf{B}^{F_{\underline{v}, \underline{z}}}\right)^{q_\mathsf{B}} \otimes  V_\mathsf{C}^{F_{\underline{v}, \underline{z}}} + V_\mathsf{B}^{F_{\underline{v}, \underline{z}}} \otimes W_\mathsf{C}^{F_{\underline{v}, \underline{z}}}\right) \ket{\psi}\|^2 \leq 9p_{\max{}}\,.$$
Let 
$$ \alpha = \mathbb{E}_{\underline{z} \leftarrow \mathcal{Z}} \| \Pi^{g(\underline{z})} \left(\left(U_\mathsf{B} O_\mathsf{B}^{F_{\underline{v}, \underline{z}}}\right)^{q_\mathsf{B}} \otimes V_\mathsf{C}^{F_{\underline{v}, \underline{z}}} \right)\ket{\psi}\|^2 \,,$$
and 
$$ \beta = \mathbb{E}_{\underline{z} \leftarrow \mathcal{Z}} \| \Pi^{g(\underline{z})} \left( V_\mathsf{B}^{F_{\underline{v}, \underline{z}}} \otimes W_\mathsf{C}^{F_{\underline{v}, \underline{z}}} \right)\ket{\psi}\|^2 \,.$$

Applying triangle inequalities, and using that the $\Pi^{g(\underline{z})}$ are orthogonal projectors, we get 
\begin{equation}
\mathbb{E}_{\underline{z} \leftarrow \mathcal{Z}} \| \Pi^{g(\underline{z})} \left(\left(U_\mathsf{B} O_\mathsf{B}^{F_{\underline{v}, \underline{z}}}\right)^{q_\mathsf{B}} \otimes  V_\mathsf{C}^{F_{\underline{v}, \underline{z}}} + V_\mathsf{B}^{F_{\underline{v}, \underline{z}}} \otimes W_C^{F_{\underline{v}, \underline{z}}}\right) \ket{\psi}\|^2  \leq \alpha + \beta + 2 \sqrt{\alpha\beta}
\end{equation}
Now, notice that $V_\mathsf{B}^{F_{\underline{v}, \underline{z}}}$ and $V_\mathsf{C}^{F_{\underline{v}, \underline{z}}}$ do not depend on $\underline{z}$, since they always project on the support of $\Id-P_{\underline{v}}$. Using standard properties of the operator norm, we get
\begin{align}
\alpha &= \mathbb{E}_{\underline{z} \leftarrow \mathcal{Z}} \| \Pi^{g(\underline{z})} \left(\left(U_\mathsf{B} O_\mathsf{B}^{F_{\underline{v}, \underline{z}}}\right)^{q_\mathsf{B}} \otimes V_\mathsf{C}^{F_{\underline{v}, \underline{z}}} \right)\ket{\psi}\|^2  \\
& \leq \mathbb{E}_{\underline{z} \leftarrow \mathcal{Z}} \| (\mathds{1}_\mathsf{B} \otimes \Pi_\mathsf{C}^{g(\underline{z})} ) \left(\mathds{1}_\mathsf{B} \otimes V_\mathsf{C}^{F_{\underline{v}, \underline{z}}} \right) \ket{\psi} \|^2  \\
&\leq \bra\psi\mathds{1}_\mathsf{B} \otimes\left( \left(V_\mathsf{C}^{F_{\underline{v}, \underline{z}}}\right)^\dagger\left(\mathbb{E}_{z \leftarrow \mathcal{Z}}\Pi_C^{g(\underline{z})}\right) V_\mathsf{C}^{F_{\underline{v}, \underline{z}}} \right) \ket{\psi}  \\
&\leq p_{\max{}},
\end{align}
where the last line follows because $\mathbb{E}_{\underline{z} \leftarrow \mathcal{Z}}\Pi_\mathsf{C}^{g(\underline{z})}\le p_{\max{}}\mathds{1}$.

Similarly, one obtains $\beta \leq 4 p_{\max{}}$, where the factor of $4$ is due to the fact that $W_\mathsf{C}^{F_{\underline{v}, \underline{z}}}$ is not unitary in general, and so we upper bound it by using

$$ \| W_\mathsf{C}^{F_{\underline{v}, \underline{z}}}\|_{\infty} \leq \|\left( U_\mathsf{C} O_\mathsf{C}^{F_{\underline{v}, \underline{z}}}\right)^{q_\mathsf{C}} \|_\infty + \| V_\mathsf{C}^{F_{\underline{v}, \underline{z}}} \|_\infty  \leq 2 \,. $$ 

It follows that $\alpha + \beta + 2\sqrt{\alpha \beta} \leq 9 p_{\max{}}$, as desired. 
\end{proof}

We specialize Lemma \ref{lem: main bound} to our exact case in the following corollary. Before stating the corollary, let's fix some notation. For functions $H_1,\ldots, H_r$, we define a ``combined'' oracle unitary $O^{H_1,\ldots, H_r}$, acting on a ``control'' register, a ``query'' register, and an auxiliary register, as follows:
$$ O^{H_1,\ldots, H_r} \ket{i}\ket{x}\ket{0} = \ket{i} O^{H_i}(\ket{x}\ket{0})\,.$$
\begin{cor}
\label{cor: main bound specialized}
Let $r \in \mathbb{N}$. Let $v \in \{0,1\}^{m(\lambda)}$ and $z, z'\in \{0,1\}^{\lambda}$.
For all unitaries $U_\mathsf{B}, U_\mathsf{C}$, for all bipartite states $\ket{\psi}$, the following holds:
\begin{align}
    \mathbb{E}_H \mathbb{E}_{w \leftarrow \{0,1\}^r} \| \Pi^w \left(U_\mathsf{B} O_\mathsf{B}^{H_{v,z}^{w_1},..,H_{v,z}^{w_r}}\right)^{q_\mathsf{B}} \otimes &\left(U_\mathsf{C} O_\mathsf{C}^{H_{v,z}^{w_1},..,H_{v,z}^{w_r}}\right)^{q_\mathsf{C}} \ket{\psi}\|^2 \\ &\leq \frac{9}{2^r} + \poly(q_\mathsf{B}, q_\mathsf{C}) \sqrt{M},\,
\end{align} 
where $\Pi^{w} = \Pi_\mathsf{B}^w \otimes \Pi_\mathsf{C}^w$, $H: \{0,1\}^{m(\lambda)} \rightarrow \{0,1\}^{\lambda}$ and, for $b \in \{0,1\}$,
$$H_{v,z}^b := \begin{cases} H_{v,z}   & \text{if } b=0\,,\\
    H_{v,z'} &\text{if } b=1 \,,   \end{cases}  $$
and
$$M = \mathbb{E}_k \mathbb{E}_l \mathbb{E}_H \mathbb{E}_w \| \ket{v}\bra{v} \otimes \ket{v}\bra{v} (U_\mathsf{B} O^{H_{v,z}^{w_1},..,H_{v,z}^{w_r}}_\mathsf{B})^{k} \otimes (U_\mathsf{C} O^{H_{v,z}^{w_1},..,H_{v,z}^{w_r}}_\mathsf{C})^{l} \ket{\psi}\|^2\,.$$
\end{cor}

\begin{proof}
We apply Lemma \ref{lem: main bound} with
\begin{itemize}
    \item $\mathcal{R} = \{0,1\}^r$, $\,\mathcal{X} = [r] \times \{0,1\}^{m(\lambda)}$, $\,\mathcal{Y} = \{0,1\}^{\lambda}$.
    \item $\mathcal{Z} = \{(z_1,\ldots, z_r): z_i \in \{z,z'\}\}$, and $Z$ the uniform distribution over $\mathcal Z$,
    \item $g: \mathcal{Z} \rightarrow \mathcal{R}$ such that $g(z_1,\ldots,z_r) = w$ where $w_i = 0$ if $z_i = z$ and $w_i = 1$ if $z_i = z'$.
    \item $F: \mathcal{X} \rightarrow \mathcal{Y}$ is such that $F(i,x) = (i, H(x))$.
    \item $\underline{v} = \big((1,v),\ldots, (r,v) \big)$.
\end{itemize}
\end{proof}

\begin{lem}
\label{lem: Gr to monogamy}
For all query-bounded adversaries $\mathcal A$ there exists a negligible function $\mu$ such that $\mathcal A$ wins game $\mathsf{G}_r$ with probability at most $ \frac{9}{2^{r}} + \mu(\lambda)$.
\end{lem}

\begin{proof}
Suppose for a contradiction that there existed an adversary $(\mathcal{P}, \mathcal{F}_1, \mathcal{F}_2)$, and a polynomial $p>0$, such that $(\mathcal{P}, \mathcal{F}_1, \mathcal{F}_2)$ wins game $\mathsf{G}_r$ with probability greater than $\frac{9}{2^{r}} + 1/p(\lambda)$ for infinitely many $\lambda$'s. We show that this implies an adversary that wins at the following variant of the monogamy game of Section \ref{sec: monogamy}. Let $\lambda, \lambda' \in \mathbb{N}$. For clarity, we will denote the adversary in this monogamy game as the triple $(\mathcal{P}'$, $\mathcal{F}_1', \mathcal{F}_2')$.
\begin{itemize}
    \item The challenger samples uniformly $\theta, v \leftarrow \{0,1\}^{\lambda}$ as in the original game. The challenger picks a uniformly random function $H: \{0,1\}^{\lambda} \rightarrow \{0,1\}^{\lambda'}$ and sends $H(v)$ to $\mathcal{P}'$. 
    \item Additionally, the challenger samples uniformly $w \leftarrow \{0,1\}^r$, and $z \leftarrow \{0,1\}^{\lambda'}$. $\mathcal{P}$, $\mathcal{F}_1$ and $\mathcal{F}_2$ get oracle access to $H_{v,z}^{w_i}$, for $i \in [r]$, where $$ H_{v,z}^b := \begin{cases} H    & \text{if } b=0\,,\\
    H_{v,z} &\text{if } b=1 \,.   \end{cases} $$
    \item The rest proceeds as in the original monogamy game. 
\end{itemize}

We claim that as long as $\lambda' <\frac15 \lambda$, any adversary wins with at most negligible probability in the game. The reason is that as long as $\lambda'<\frac15 \lambda$, we have $\Hmin(V| \mathsf{E}) > \frac45 \lambda$, where $V$ is the random variable for the string $v$, and $\mathsf{E}$ represents all classical and quantum information that the adversary gets. Thus, Corollary \ref{cor: monogamy entropy} applies.

Next, we are ready to construct an adversary $(\mathcal{P}', \mathcal{F}_1', \mathcal{F}_2')$ for the above monogamy game from an adversary $(\mathcal{P}, \mathcal{F}_1, \mathcal{F}_2)$ winning game $\mathsf{G}_r$ with probability $>\frac{9}{2^r} + 1/p(\lambda)$ for some $\lambda$. $(\mathcal{P}', \mathcal{F}_1', \mathcal{F}_2')$ is as follows:
\begin{itemize}
    \item $\mathcal{P}'$ runs $\mathcal{P}$ on the state received from the challenger. The output is a state on two registers: $\mathcal{P}'$ sends the first half to $\mathcal{F}'_1$ and the second half to $\mathcal{F}'_2$.
    \item Let $q_1$ and $q_2$ be the number of oracle queries performed by the $\mathcal{F}_1$ and $\mathcal{F}_2$ algorithms. $\mathcal{F}_1'$ and $\mathcal{F}_2'$ respectively pick uniformly $k \leftarrow [q_1]$ and $l \leftarrow [q_2]$. Then, $\mathcal{F}_1'$ and $\mathcal{F}_2'$ respectively run the $\mathcal{F}_1$ and $\mathcal{F}_2$ algorithms for $k$ and $l$ queries, using oracle access to $H^{w_i}_{v,z}$, for $i \in [r]$. Finally, $\mathcal{F}'_1$ and $\mathcal{F}'_2$ measure the respective oracle registers and return their outcomes to the challenger.
\end{itemize}

By Corollary \ref{cor: main bound specialized}, the outcome returned by $\mathcal{F}_1'$ and $\mathcal{F}_2'$ is $v$ with inverse-polynomial probability. 

\end{proof}

\begin{proof}[Proof of Lemma \ref{lem: 8}]
Combining lemmas \ref{lem: reduction to Gr} and \ref{lem: Gr to monogamy}, we deduce that, for any $r \geq 4$, the statement of Lemma \ref{lem: 8} must hold for any $\epsilon^*$ such that
$$ 1-4r\sqrt{ \epsilon^*}  > \frac{10}{2^r} \,.$$

Taking $r=5$, we deduce that the statement of Lemma \ref{lem: 8} holds for any $\epsilon^* \leq 10^{-4}$.
\end{proof}

\begin{proof}[Proof of Lemma \ref{lem: 4}]
Let $(\mathcal{P}, \mathcal{F}_1, \mathcal{F}_2)$ be a strategy that wins with probability $p$ in $H_3$. We will argue that $(\mathcal{P}, \mathcal{F}_1, \mathcal{F}_2)$ wins with probability at least $p - \frac{2}{3} + \delta^* - \negl(\lambda)$ in $H_4$, where $\delta^*>0$. When considering the winning probability of a strategy, we can divide the analysis into three cases bases on what type of inputs the challenger is challenging $\mathcal{F}_1$ and $\mathcal{F}_2$ on. Let $\theta \in \{0,1\}^{m(\lambda)}$ denote the basis choice used to encode a string $v \in \{0,1\}^{m(\lambda)}$. Then, in the input challenge phase, $\mathcal{F}_1$ and $\mathcal{F}_2$ receive basis choices $\theta_1$ and $\theta_2$, respectively, according to the following distribution:
\begin{itemize}
\item[(i)] with probability $\frac{1}{3}$, the challenger picks uniformly random $\theta', \theta'' \in \{0,1\}^{m(\lambda)}$, sends $\theta'$ to $\mathcal{F}_1$ and $\theta''$ to $\mathcal{F}_2$.
\item[(ii)] with probability $\frac{1}{3}$, the challenger picks a uniformly random $\theta' \in \{0,1\}^{m(\lambda)}$, sends $\theta'$ to $\mathcal{F}_1$ and the correct basis choice $\theta$ to $\mathcal{F}_2$.
\item[(iii)] with probability $\frac{1}{3}$, the challenger picks a uniformly random $\theta' \in \{0,1\}^{\lambda}$, sends the correct basis choice $\theta$ to $\mathcal{F}_1$ and $\theta'$ to $\mathcal{F}_2$.
\end{itemize}

Conditioned on case (i), we argue that the winning probabilities of $\mathcal{P}, \mathcal{F}_1, \mathcal{F}_2$ in $H_3$ and in $H_4$ are negligibly close, i.e. there exists a negligible function $\nu$ such that:
\begin{align}
\label{eq: recall}
    &\Big| \Pr\Big[(\mathcal{P}, \mathcal{F}_1, \mathcal{F}_2) \text{ win in } H_4 \,|\, \theta_1 \text{ is random and } \theta_2 \text{ is random} \Big] \nonumber \\
   & - \Pr\Big[(\mathcal{P}, \mathcal{F}_1, \mathcal{F}_2) \text{ win in } H_3 \,|\, \theta_1 \text{ is random and } \theta_2 \text{ is random} \Big]\Big| \leq \nu(\lambda).
\end{align}
This follows from a similar argument made earlier. By an application of the one-way-to-hiding lemma, one can bound the advantage in distinguishing $H_3$ from $H_4$ by the probability of querying at the encoded point $v$. An adversary that queries with non-negligible probability at $v$ straightforwardly yields an adversary that wins with non-negligible probability at the monogamy of entanglement game of Lemma \ref{lem: monogamy}.
This implies that the winning probability of $(\mathcal{P},\mathcal{F}_1,\mathcal{F}_2)$ in $H_4$ is at least $p-\frac{2}{3}-\nu(\lambda)$. 

What we will argue next is that in fact the lower bound is $p-\frac{2}{3} + \delta^* - \negl(\lambda)$, for some $\delta^* > 0$.
We appeal to Lemma \ref{lem: 8}, by which there exists $\epsilon^* > 0$ (where we can take $\epsilon^* = 10^{-4}$), and a negligible function $\mu$, such that for all $\lambda$, there exists an $i \in \{1,2\}$ such that $\mathcal{F}_i$ satisfies
\begin{align}
\label{eq: 26}
    \Big|\Pr&[\mathcal{F}_i \text{ returns 1 in } H_4 \,|\, \theta_i \text{ is correct } ] \nonumber \\
& -\Pr[\mathcal{F}_i \text{ returns 1 in } H_3 \,|\, \theta_i \text{ is correct} ] \Big| < 1-\epsilon^* +\mu(\lambda)\,.
\end{align}
Fix a $\lambda$, and assume \eqref{eq: 26} holds for $i = 1$ for this particular $\lambda$ (the other case being analogous).

We consider the two cases:
\begin{itemize}
    \item[(a)] \begin{equation}
    \Pr[(\mathcal{P}, \mathcal{F}_1, \mathcal{F}_2) \text{ win in } H_3 \,|\, \theta_1 \text{ is correct and } \theta_2 \text{ is random} ] > 1-\frac{\epsilon^*}{3}\,. \label{eq: 27 bis}
\end{equation}
    \item[(b)] \begin{equation} 
    \Pr[(\mathcal{P}, \mathcal{F}_1, \mathcal{F}_2) \text{ win in } H_3 \,|\, \theta_1 \text{ is a correct and } \theta_2 \text{ is random} ] \leq 1-\frac{\epsilon^*}{3}\,. 
    \end{equation}
\end{itemize}

We start by assuming (a). Then, Equation \eqref{eq: 26} implies that 
\begin{equation}
    \Pr[\mathcal{F}_1 \text{ returns 1 in } H_4 \,|\, \theta_1 \text{ is correct} ] > \frac{2\epsilon^*}{3} - \mu(\lambda)\,, \label{eq: 27}
\end{equation}
where we are using the fact that the correlation generated by $\mathcal{F}_1$ and $\mathcal{F}_2$ is non-signalling to remove the condition of $\theta_2$ being random.

By the same argument used for case (i), notice that $$| \Pr[\mathcal{F}_2 \text{ returns 0 in } H_3 \,|\, \theta_2 \text{ is random } ] - \Pr[\mathcal{F}_2 \text{ returns 0 in } H_4 \,|\, \theta_2 \text{ is random } ]| \leq  \negl(\lambda) \,.$$
Then, Equation \eqref{eq: 27 bis} implies 
\begin{equation}
    \Pr[\mathcal{F}_2 \text{ returns 0 in } H_4 \,|\, \theta_2 \text{ is random}] > 1-\frac{\epsilon^*}{3} - \negl(\lambda) \,, \label{eq: 29}
\end{equation}
where we again use the fact that the correlation generated by $\mathcal{F}_1$ and $\mathcal{F}_2$ is non-signalling.

Combining \eqref{eq: 27} and \eqref{eq: 29}, and using a union bound implies that
\begin{equation}
\label{eq: 30}
\Pr[(\mathcal{P}, \mathcal{F}_1, \mathcal{F}_2) \text{ win in } H_4 \,|\, \theta_1 \text{ is correct and } \theta_2 \text{ is random} ] > \frac{\epsilon^*}{3}- \negl(\lambda)\,.
\end{equation}

Thus, Equations \eqref{eq: 27 bis} and \eqref{eq: 30} trivially imply that 
\begin{align}
\label{eq: 31}
    &\Pr[(\mathcal{P}, \mathcal{F}_1, \mathcal{F}_2) \text{ win in } H_4 \,|\, \theta_1 \text{ is correct and } \theta_2 \text{ is random} ] \nonumber\\
    &> \Pr[(\mathcal{P}, \mathcal{F}_1, \mathcal{F}_2) \text{ win in } H_3 \,|\, \theta_1 \text{ is correct and } \theta_2 \text{ is random} ] - \left(1-\frac{\epsilon^*}{3}\right) - \negl(\lambda) \,.
\end{align}

Now, equations \eqref{eq: recall} and \eqref{eq: 31} imply:

\begin{align}
    &\Pr[(\mathcal{P}, \mathcal{F}_1, \mathcal{F}_2) \text{ win in } H_4] \nonumber\\
    &= \frac13 \Pr\big[(\mathcal{P}, \mathcal{F}_1, \mathcal{F}_2) \text{ win in } H_4 \,|\, \theta_1 \text{ is random and } \theta_2 \text{ is random} \big] \nonumber\\
    &+ \frac13 \Pr\big[(\mathcal{P}, \mathcal{F}_1, \mathcal{F}_2) \text{ win in } H_4 \,|\, \theta_1 \text{ is correct and } \theta_2 \text{ is random} \big] \nonumber\\
    &+ \frac13 \Pr\big[(\mathcal{P}, \mathcal{F}_1, \mathcal{F}_2) \text{ win in } H_4 \,|\, \theta_1 \text{ is random and } \theta_2 \text{ is correct} \big] \nonumber\\
    &> \frac13 \Pr\big[(\mathcal{P}, \mathcal{F}_1, \mathcal{F}_2) \text{ win in } H_3 \,|\, \theta_1 \text{ is random and } \theta_2 \text{ is random} \big] \nonumber\\
    &+ \frac13 \left(\Pr\big[(\mathcal{P}, \mathcal{F}_1, \mathcal{F}_2) \text{ win in } H_3 \,|\, \theta_1 \text{ is correct and } \theta_2 \text{ is a random} \big]  -  1 + \frac{\epsilon^*}{3} \right) \nonumber\\
    &+ \frac13 \left( \Pr\big[(\mathcal{P}, \mathcal{F}_1, \mathcal{F}_2) \text{ win in } H_3 \,|\, \theta_1 \text{ is random and } \theta_2 \text{ is correct}\big] -1 \right)  - \negl(\lambda) \nonumber 
    \\&= p - \frac23 +  \frac{\epsilon^*}{9} - \negl(\lambda)\,. \label{eq: first bound} 
\end{align}
Now, instead, we assume (b), i.e. 
\begin{equation} 
    \Pr\big[(\mathcal{P}, \mathcal{F}_1, \mathcal{F}_2) \text{ win in } H_3 \,|\, \theta_1 \text{ is correct and } \theta_2 \text{ is random} \big] \leq 1-\frac{\epsilon^*}{3}\,. 
\end{equation}
Then, we immediately have 
\begin{align}
\label{eq: 38}
    &\Pr\big[(\mathcal{P}, \mathcal{F}_1, \mathcal{F}_2) \text{ win in } H_4 \,|\, \theta_1 \text{ is correct and } \theta_2 \text{ is random} \big] \nonumber\\
    &> \Pr\big[(\mathcal{P}, \mathcal{F}_1, \mathcal{F}_2) \text{ win in } H_3 \,|\, \theta_1 \text{ is correct and } \theta_2 \text{ is random} \big] - \left(1-\frac{\epsilon^*}{3}\right) \,.
\end{align}
A similar calculation to Equation \eqref{eq: first bound} gives 
\begin{equation}
    \Pr\big[(\mathcal{P}, \mathcal{F}_1, \mathcal{F}_2) \text{ win in } H_4\big] > p -\frac23 + \frac{\epsilon^*}{9} - \negl(\lambda) \,.
\end{equation}
This gives the desired bound with $\delta^* = \frac{\epsilon^*}{9}$.

\end{proof}

This concludes the proof of Lemma \ref{lem: 4}, and hence of Theorem \ref{thm: security}. 

\subsection{Proof of quantum virtual black-box obfuscation}\label{subsec:VBB}

In this section, we show that our quantum copy-protection scheme for point functions is also a quantum virtual black-box ($\mathsf{VBB}$) obfuscator \cite{AlagicFefferman16}. In particular, we will show that the algorithm $\mathsf{CP.Protect}$ from Construction \ref{cons:cp} satisfies a notion of obfuscation called \emph{distributional indistinguishability}, which for evasive classes of circuits is equivalent to $\mathsf{VBB}$ obfuscation \cite{Wichs_Zirdelis_Compute_and_Compare} (it is straightforward to see that distributional indistinguishability implies VBB - the reverse implication requires slightly more work). $\mathsf{CP.Protect}$ is already functionality preserving, this follows from the definition of a copy-protection scheme. All that is left to show is security. In what follows, we assume that a program has an associated set of parameters $P.\mathsf{params}$ (input size, output size, circuit size) which we are not required to hide.
\begin{definition}[Distributional indistinguishability]
\label{def: dist indist}
An obfuscator $\mathsf{Obf}$ for the class of distributions $\mathcal{D}$ over programs $\mathcal C$, satisfies distributional indistinguishability if there exists a $\QPT$ simulator $\mathsf{Sim}$, such that for every distribution ensemble $D = \{D_{\lambda}\} \in \mathcal{D}$, we have
\begin{equation}
\mathbb{E}_{P \leftarrow D_{\lambda}} \mathsf{Obf}(1^{\lambda}, P) \, \approx_c \,\mathsf{Sim}(1^{\lambda}, P.\mathsf{params}) \,.
\end{equation}
\end{definition}
\noindent(where the notation ``$\approx_c$'' was introduced in Definition \ref{def: indistinguishability}.) We remark that the notion of 
indistinguishability of ensembles of quantum states (Definition \ref{def: indistinguishability}) already accounts for auxiliary quantum information in the two ensembles.

Distributional indistinguishability relative to any oracle is analogous to Definition \ref{def: dist indist}, except that the algorithms $\mathsf{Obf}$ and $\mathsf{Sim}$ are quantum oracle algorithms, and the notation ``$\approx_c$'' refers to Definition \ref{def: indistinguishability qrom}.

\begin{theorem}
The $\QPT$ algorithm $\mathsf{CP.Protect}$ from Construction \ref{cons:cp} satisfies distributional indistinguishability in the QROM for the class of distributions $\mathcal{D}_{\mathsf{PF}\mbox{-}\mathsf{UNP}}$.
\end{theorem}

\begin{proof}
We define the following simulator $\mathsf{Sim}$:
\begin{itemize}
    \item $\mathsf{Sim}$ takes as input $1^{\lambda}$, an auxiliary state $\nu$, and outputs the state: $\frac{\mathds{1}}{2^{m(\lambda)+\lambda}} \otimes \nu$, where the first factor is the maximally mixed state on $m(\lambda) + \lambda$ qubits.
\end{itemize}

Let
$$ \rho_{\lambda,G, H} :=  \mathbb{E}_{P_y \leftarrow D_{\lambda}} \mathsf{CP.Protect}^{G,H}(1^{\lambda}, P_y) = \mathbb{E}_{P_y \leftarrow D_{\lambda}} \mathbb{E}_{v} \ket{v^{G(y)}}\bra{v^{G(y)}} \otimes \ket{H(v)}\bra{H(v)}\,,$$
and
$$\sigma_{\lambda,H} := \mathbb{E}_{\theta}\mathbb{E}_{v} \ket{v^{\theta}}\bra{v^{\theta}} \otimes \ket{H(v)}\bra{H(v)}\,.$$

By an argument analogous to that of the proof of Lemma \ref{lem: new hybrids 2},  it holds that, for any computationally bounded oracle adversary $\mathcal{A}$, and any auxiliary state $\nu_{\lambda}$, 
$$\mathbb{E}_{G,H} \big| \Pr[\mathcal{A}^{G, H}(\rho_{\lambda, G,H} \otimes \nu_{\lambda}) = 1] -  \Pr[\mathcal{A}^{G,H}(\sigma_{\lambda, H}\otimes \nu_{\lambda}) = 1] \big| = \negl(\lambda)\,.$$

Next, we replace the state $\sigma_{\lambda,H}$ with the state $\sigma'_{\lambda} :=  \mathbb{E}_{\theta} \mathbb{E}_{v}\mathbb{E}_z \ket{v^{\theta}}\bra{v^{\theta}} \otimes \ket{z}\bra{z}$. We argue that for any query bounded adversary $\mathcal{A}$, the following holds:

\begin{equation} \mathbb{E}_{H} \big| \Pr[\mathcal{A}^{H}(\sigma_{\lambda,H} \otimes \nu_{\lambda}) = 1] -  \Pr[\mathcal{A}^{H}(\sigma'_{\lambda}\otimes \nu_{\lambda}) = 1] \big| = \negl(\lambda)\,.
\label{eq: obf}
\end{equation}

Let $\mathcal{A}$ be a distinguisher making $q$ queries. Without loss of generality, let $\mathcal{A}$ be specified by the unitary $(U O^H)^q$, for some unitary $U$. 

We apply one-way-to-hiding (Lemma \ref{lem: qrom technical step}) to deduce that the LHS of \eqref{eq: obf} is negligible if the quantity  
$$\mathbb{E}_{H}\mathbb{E}_{v}\mathbb{E}_{\theta} \mathbb{E}_{z \leftarrow \{0,1\}^m} \mathbb{E}_k \Tr{\ket{v}\bra{v} (U O^H)^k \left(\ket{v^{\theta}}\bra{v^{\theta}} \otimes \ket{z}\bra{z} \otimes \nu_{\lambda}\right)(U O^H)^k } $$
is negligible. 

Suppose for a contradiction the latter is not negligible. Then, we can construct an adversary that wins at the monogamy of entanglement game of Lemma \ref{lem: monogamy}. The reduction is straightforward: the adversary for the monogamy of entanglement game prepares the auxiliary state $\nu_{\lambda}$, and runs $\mathcal{A}$ (by simulating an oracle) to extract $v$. Then sends $v$ to both $\mathcal{B}$ and $\mathcal{C}$ (using the notation for the monogamy game of Lemma \ref{lem: monogamy}). 

The conclusion of the theorem follows by observing that $\sigma'_{\lambda}$ is the maximally mixed state.
\end{proof}

\section{Extension to compute-and-compare programs}\label{sec:CnC}
In this section, we show that a quantum copy-protection scheme for point functions, which is secure with respect to the appropriate program and challenge ensembles, implies a quantum copy-protection scheme for compute-and-compare programs with the same level of security.

The idea is simple: to copy-protect the compute-and-compare program $\mathsf{CC}[f,y]$, we copy-protect the point function $P_y$, and give $f$ in the clear. By copy-protecting $P_y$ we are copy-protecting the portion of the compute-and-compare program which checks equality with $y$. The intuition is that this is sufficient to make the functionality unclonable because its output is not already determined by $f$. More generally, one might suspect that, for copy-protecting a function $F=f_1\circ f_2 ... \circ f_\ell$, it should be sufficient to copy-protect any  of the functions $f_i$ that is sufficiently non-constant \emph{within its context}.

Let $(\mathsf{CP}\mbox{-}\mathsf{PF.Protect},\mathsf{CP}\mbox{-}\mathsf{PF.Eval})$ be a copy-protection scheme for point functions. 

\begin{construction}[Copy-protection scheme for compute-and-compare programs]\label{cons: pf to cc} The quantum copy-protection scheme $(\mathsf{CP}\mbox{-}\mathsf{CC.Protect},\mathsf{CP}\mbox{-}\mathsf{CC.Eval})$  for compute-and-compare programs is defined as follows:
\begin{itemize}
\item $\mathsf{CP}\mbox{-}\mathsf{CC.Protect}(1^{\lambda}, (f,y))$: Takes as input a security parameter $\lambda$ and a compute-and-compare program $\mathsf{CC}[f,y]$, specified succinctly by $f$ and $y$. Then,
\begin{itemize}
    \item Let $\rho = \mathsf{CP}\mbox{-}\mathsf{PF.Protect}(\lambda, y))$.
    \item Output $(f, \rho)$.
\end{itemize}
\item $\mathsf{CP}\mbox{-}\mathsf{CC.Eval}(1^{\lambda}, (f,\rho); x)$: Takes as input a security parameter $\lambda$, an alleged copy-protected program $(f, \rho)$, and a string $x \in \{0,1\}^n$ (where $n$ is the size of the inputs to $f$). Then,
\begin{itemize}
    \item Compute $y' = f(x)$.
    \item Let $b \leftarrow \mathsf{CP}\mbox{-}\mathsf{PF.Eval}(\rho; y')$. Output $b$.
\end{itemize}
\end{itemize}
\end{construction}

Recall the definition of the class of distributions over point functions $\mathcal{D}_{\mathsf{PF}\mbox{-}\mathsf{UNP}}$ from Section \ref{sec: main}. We define a related class of distributions over compute-and-compare programs.
\begin{itemize}
    \item $\mathcal{D}_{\mathsf{CC}\mbox{-}\mathsf{UNP}}$. We refer to this class as the class of \textit{unpredictable compute-and-compare programs}. This consists of ensembles $D = \{D_{\lambda}\}$ where $D_{\lambda}$ is a distribution over compute-and-compare programs such that $\mathsf{CC}[f,y] \leftarrow D_{\lambda}$ satisfies $\Hmin(y|f) \geq \lambda^\epsilon$ for some $\epsilon>0$, and where the input length of $f$ is $\lambda$ and the output length is bounded by some polynomial $t(\lambda)$.
\end{itemize}

We also define the following class of distributions over input challenges:
\begin{itemize}
\item $\mathcal{D}_{\mathsf{CC}\mbox{-}\mathsf{Chall}}$. An ensemble $D = \{D_{f,y}\}$, where each $D_{f,y}$ is a distribution over pairs of elements in the domain of $f$, belongs to the class $\mathcal{D}_{\mathsf{CC}\mbox{-}\mathsf{Chall}}$ if there exists an efficiently sampleable family $\{X_{\lambda}\}$  of distributions over $\{0,1\}^{\lambda}$ with $\Hmin(X_{\lambda}) \geq \lambda^\epsilon$, for some $\epsilon >0$, and an efficiently sampleable family $\{Z_{f,y}\}$, where $Z_{f,y}$ is a distribution over the set $f^{-1}(y)$, such that $D_{f,y}$ is the following distribution (where $\lambda$ is the size of inputs to $f$):
    \begin{itemize}
        \item With probability $1/3$, sample $z \leftarrow Z_{f,y}$ and $x \leftarrow X_{\lambda}$, and output $(x,z)$. 
    \item With probability $1/3$, sample $z \leftarrow Z_{f,y}$ and $x \leftarrow X_{\lambda}$, and output $(z,x)$. 
    \item With probability $1/3$, sample $x,x' \leftarrow X_{\lambda}$, and output $(x,x')$.  
    \end{itemize}
    We say the ensemble $D$ is \emph{specified} by the families $\{X_{\lambda}\}$ and $\{Z_{f,y}\}$.
\end{itemize}

Just like in the point function case, we also define two classes of distributions over pairs of programs and challenges.

\begin{itemize}
\item $\mathcal{D}_{\mathsf{CC}\mbox{-}\mathsf{pairs}\mbox{-}\mathsf{stat}}.$ This consists of pairs of ensembles $\left(D = \{D_{\lambda}\}, D' = \{D'_{f,y}\} \right)$ where $D \in \mathcal{D}_{\mathsf{CC}\mbox{-}\mathsf{UNP}}$ and $D' \in \mathcal{D}_{\mathsf{CC}\mbox{-}\mathsf{Chall}}$ satisfying the following. Let $D'$ be parametrized by the families  $\{X_{\lambda}\}$ and $\{Z_{f,y}\}$ (following the notation introduced above), and denote by  $\mathsf{MarkedInput}\left(D_{\lambda}, \{Z_{f,y}\} \right)$ the distribution over $\{0,1\}^{\lambda}$ induced by $D_{\lambda}$ and $\{Z_{f,y}\}$, i.e.:
\begin{itemize}
    \item Sample $(f,y) \leftarrow D_{\lambda}$, then $z \leftarrow Z_{f,y}$.
\end{itemize}
For any fixed $f_*$ with domain $\{0,1\}^{\lambda}$ such that $(f_*,y_*)$ is in the support of $D_{\lambda}$ for some $y_*$, denote by $\mathsf{MarkedInput}(D_{\lambda}, \{Z_{f,y}\})|_{f_*}$, the distribution $\mathsf{MarkedInput}(D_{\lambda}, \{Z_{f,y}\})$ conditioned on $D_{\lambda}$ sampling $f_*$.
Then, we require that, for any sequence $\{f_*^{(\lambda)}\}$ (where, for all $\lambda$, $(f_*^{(\lambda)}, y_*)$ is in the support of $D_{\lambda}$ for some $y_*$), the families $\{X_{\lambda}\}$ and $\{\mathsf{MarkedInput}(D_{\lambda}, \{Z_{f,y}\})|_{f_*^{(\lambda)}}\}$ are statistically indistinguishable. 
\item $\mathcal{D}_{\mathsf{CC}\mbox{-}\mathsf{pairs}\mbox{-}\mathsf{comp}}.$ This is defined in the same way as $\mathcal{D}_{\mathsf{CC}\mbox{-}\mathsf{pairs}\mbox{-}\mathsf{stat}}$, except that we only require $\{X_{\lambda}\}$ and $\{\mathsf{MarkedInput}(D_{\lambda}, \{Z_{f,y}\})|_{f_*^{(\lambda)}}\}$ to be \emph{computationally} indistinguishable. 
\end{itemize}
\begin{theorem}
\label{thm: from pf to cc}
Let $(\mathsf{CP}\mbox{-}\mathsf{PF.Protect},\mathsf{CP}\mbox{-}\mathsf{PF.Eval})$ be a copy-protection scheme for point functions that is $\delta$-secure with respect to all pairs $(D,D') \in \mathcal{D}_{\mathsf{PF}\mbox{-}\mathsf{pairs}\mbox{-}\mathsf{stat}}$ ($\in \mathcal{D}_{\mathsf{PF}\mbox{-}\mathsf{pairs}\mbox{-}\mathsf{comp}}$). Then, the scheme of Construction \ref{cons: pf to cc}, instantiated with  $(\mathsf{CP}\mbox{-}\mathsf{PF.Protect},\mathsf{CP}\mbox{-}\mathsf{PF.Eval})$, is a $\delta$-secure copy-protection scheme for compute-and-compare programs with respect to all pairs $(D,D') \in \mathcal{D}_{\mathsf{CC}\mbox{-}\mathsf{pairs}\mbox{-}\mathsf{stat}}$ ($\in \mathcal{D}_{\mathsf{CC}\mbox{-}\mathsf{pairs}\mbox{-}\mathsf{comp}}$). The same conclusion holds relative to any oracle, i.e. when all algorithms have access to the same oracle, with respect to query-bounded (computationally bounded) adversaries.
\end{theorem}

\begin{proof}
We prove the theorem for the case of $\left(\{D_{\lambda}\}, \{D_{f,y}\} \right)\in \mathcal{D}_{\mathsf{CC}\mbox{-}\mathsf{pairs}\mbox{-}\mathsf{stat}}$ (the case of $\left(\{D_{\lambda}\}, \{D_{f,y}\} \right)\in \mathcal{D}_{\mathsf{CC}\mbox{-}\mathsf{pairs}\mbox{-}\mathsf{comp}}$ being virtually identical). Let $t(\lambda)$ be the length of strings in the range of $f$'s sampled from $D_{\lambda}$. Let the ensemble $\{D_{f,y}\}$ be specified by $\{Z_{f,y}\}$ and $\{X_{\lambda}\}$ (using the notation introduced above for ensembles in $\mathcal{D}_{\mathsf{CC}\mbox{-}\mathsf{Chall}}$). 

Let $\mathcal{A} = (\mathcal{P},\mathcal{F}_1,\mathcal{F}_2)$ be an adversary for the compute-and-compare copy-protection scheme of Construction \ref{cons: pf to cc} with respect to ensembles $\{D_{\lambda}\}$ and $\{D_{f,y}\}$. Suppose $\mathcal{A}$ wins with probability $p(\lambda)>0$.
It follows that for each $\lambda$ there exists $f_*^{(\lambda)}$ such that $(f_*^{(\lambda)},y)$ is in the support of $D_{\lambda}$ for some $y$, and the probability that $\mathcal{A}$ wins is at least $p(\lambda)$, conditioned on $f^{(\lambda)}$ being sampled. 

We will construct an adversary $\mathcal{A}'$ that wins with probability $p(\lambda) - \negl(\lambda)$ in the point function security game with respect to $\{D'_{t(\lambda)}\}$ and $\{D'_y\}$, defined as follows:
\begin{itemize}
    \item $D'_{t(\lambda)}$: sample $x \leftarrow X_{\lambda}$ and output the point function $P_{f^{(\lambda)}_*(x)}$.
    \item $D'_y$: sample $(x,x') \leftarrow D_{f_*^{(\lambda)},y}$ and output $(f_*^{(\lambda)}(x), f_*^{(\lambda)}(x'))$.
\end{itemize}
The adversary $\mathcal{A}' = (\mathcal{P}', \mathcal{F}_1', \mathcal{F}_2')$ then acts as follows:
\begin{itemize}
    \item $\mathcal{P}'$ receives as input a state $\rho$. Then, $\mathcal{P}'$ provides $(f_*^{(\lambda)}, \rho)$ as input to $\mathcal{P}$. Let $\textsf{A}_1$ and $\textsf{A}_2$ be the registers returned by $\mathcal{P}$. $\mathcal{P}'$ forwards $\textsf{A}_1$ and $\textsf{A}_2$ to $\mathcal{F}_1'$ and $\mathcal{F}_2'$ respectively.
    \item Upon receiving a challenge $x_i$, $\mathcal{F}_i'$ samples $x_i' \leftarrow Z_{f,x_i}$. $\mathcal{F}_i'$ then runs $\mathcal{F}_i$ on input $x'_i$ and the register $\textsf{A}_i$. Let $b_i$ be the output returned by $\mathcal{F}_i$. $\mathcal{F}_i'$ returns $b_i$ to the challenger.
\end{itemize}

It is straightforward to check that the game ``simulated'' by $(\mathcal{P}', \mathcal{F}_1',\mathcal{F}_2')$ for $(\mathcal{P}, \mathcal{F}_1, \mathcal{F}_2)$ is statistically indistinguishable from a security game with respect to $\{D_{\lambda}\}$ and $\{D_{f,y}\}$, conditioned on $f^{(\lambda)}_*$. 
Thus, we deduce, by hypothesis, that $\mathcal{F}_1$ and $\mathcal{F}_2$ both return the correct bits with probability at least $p(\lambda)-\negl(\lambda)$, and thus $(\mathcal{P}', \mathcal{F}_1', \mathcal{F}_2')$ wins with probability at least $p(\lambda)-\negl(\lambda)$.
Crucially, note that $\left(\{D'_{t(\lambda)}\}, \{D'_y\}\right) \in \mathcal{D}_{\mathsf{PF}\mbox{-}\mathsf{pairs}\mbox{-}\mathsf{stat}}$. It follows that if the point function copy-protection scheme is $\delta$-secure, then the compute-and-compare copy-protection scheme must also be $\delta$-secure. It is easy to check that all steps of the proof also hold relative to an oracle.
\end{proof}

\begin{cor}
Let $\delta^*>0$ be the constant from Theorem \ref{thm: security}. There exists a $\delta^*$-secure copy-protection scheme for compute-and-compare programs with respect to ensembles $(D,D') \in \mathcal{D}_{\mathsf{PF}\mbox{-}\mathsf{pairs}\mbox{-}\mathsf{stat}}$ ($\in \mathcal{D}_{\mathsf{PF}\mbox{-}\mathsf{pairs}\mbox{-}\mathsf{comp}}$) against query-bounded adversaries (computationally bounded adversaries, assuming the existence of quantum-secure one-way functions).
\end{cor}

\begin{proof}
Combining Theorem \ref{thm: security} with Theorem \ref{thm: from pf to cc} immediately gives this corollary.
\end{proof}

\section{Secure software leasing}\label{sec:SSL}

Recall that the level of security achieved by our copy-protection scheme in Theorem \ref{thm: security} is far from ideal (our security guarantees that the freeloaders fail at least with some constant probability at simultaneously answering correctly). In this section, we show that our construction (augmented with a verification routine) satisfies a weaker notion of copy-protection called ``secure software leasing'' $(\SSL)$, introduced in \cite{ananth2020secure}, but with a standard level of security, i.e.~the adversarial advantage is negligible in the security parameter.

We formalize the notion of $\SSL$.
\begin{definition}[Secure software leasing]
\label{def: ssl}
Let $\mathcal{C}$ be a family of classical circuits with a single bit output. A secure software leasing $(\SSL)$ scheme consists of $\QPT$ algorithms $(\SSL.\Gen,
\SSL.\Lease,\SSL.\Eval,\SSL.\Verify)$ defined as follows:
\begin{itemize}
    \item $\SSL.\Gen(1^\lambda)$ takes as input the security parameter $\lambda$ and outputs a secret key $\sk$.
    \item $\SSL.\Lease(\sk, C)$ takes as input a secret key $\sk$ and a circuit $C \in \mathcal{C}$ with input size $n$, and outputs a quantum state $\rho_C$.
    \item $\SSL.\Eval(x, \rho_C)$ takes a string $x$ as input to $C$ together with a state $\rho_C$, and outputs a bit and a post-evaluation state $\tilde \rho_C$.
    \item $\SSL.\Verify(\sk, C, \sigma)$ takes as input the secret key $\sk$, the circuit $C \in \mathcal{C}$ and a state $\sigma$, and outputs $1$, if $\sigma$ is a valid lease state for $C$, and $0$ otherwise. 
\end{itemize}
There exists a negligible function $\mu$ such that the scheme satisfies:
\begin{itemize}
    \item Correctness of evaluation: for all $\lambda$, for all $C \in \mathcal{C}$, and for all $x$ in the domain of $C$,
    $$ \Pr\big[ \SSL.\Eval(x,\rho) = C(x): \rho \leftarrow  \SSL.\Lease(\sk, C), \, \sk \leftarrow \SSL.\Gen(1^\lambda) \big] \geq 1 - \mu(\lambda). \quad \quad \quad \quad \quad \quad \quad \,\,\,
    $$
    \item Correctness of verification: for all $\lambda$, for all $C \in \mathcal{C}$, and for all $x$ in the domain of $C$,
    $$\Pr\big[ \SSL.\Verify(\sk, C, \rho) = 1: \rho \leftarrow  \SSL.\Lease(\sk, C), \, \sk \leftarrow \SSL.\Gen(1^\lambda) \big] \geq 1 - \mu(\lambda).
    $$
\end{itemize}

\end{definition}
    Security is defined in terms of a security game between a lessor and an adversary $\mathcal{A}$ (the lessee). 
    Informally, any secure software leasing ($\SSL$) scheme should satisfy the following key property. After receiving a leased copy of $C$ denoted by $\rho_C$ (generated using $\SSL.\Lease$), the adversary should not be able to produce a quantum state $\sigma$ on registers $\mathsf{R}_1$ and $\mathsf{R}_2$ such that:
\begin{itemize}
    \item $\SSL.\Verify$ deems the contents of register $\mathsf{R}_1$ of $\sigma_{\mathsf{R_1R_2}}$ to be valid, once it is returned.
    \item The adversary can still compute $C$ (on inputs chosen by the lessor) from the post-measurement state in register $\mathsf{R_2}$ given by $\sigma_{\mathsf{R_2}}^* \propto \mathrm{Tr}_{\mathsf{R_1}}\big[ \Pi_1 \big[ \big( \SSL.\Verify(\cdot)_{\mathsf{R_1}} \otimes \Id_{\mathsf{R_2}} \big)\sigma_{\mathsf{R_1 R_2}} \big] \big]$.
\end{itemize}
As in the case of copy protection, the security game is specified by a security parameter $\lambda$, a distribution $D_\lambda$ over circuits $C$ from $\mathcal C$, and a family of distributions $\{D_C\}_{C\in\mathcal C}$ over inputs to $C$. The game proceeds as follows:

\begin{itemize}
    \item The lessor samples a circuit $C \leftarrow D_{\lambda}$ and runs $\sk \leftarrow \SSL.\Gen(1^\lambda)$. Then, the lessor runs $\rho \leftarrow \SSL.\Lease(\sk, C)$. The lessor sends $\rho$ to $\mathcal{A}$.
    
    \item $\mathcal{A}$ outputs a (possibly entangled) state $\sigma$ on two registers $\mathsf{R}_1$ and $\mathsf{R}_2$, and then sends the first register $\mathsf{R}_1$ to the lessor.

    \item For verification, the lessor runs $\SSL.\Verify$ on input the secret key $\sk$, the circuit $C \in \mathcal{C}$ and the register $\mathsf{R}_1$ of the state $\sigma_{\mathsf{R}_1 \mathsf{R}_2}$. If $\SSL.\Verify$ accepts, the lessor outputs $\mathsf{ok}=1$ and lets the game continue, otherwise, the lessor outputs $\mathsf{ok}=0$ and $\mathcal{A}$ loses. 
    \item The lessor samples $x \leftarrow \mathcal{D}_C$, and sends $x$ to the adversary.
    \item $\mathcal{A}$ responds with a bit $b$. If $b = C(x)$, the lessor outputs $1$. Otherwise, the lessor outputs $0$. 
\end{itemize}
We let $\mathsf{SSLGame}(\lambda, \mathcal{A}, D_{\lambda}, \{D_C\})$ denote a random variable that is equal to $1$, if the above security game is won, and $0$ otherwise.

As in the case of full copy protection, we define the trivial winning probability, which in the case of $\SSL$ is just the straightforward guessing probability for the answer to the challenge,
\begin{equation}
    p^{\mathrm{triv}, \SSL}_{D_\lambda,\{D_C\}_{C\in\mathcal C}}=\max_{b\in\{0,1\}}\mathbb E_{C\leftarrow D_\lambda}\hat D_C(b),
\end{equation}
where $\hat D_C(b)$ is the probability that the correct answer to a challenge sampled from $D_C$ is $b$.
\begin{definition}[Security]
\label{def: security ssl}
A secure software leasing $(\SSL)$ scheme for a family of classical circuits $ \mathcal{C} = \{C_{\lambda}\}_{\lambda \in \mathbb{N}}$ is said to be $\delta$-secure with respect to the ensemble $D = \{D_{\lambda}\}_{\lambda \in \mathbb{N}}$ of distributions over circuits in $\mathcal{C}$, and with respect to the ensemble $\{D_{C}\}_{C \in \mathcal{C}}$, where $D_{C}$ is a distribution over challenge inputs to program $C$, if for any $\lambda \in \mathbb{N}$ and any $\QPT$ adversary $\mathcal{A}$,
$$ \Pr[\mathsf{SSLGame}(\lambda, \mathcal{A}, D_{\lambda}, \{D_{C}\}) =1 ] \leq 1-\delta(\lambda)+\negl(\lambda)\,.$$
If $\delta(\lambda)=1-p^{\mathrm{triv}, \SSL}_{D_\lambda,\{D_C\}_{C\in\mathcal C}}$, we simply call the scheme secure.
\end{definition}
We refer to $D = \{D_{\lambda}\}_{\lambda \in \mathbb{N}}$ as the \textit{program ensemble}, and to $\{D_{C}\}_{C \in \mathcal{C}}$ as the \textit{input challenge ensemble}.
 
\subsection{Secure software leasing for point functions}\label{sec:SSL_PF}

We first consider a version of the scheme for point functions in Construction \ref{cons:cp}, augmented with a verification procedure. Then, in Section \ref{sec:SSL-CnC}, we extend the scheme to the class of compute-and-compare programs, in a similar way as in the case of copy-protection. Again, for simplicity, we hand the marked input $y \in \{0,1\}^{n}$ to the leasing algorithm as an input, rather than a circuit for the point function $P_y$ itself. We also omit the procedure $\SSL.\Gen$ as we do not require it in our construction. 

\begin{construction}[$\SSL$ scheme for point functions] 
\label{const: ssl}Let $\lambda$ be the security parameter, and let $H: \{0,1\}^{m(\lambda)} \rightarrow  \{0,1\}^{\lambda}$ and $G: \{0,1\}^{n} \rightarrow \{0,1\}^{m(\lambda)}$ be hash functions, where $m(\lambda) \geq \lambda$. Consider the following secure software leasing $(\SSL)$ scheme $(\SSL.\Lease,\SSL.\Eval, \SSL.\Verify)$ for point functions $P_y$ with marked input $y \in \{0,1\}^n$:

\begin{itemize}
\item $\SSL.\Lease(1^\lambda,y)$: Takes as input a security parameter $\lambda$ and a point function $P_y$, succinctly specified by the marked input $y$ (of size $n$)
\begin{itemize}
    \item Set $\theta = G(y)$.
    \item Sample $v \leftarrow \{0,1\}^{m(\lambda)}$ uniformly at random and let $z = H(v)$.
        \item Output $(\ket{v^{\theta}}, z)$.
\end{itemize}
\item $\SSL.\Eval(1^{\lambda}, (\rho, z) ; x)$: Takes as input a security parameter $\lambda$, a program $(\rho, z)$, and a string $x \in \{0,1\}^{n}$ (the input on which the program is to be evaluated).
    \begin{itemize}
    \item Set $\theta' = G(x)$.
    \item Apply Hadamards $H^{\theta'} = H^{\theta'_1} \otimes \dots \otimes H^{\theta'_\lambda}$ to $\rho$. Append $n+1$ ancillary qubits, all in state $\ket 0$, and compute the hash function $H$ with input $\rho$ into the first $n$ of them (possibly making use of additional ancillary qubits). Then, coherently measure whether the first $n$ ancilla qubits are in state $\ket z$, recording the result in the last ancilla qubit, uncompute the hash function $H$ and undo the Hadamards $H^{\theta'}$. Finally, measure the last ancilla qubit to obtain a bit $b$ as output.
    \end{itemize}
\item $\SSL.\Verify(1^{\lambda},y,z, \sigma)$: Apply $H^\theta$ to the input state $\sigma$, where $\theta = G(y)$, and measure in the standard basis. Output $1$ if the result is $v$ such that $H(v) = z$, and $0$ otherwise.
\end{itemize}
\end{construction}
\vspace{2mm}

The correctness property of Construction \ref{const: ssl} according to Definition \ref{def: ssl} is immediate to verify. Before stating our main theorem on the security of Construction \ref{const: ssl}, we recall the following class of distributions over point functions, which was defined in Section \ref{sec: main}.
\begin{itemize}
    \item $\mathcal{D}_{\mathsf{PF}\mbox{-}\mathsf{UNP}}$. The class of \textit{unpredictable point function distributions} $\mathcal{D}_{\mathsf{PF}\mbox{-}\mathsf{UNP}}$ consists of ensembles $D = \{D_{\lambda}\}$ where $D_{\lambda}$ is a distribution over point functions on $\{0,1\}^{\lambda}$ such that $P_y \leftarrow D_{\lambda}$ satisfies $\Hmin(y) \geq \lambda^\epsilon$ for some $\epsilon>0$.
\end{itemize}
We also define the following class of distributions over input challenges.
\begin{itemize}
\item $\mathcal{D}_{\mathsf{PF}\mbox{-}\mathsf{Chall}\mbox{-}\mathsf{SSL}}$. An ensemble $D = \{D_y\}$, where each $D_y$ is a distribution over $\{0,1\}^{|y|}$, belongs to the class $\mathcal{D}_{\mathsf{PF}\mbox{-}\mathsf{Chall}\mbox{-}\mathsf{SSL}}$ if there exists an efficiently sampleable family $\{X_{\lambda}\}$  of distributions over $\{0,1\}^{\lambda}$ with $\Hmin(X_{\lambda}) \geq \lambda^\epsilon$, for some $\epsilon >0$, such that $D_y$ is the following distribution (where $\lambda = |y|$):
    \begin{itemize}
    \item with probability $1/2$, output $y$. 
    \item with probability $1/2$, sample $x \leftarrow X_{\lambda}$, and output $x$. 
\end{itemize}
We say the ensemble $D$ is \emph{specified} by the ensemble $X_{\lambda}$.
\end{itemize}

We finally define two classes of distributions over pairs of programs and challenges.

\begin{itemize}
\item $\mathcal{D}_{\mathsf{PF}\mbox{-}\mathsf{pairs}\mbox{-}\mathsf{stat}\mbox{-}\mathsf{SSL}}.$ This consists of pairs of ensembles $\left(D = \{D_{\lambda}\}, D' = \{D'_{y}\} \right)$ where $D \in \mathcal{D}_{\mathsf{PF}\mbox{-}\mathsf{UNP}}$ and $D' \in \mathcal{D}_{\mathsf{PF}\mbox{-}\mathsf{Chall}\mbox{-}\mathsf{SSL}}$ satisfying the following. Let $D'$ be parametrized by the family  $\{X_{\lambda}\}$ (following the notation introduced above), and denote by  $\mathsf{MarkedInput}(D_{\lambda})$ the distribution over marked points in $\{0,1\}^{\lambda}$ induced by $D_{\lambda}$. Then, the families $\{X_{\lambda}\}$ and $\{\mathsf{MarkedInput}(D_{\lambda})\}$ are statistically indistinguishable. 
\item $\mathcal{D}_{\mathsf{PF}\mbox{-}\mathsf{pairs}\mbox{-}\mathsf{comp}\mbox{-}\mathsf{SSL}}.$ This is defined in the same way as $\mathcal{D}_{\mathsf{PF}\mbox{-}\mathsf{pairs}\mbox{-}\mathsf{stat}\mbox{-}\mathsf{SSL}}$, except that we only require $\{X_{\lambda}\}$ and $\{\mathsf{MarkedInput}(D_{\lambda})\}$ to be \emph{computationally} indistinguishable. 
\end{itemize}

The following is our main result on the security of Construction \ref{const: ssl}.

\begin{theorem}\label{thm:ssl}
The scheme of Construction \ref{const: ssl}, with $m(\lambda) = \poly(\lambda)$, is a secure software leasing scheme for point functions with respect to any pair of ensembles $(D, D') \in \mathcal{D}_{\mathsf{PF}\mbox{-}\mathsf{pairs}\mbox{-}\mathsf{stat}\mbox{-}\mathsf{SSL}}$ ($\in \mathcal{D}_{\mathsf{PF}\mbox{-}\mathsf{pairs}\mbox{-}\mathsf{comp}\mbox{-}\mathsf{SSL}}$), against query-bounded (computationally bounded) adversaries in the quantum random oracle model.
\end{theorem}
Theorem \ref{thm:ssl} implies that, once a leased copy is successfully returned to the lessor, no adversary can distinguish the marked input of a point function from a random (non-marked) input with probability better than $1/2$, except for a negligible advantage (in the parameter $\lambda$).\ \\

We give a proof of Theorem \ref{thm:ssl} in the next section.

\subsection{Proof of security}

To prove the theorem, we rely on a few technical results. 
\begin{lem}\label{lem:trace_ineq}
Let $\alpha \in \mathbb{C}^n$ and $A_1, \dots, A_n \in \mathbb{C}^{m \times m}$. Then, it holds that
$$
\mathrm{Tr}\Big[\sum_{i=1}^n \alpha_i A_i \Big] \, \leq \, \|\alpha\|_1 \cdot  \sum_{i = 1}^n |\mathrm{Tr}[A_i]|.
$$
\end{lem}
\begin{proof}
Using the Cauchy-Schwarz inequality, we have
$$
\mathrm{Tr}\Big[\sum_{i=1}^n \alpha_i A_i \Big] = \sum_{i=1}^n \alpha_i \mathrm{Tr}\big[A_i \big] \leq \sqrt{ \sum_{i =1}^n |\alpha_i|^2} \cdot \sqrt{ \sum_{i =1}^n |\mathrm{Tr}\big[A_i \big]|^2}.
$$
The claim follows from the norm inequality $\|x\|_2 \leq \|x\|_1$, for all $x \in \mathbb{C}^n$.
\end{proof}

\begin{lem}\label{lem:TD_inequalities}
Let $0 \leq \Pi \leq \Id$ and let $\rho$ and $\sigma$ be states such that $\mathsf{TD}(\rho,\sigma) \leq \gamma.$ Then,
$$
\mathrm{Tr}[\Pi \rho] - \gamma \leq \mathrm{Tr}[\Pi \sigma] \leq \mathrm{Tr}[\Pi \rho] + \gamma
$$
\end{lem}
\begin{proof}
By the standard identity $\mathsf{TD}(\sigma, \rho) = \underset{0 \leq \Lambda \leq \Id}{\max} \mathrm{Tr}[\Lambda (\sigma - \rho)]$, it follows that:
\begin{align*}
\mathrm{Tr}[\Pi \sigma] &=  \mathrm{Tr}[\Pi \rho] + \mathrm{Tr}[\Pi (\sigma - \rho)]\\
&\leq \mathrm{Tr}[\Pi \rho] + \underset{0 \leq \Lambda \leq \Id}{\max} \mathrm{Tr}[\Lambda ( \sigma - \rho)] \\
&= \mathrm{Tr}[\Pi \rho] + \mathsf{TD}(\sigma, \rho) \\
&\leq \mathrm{Tr}[\Pi \rho] + \gamma.
\end{align*}
The other inequality can be shown by reversing the role of $\rho$ and $\sigma$.
\end{proof}
\begin{lem}[\cite{Unruh15}, Lemma 18]\label{lem:equalityEPR}
Let $\theta \in \{0,1\}^m$ and define $\Pieq = \sum_{v \in \{0,1\}^{m}} H^\theta\ket{v} \bra{v} H^\theta \otimes H^\theta\ket{v} \bra{v} H^\theta$ (i.e the projector that checks if two registers yield the same outcome if measured in the $H^\theta$ basis). Then, the following is true for every $t \in [m]$. For any approximate $\EPR$ state, $$\ket{\phi^+_{ab}} = \frac{1}{\sqrt{2^m}}\sum_{v \in \{0,1\}^{m}}  \ket{v} \otimes X^a Z^b \ket{v},$$
where $a,b \in \{0,1\}^m$ have Hamming weight at most $t$, it follows that:
\begin{itemize}
    \item $\Pieq \ket{\phi^+_{ab}} = \ket{\phi^+_{ab}}$ holds if and only if for all $i \in [m]$:
    $$
    (\theta_i=0 \land a_i= 0) \lor (\theta_i = 1 \land b_i=0). 
    $$
    \item $\Pieq \ket{\phi^+_{ab}} = 0$ holds for all other cases.
\end{itemize}
\end{lem}
\vspace{2mm}
 We also rely on the next lemma which is based on a result by Unruh \cite[Lemma 15]{Unruh15}. To state the lemma, we define the projector onto the subspace spanned by  $\EPR$-pairs in registers $\mathsf{XY}$ with up to $t \in \mathbb{N}$ single-qubit Pauli operators applied to register $\mathsf{Y}$:
 $$\Pi_t^{\EPR} = \sum_{\substack{a,b \in \{0,1\}^{m}\\
    w(a),w(b) \leq t}} \ket{\phi^+_{ab}} \bra{\phi^+_{ab}}, \quad  \quad \ket{\phi^+_{ab}} = \frac{1}{\sqrt{2^m}}\sum_{v \in \{0,1\}^{m}}  \ket{v} \otimes X^a Z^b \ket{v},$$
 where $w(a),w(b)$ denote the Hamming weights of the strings $a$ and $b$.  Since  $\big\{\ket{\phi_{ab}^+} : a,b \in \{0,1\}^m\big\}$ forms an orthogonal basis of $\mathsf{XY}$, any state $\rho$ such that $\left(\Pi_t^{\EPR} \otimes \Id_{\mathsf{R}}\right)\rho_{\mathsf{XYR}}=\rho_{\mathsf{XYR}}$ on registers $\mathsf{X,Y}$ and $\mathsf{R}$ can be written as follows (where $a,b$ of weight greater than $t$ have probability zero):
 \begin{equation}\label{eq:decomposition}
 \rho_{\mathsf{XYR}} \, = \sum_{\substack{a,b \in \{0,1\}^{m}\\
    w(a),w(b) \leq t}} p_{ab} \, \left( \ket{\phi^+_{ab}}\bra{\phi^+_{ab}}_{\mathsf{XY}} \otimes \sigma^{a,b}_{\mathsf{R}} \right),
\end{equation}
 for some arbitrary states $\sigma^{a,b}$ and indices $a,b \in \{0,1\}^{m}$. We show the following lemma:
 \begin{lem}[Monogamy uncertainty relation]\label{lem:monogamy-uncertainty}
Fix a parameter $t \in \mathbb{N}$ and string $\theta \in \{0,1\}^m$. Let $\rho$ be a density matrix on registers $\mathsf{X,Y}$ and $\mathsf{R}$ such that $\left(\Pi_t^{\EPR} \otimes \Id_{\mathsf{R}}\right)\rho_{\mathsf{XYR}}=\rho_{\mathsf{XYR}}$. Let $\{\Pi_{v'}\}_{v' \in \{0,1\}^m}$ be an arbitrary complete set of orthogonal projectors on register $\mathsf{R}$ and measure according to the set $\left\{ H^\theta \ket{v}\bra{v}_{\mathsf{X}} H^\theta \otimes \Id_{\mathsf{Y}} \otimes {\Pi_{v'}}_{\mathsf{R}}\right\}_{v' \in \{0,1\}^m}$. Then,
$$\Pr[v' =v] \,\,= \sum_{v \in \{0,1\}^m} \mathrm{Tr}\big[ \big(H^\theta \ket{v}\bra{v}_{\mathsf{X}} H^\theta \otimes \Id_{\mathsf{Y}} \otimes {\Pi_{v}}_{\mathsf{R}} \big) \,\rho_{\mathsf{XYR}} \big] \,\leq\, 2^{-m} (m+1)^{2t}.$$
In other words, the min-entropy of the random variable $V$ (with outcome $v$) given register $\mathsf{R}$ is at least $\Hmin(V|\mathsf{R}) \geq m - 2t \log(m+1)$.
 \end{lem}
  \begin{proof}
For brevity, we define a family of projectors $\{\Lambda_{u}^\theta\}_{u}$ acting on registers $\mathsf{X}$ and $\mathsf{Y}$, where
$$\Lambda_{u}^\theta = \big(H^\theta \ket{u}\bra{u}_{\mathsf{X}} H^\theta \otimes \Id_{\mathsf{Y}}\big).$$
 Let $T$ be the set of all possible indices of weight less or equal than $t$.
 Now, using decomposition \eqref{eq:decomposition}, we can bound the success probability of measuring $v' = v$ using the information in the ancilla register $\mathsf{R}$ as follows:
 \begin{align*}
&\Pr[v' = v]\\
&= \sum_{v \in \{0,1\}^m} \mathrm{Tr}\big[ \left(H^\theta \ket{v}\bra{v}_{\mathsf{X}} H^\theta \otimes \Id_{\mathsf{Y}} \otimes {\Pi_{v}}_{\mathsf{R}} \right) \rho_{\mathsf{XYR}} \big]\\
&= \sum_{v \in \{0,1\}^m} \mathrm{Tr}\Big[ \sum_{\substack{a,b \in \{0,1\}^{m}\\
    w(a),w(b) \leq t}} p_{ab} \, \left( \Lambda_{v}^\theta \ket{\phi^+_{ab}}\bra{\phi^+_{ab}}_{\mathsf{XY}} \Lambda_{v}^\theta \right) \otimes \left( {\Pi_{v}} \sigma^{a,b}_{\mathsf{R}} \right)\Big] && \text{(by def.)}\\
& \leq \sum_{v \in \{0,1\}^m} \Big(\sum_{\substack{a,b \in \{0,1\}^{m}\\
w(a),w(b) \leq t}} p_{ab} \Big) \cdot \Big( \sum_{\substack{a,b \in \{0,1\}^{m}\\
w(a),w(b) \leq t}} \| \Lambda_{v}^\theta \ket{\phi^+_{ab}}\|^2 \cdot  \mathrm{Tr}\big[{\Pi_{v}} \sigma^{a,b}_{\mathsf{R}} \big] \Big) && \text{(Lem.~\ref{lem:trace_ineq})} \\
&= \sum_{v \in \{0,1\}^m}  \sum_{\substack{a,b \in \{0,1\}^{m}\\
w(a),w(b) \leq t}} \| H^\theta \ket{v}\bra{v}_{\mathsf{X}} H^\theta  \otimes \Id_{\mathsf{Y}} \ket{\phi^+_{ab}}\|^2 \cdot  \mathrm{Tr}\big[{\Pi_{v}} \sigma^{a,b}_{\mathsf{R}} \big] && \text{(by def.)} \\
&= \sum_{v \in \{0,1\}^m}  \sum_{\substack{a,b \in \{0,1\}^{m}\\
w(a),w(b) \leq t}} \| H^\theta \ket{v}\bra{v}_{\mathsf{X}} H^\theta  \otimes X^a Z^b_{\mathsf{Y}} \ket{\phi^+}\|^2 \cdot  \mathrm{Tr}\big[{\Pi_{v}} \sigma^{a,b}_{\mathsf{R}} \big] \\
&= \sum_{v \in \{0,1\}^m}  \sum_{\substack{a,b \in \{0,1\}^{m}\\
w(a),w(b) \leq t}} \| H^\theta \otimes X^a Z^b H^\theta \big( \ket{v}\bra{v}_{\mathsf{X}} \otimes \Id_{\mathsf{Y}} \big) \ket{\phi^+}\|^2 \cdot  \mathrm{Tr}\big[{\Pi_{v}} \sigma^{a,b}_{\mathsf{R}} \big] && \text{(Lem.~\ref{lem:ricochet})} \\
&= \sum_{v \in \{0,1\}^m}  \sum_{\substack{a,b \in \{0,1\}^{m}\\
w(a),w(b) \leq t}}  \frac{\mathrm{Tr}\big[{\Pi_{v}} \sigma^{a,b}_{\mathsf{R}} \big]}{2^{m}}
\, = \,  \sum_{\substack{a,b \in \{0,1\}^{m}\\
w(a),w(b) \leq t}}  \frac{\mathrm{Tr}\big[ \sigma^{a,b}_{\mathsf{R}} \big]}{2^{m}}
  =  \frac{|T|}{2^{m}},
 \end{align*}
 
where in the second-to-last step we used the completeness property that $\sum_{v} \Pi_v = \Id$, and in the last step we use that the ${\sigma^{a,b}}$ have unit trace, for every $a,b \in \{0,1\}^m$. It now suffices to bound $|T|$, the number of error indices of weight less or equal to $t$. In total we have $t$ indices to assign to $m+1$ possible choices (we add an additional degree of freedom to account for when there are no errors assigned). Since we have two independent indices $a,b \in \{0,1\}^m$, we get:
$$
\Pr[v' = v] \leq 2^{-m} |T| \leq 2^{-m} (m+1)^{2t}.
$$
This proves the claim.
 \end{proof}

Let us now proceed with the security proof. We consider the following sequence of hybrids of $\mathsf{SSLGame}$. We will show that the optimal winning probability in each successive hybrid changes at most negligibly. We will then bound the optimal winning probability in the final hybrid. \ \\
\ \\
$H_0$: This is the original game $\mathsf{SSLGame}$ in Section \ref{sec:SSL}:
\begin{itemize}
 \item The lessor runs $\SSL.\Lease(1^\lambda,y \in \{0,1\}^{\lambda})$ to sample $ v \leftarrow \{0,1\}^m$ and $\theta \leftarrow G(y) \in \{0,1\}^m$, and sends $(\ket{v^\theta},H(v))$ together with a circuit for $\SSL.\Eval$ to the adversary $\mathcal{A}$. 
 
 \item Having access to the random oracles $G$ and $H$, the adversary $\mathcal{A}$ outputs a (possibly entangled) state $\sigma$ on two registers $\mathsf{Y}$ and $\mathsf{R}$, and sends the register $\mathsf{Y}$ to the lessor.

\item For verification, the lessor runs $\SSL.\Verify(y,\mathsf{Y})$: Measure the register $\mathsf{Y}$ in the $H^\theta$ basis according to $\theta = G(y)$.
If the outcome is equal to $v$ such that $H(v)=z$, the lessor outputs $\mathsf{ok}=1$ and lets the game continue, otherwise, the lessor outputs $\mathsf{ok}=0$ and $\mathcal{A}$ loses.

\item Conditioning on $\mathsf{ok}=1$, the lessor sends the adversary a sample $x \leftarrow \mathcal{D}_y$ to which $\mathcal{A}$ responds with a bit (we refer to this phase of the security game as the ``input challenge phase''). Using the string $y$ given as input, the lessor outputs $1$, if the bit is equal to $P_y(x)$, and $0$ otherwise.  
\end{itemize}
\ \\
$H_1$: The game is the same as before, except that in the input challenge phase the lessor samples $x \leftarrow D_y$, and sends $G(x)$ to $\mathcal A$, (instead of sending $x$ directly).\ \\
\ \\
$H_2$: The game is the same as before, except for the input challenge phase. The lessor samples $x \leftarrow D_y$. Then, if $x \neq y$, the lessor chooses $\theta' \leftarrow \{0,1\}^m$ and sends $\theta'$ to $\mathcal{A}$ (instead of $G(x)$).\ \\
\ \\
$H_3$: The game is the same as before, except that the lessor samples $\theta \leftarrow \{0,1\}^m$ (instead of $\theta \leftarrow G(y)$). Then, in the input challenge phase, the lessor samples $x \leftarrow D_y$. If $x = y$, the lessor sends $\theta$ to $\mathcal{A}$.\ \\
\ \\
$H_4$: The game is identical to the game before, except that we replace $H(v)$ with a uniformly random string $z \leftarrow \{0,1\}^\lambda$.\\

First, we show that the advantage of any adversary in $H_4$ is negligible. Again, in the rest of the section, we denote by $p(H_i)$ the optimal winning probability in hybrid $H_i$ (see the proof in Section \ref{sec: security} for a clarification on what the expression $p(H_i)$ means formally).
\begin{lem}
$p(H_4) \leq \frac12$.
\end{lem}
\begin{proof}
First, the optimal probability of the adversary winning the game can only increase if we remove the verification portion of the game, and the lessor directly executes the input challenge phase. 

Then, we consider the state received by the adversary in the two distinct cases of the input challenge phase.
\begin{itemize}
    \item The lessor samples the marked point. In this case, the state received by the adversary is the following, which is completely independent of the oracle $H$: 
    $$ \mathbb{E}_{\theta, v} \left(\ket{v^{\theta}}\bra{v^{\theta}} \otimes \ket{\theta}\bra{\theta} \right) \otimes \mathbb{E}_{z} \ket{z} \bra{z} \,.$$ 
    Notice that the latter state is maximally mixed.
    \item The lessor samples a point other than the marked point. In this case, the adversary receives the following state, which is again independent of the oracle:
    $$ \mathbb{E}_{\theta, \theta', v} \left(\ket{v^{\theta}}\bra{v^{\theta}} \otimes \ket{\theta'}\bra{\theta'} \right) \otimes \mathbb{E}_{z} \ket{z} \bra{z}\,.$$ 
    The latter state is again maximally mixed.
\end{itemize}
Thus, an adversary can win the game $H_4$ with probability at most $\frac12$.
\end{proof}

We will now show that the optimal success probabilities in successive hybrids do not deviate by more than a negligible amount.

\begin{lem}
$|p(H_1) - p(H_0)| = \negl(\lambda)$.
\end{lem}
\begin{proof}
The proof is analogous to the proof of Lemma \ref{lem: new hybrids 1} in the security of our copy-protection scheme. The intuition is that, since $G$ is a random oracle, the pre-image $x$ does not help the adversary, and can be simulated.
\end{proof}

\begin{lem}
$|p(H_2) - p(H_1)| = \negl(\lambda)$.
\end{lem}
\begin{proof}
The proof is analogous to the proof of Lemma \ref{lem: new hybrids 2}, where an adversary that wins with non-negligible difference in $H_2$ and $H_1$ yields a distinguisher for $G(X_{\lambda})$ and $U_{m(\lambda)}$.
\end{proof}

\begin{lem}
$|p(H_3) - p(H_2)| = \negl(\lambda)$.
\end{lem}
\begin{proof}
The proof is analogous to the proof of Lemma \ref{lem: new hybrids 3}, where an adversary that wins with probabilities that differ non-negligibly in $H_3$ and $H_2$ yields a distinguisher for $G(X_{\lambda})$ and $U_{m(\lambda)}$.
\end{proof}

The crux of the security proof is showing that $p(H_3)$ and $p(H_4)$ are negligibly close.

\begin{lem}
\label{lem: 29}
$|p(H_4) - p(H_3)| = \negl(\lambda) \,.$
\end{lem}

The rest of the section is devoted to proving this lemma. At a high level, the proof has two parts:
\begin{itemize}
    \item For any adversary making $q$ queries to the oracle, we bound the difference between the winning probability in $H_3$ and in $H_4$ by $\poly(q) \cdot M$, where $M$ is a quantity related to the probability that the adversary queries the oracle at the encoded string $v$.
    \item Then, we show that the quantity $M$ is negligible. 
\end{itemize}

\begin{lem}
\label{lem: 30}
Let $\mathcal{A}$ be an adversary for $H_3$ and $H_4$, making $\poly(\lambda)$ oracle queries (pre and post verification). Suppose that $\mathcal{A}$ passes the verification step with probability at least $\frac12- \negl(\lambda)$ in $H_3$. Let $\mathcal{A}$ be specified by the unitary $U$ (i.e. $\mathcal{A}$ alternates oracles calls with applications of $U$). Let $p_{v,\theta,z,H} \in [0,1]$, and let $\rho_{\mathsf{R}}^{v,\theta,z,H}$ be density matrices, for all $v,\theta, z, H$. Let $$\sigma_{\mathsf{LR}} = \mathbb{E}_{v,\theta, z, H} \,p_{v,\theta, z, H} \left(\ket{H}\bra{H}\otimes  \ket{v}\bra{v} \otimes \ket{\theta}\bra{\theta} \otimes \ket{z}\bra{z} \right)_{\mathsf{L}} \otimes \rho^{v,\theta, z, H}_{\mathsf{R}}$$ 
be the post-verification state of the lessor and $\mathcal A$ in $H_4$ conditioned on $\mathcal A$ passing the verification step. Let $\tau_{\theta} = \frac12 \ket{\theta}\bra{\theta} + \frac12 \mathbb{E}_{\theta'} \ket{\theta'}\bra{\theta'}$. Then, 
$$|\Pr[\mathcal{A} \text{ wins in } H_3] - \Pr[\mathcal{A} \text{ wins in } H_4] | \leq \poly(\lambda) \cdot M + \negl(\lambda) \,,$$
where 
\begin{align*}
    M &=\frac12  \,\mathbb{E}_{H}\mathbb{E}_{v}\mathbb{E}_{\theta} \mathbb{E}_z \mathbb{E}_k p_{v,\theta,z,H} \Tr{\ket{v}\bra{v} (U O^{H_{v,z}})^k \left(\rho^{v,\theta, z, H}_{\mathsf{R}} \otimes \tau_{\theta} \right) \left(U O^{H_{v,z}})^k \right)^{\dagger} } \\ &+ \frac12  \,\mathbb{E}_{H}\mathbb{E}_{v} \mathbb{E}_{\theta} \mathbb{E}_z\mathbb{E}_k p_{v,\theta,z,H} \Tr{\ket{v}\bra{v} (U O^{H})^k \left(\rho^{v,\theta, z, H}_{\mathsf{R}} \otimes \tau_{\theta} \right) \left(U O^{H})^k \right)^{\dagger}}.  \nonumber
\end{align*}
\end{lem}

\begin{proof}
As we have done in several earlier proofs, we can recast $H_3$ as follows: $\mathcal{A}$ receives a uniformly random $z$, and gets access to a the reprogrammed oracle $H_{v,z}$. Let $\ket{v^{\theta}}$ denote the encoding of string $v$ using basis $\theta$. Let $q_1$ and $q_2$ denote the number of queries made by the adversary respectively before and after the verification phase.

First notice that the global states of the lessor and adversary right before the verification is executed are negligibly close in trace distance in $H_3$ and $H_4$.  

\begin{align}& \mathbb{E}_H \mathbb{E}_v \mathbb{E}_{\theta} \mathbb{E}_{z} \ket{H}\bra{H} \otimes \ket{v}\bra{v} \otimes \ket{\theta}\bra{\theta} \otimes \left((U O^{H_{v,z}})^{q_1} \ket{v^{\theta}} \bra{v^{\theta}} \otimes \ket{z}\bra{z} \left((U O^{H_{v,z'}})^{q_1}\right)^{\dagger}\right)\nonumber \\ \approx \,& \mathbb{E}_H \mathbb{E}_v \mathbb{E}_{\theta} \mathbb{E}_{z } \ket{H}\bra{H}\otimes \ket{v}\bra{v} \otimes \ket{\theta}\bra{\theta} \otimes \left((U O^{H})^{q_1} \ket{v^{\theta}} \bra{v^{\theta}} \otimes \ket{z}\bra{z}\left((U O^{H})^{q_1}\right)^{\dagger} \right)\,. \label{eq: states almost equal 2}
\end{align}
Here we have stored the complete function $H$ in an additional register, the quantum way of formulating indistinguishability of the joint distribution of $H$ and the adversary's state.

Equation \eqref{eq: states almost equal 2} follows from the one-way-to-hiding lemma (Lemma \ref{lem: qrom technical step}), and the fact that $\mathcal A$ only queries at $v$ with negligible probability (otherwise $\mathcal A$ would straightforwardly imply an adversary that wins the monogamy game (more precisely the variant of Lemma \ref{lem: monogamy}).

It follows that:
\begin{itemize}
    \item The probabilities of $\mathcal A$ passing the verification step in $H_3$ and in $H_4$ are negligibly close.
    \item The post-verification states, conditioned on passing verification must be negligibly close (this uses \eqref{eq: states almost equal 2} together with the fact that, by hypothesis, $\mathcal{A}$ passes verification with probability at least $\frac12 - \negl(\lambda)$). 
\end{itemize}

By definition, the joint state of lessor and adversary post-verification state in $H_4$ conditioned on $\mathcal{A}$ passing verification is 
$$\sigma_{\mathsf{LR}} = \mathbb{E}_{v,\theta, z, H} \,p_{v,\theta, z, H} \left(\ket{H}\bra{H}\otimes  \ket{v}\bra{v} \otimes \ket{\theta}\bra{\theta} \otimes \ket{z}\bra{z} \right)_{\mathsf{L}} \otimes \rho^{v,\theta, z, H}_{\mathsf{R}}\,.$$
Let the analogous state in $H_3$ be 
$$\tilde{\sigma}_{\mathsf{LR}} = \mathbb{E}_{v,\theta, z, H} \,p_{v,\theta, z, H} \left(\ket{H}\bra{H}\otimes  \ket{v}\bra{v} \otimes \ket{\theta}\bra{\theta} \otimes \ket{z}\bra{z} \right)_{\mathsf{L}} \otimes \tilde{\rho}^{v,\theta, z, H}_{\mathsf{R}} \,.$$

Then $ \sigma_{\mathsf{L,R}} \approx \tilde{\sigma}_{\mathsf{L,R}}$. 
Now, denote by $\{\Pi^0,\Pi^1 \}$ the projective measurement performed by $\mathcal{A}$ to guess the answer to the input challenge phase. Then, 
\begin{align}
\label{eq: 57}
    &\Pr[\mathcal{A} \text{ wins in } H_4 | \text{verification is passed}] \nonumber\\
    &= \mathbb{E}_{v,\theta, z, H} \,p_{v,\theta, z, H} \Bigg[ \frac12 \Tr{\Pi^1 (U O^{H})^{q_2} \rho^{v,\theta, z, H}_{\mathsf{R}} \otimes \ket{\theta}\bra{\theta} \left((UO^{H})^{q_2}\right)^{\dagger} }\nonumber\\ &+ \frac12 \mathbb{E}_{\theta'}\Tr{\Pi^0 (UO^{H})^{q_2} \rho^{v,\theta, z, H}_{\mathsf{R}} \otimes \ket{\theta'}\bra{\theta'} \left((UO^{H})^{q_2}\right)^{\dagger}} \Bigg] \,.
\end{align}
And, similarly,
\begin{align}
    &\Pr[\mathcal{A} \text{ wins in } H_3 | \text{verification is passed}] \nonumber\\
    &= \mathbb{E}_{v,\theta, z, H} \,p_{v,\theta, z, H}\Bigg[ \frac12 \Tr{\Pi^1 (U O^{H_{v,z}})^{q_2} \tilde{\rho}^{v,\theta, z, H}_{\mathsf{R}} \otimes \ket{\theta}\bra{\theta} \left((UO^{H_{v,z}})^{q_2}\right)^{\dagger} }\nonumber\\
    &+ \frac12 \mathbb{E}_{\theta'}\Tr{\Pi^0 (UO^{H_{v,z}})^{q_2} \tilde{\rho}^{v,\theta, z, H}_{\mathsf{R}} \otimes \ket{\theta'}\bra{\theta'} \left((UO^{H_{v,z}})^{q_2}\right)^{\dagger}} \Bigg] \nonumber\\
    \label{eq: 59}
    & \approx \mathbb{E}_{v,\theta, z, H} \,p_{v,\theta, z, H} \Bigg[\frac12\Tr{\Pi^1 (U O^{H_{v,z}})^{q_2} \rho^{v,\theta, z, H}_{\mathsf{R}} \otimes \ket{\theta}\bra{\theta} \left((UO^{H_{v,z}})^{q_2}\right)^{\dagger} }\nonumber\\ &+ \frac12 \mathbb{E}_{\theta'}\Tr{\Pi^0 (UO^{H_{v,z}})^{q_2} \rho^{v,\theta, z, H}_{\mathsf{R}} \otimes \ket{\theta'}\bra{\theta'} \left((UO^{H_{v,z}})^{q_2}\right)^{\dagger}} \Bigg]
    \,.
\end{align}
Using equations \eqref{eq: 57} and \eqref{eq: 59}, and applying the O2H lemma twice (once to bound the distance between the first terms in expressions \eqref{eq: 57} and \eqref{eq: 59}, and once to bound the distance between the second terms in \eqref{eq: 57} and \eqref{eq: 59}), we obtain:
\begin{align}
     &\left|\Pr[\mathcal{A} \text{ wins in } H_4 | \text{verification is passed}] - \Pr[\mathcal{A} \text{ wins in } H_3 | \text{verification is passed}] \right| \nonumber \nonumber \\
     &\leq \poly(\lambda) \cdot \frac12 \mathbb{E}_{v, \theta,z,H} \,p_{v,\theta,z,H} \Tr{\ket{v}\bra{v} (U O^{H})^k \left(\rho^{v,\theta, z, H}_{\mathsf{R}} \otimes \ket{\theta}\bra{\theta} \right) \left(U O^{H})^k \right)^{\dagger}} \nonumber \\
     &+ \poly(\lambda) \cdot \frac12 \mathbb{E}_{v, \theta,z,H} \,p_{v,\theta,z,H} \Tr{\ket{v}\bra{v} (U O^{H_{v,z}})^k \left(\rho^{v,\theta, z, H}_{\mathsf{R}} \otimes \ket{\theta}\bra{\theta} \right) \left(U O^{H_{v,z}})^k \right)^{\dagger}} \nonumber \\
     &+ \poly(\lambda) \cdot \frac12 \mathbb{E}_{v, \theta,z,H, \theta'} \,p_{v,\theta,z,H} \Tr{\ket{v}\bra{v} (U O^{H})^k \left(\rho^{v,\theta, z, H}_{\mathsf{R}} \otimes \ket{\theta'}\bra{\theta'} \right) \left(U O^{H})^k \right)^{\dagger}} \nonumber\\
     &+ \poly(\lambda) \cdot \frac12 \mathbb{E}_{v, \theta,z,H, \theta'} \,p_{v,\theta,z,H} \Tr{\ket{v}\bra{v} (U O^{H_{v,z}})^k \left(\rho^{v,\theta, z, H}_{\mathsf{R}} \otimes \ket{\theta'}\bra{\theta'} \right) \left(U O^{H_{v,z}})^k \right)^{\dagger}} + \negl(\lambda) \nonumber\\
     &= \poly(\lambda) \cdot \frac12 \mathbb{E}_{v, \theta,z,H} \,p_{v,\theta,z,H} \Tr{\ket{v}\bra{v} (U O^{H})^k \left(\rho^{v,\theta, z, H}_{\mathsf{R}} \otimes \tau_{\theta} \right) \left(U O^{H})^k \right)^{\dagger}} \nonumber\\
     &+ \poly(\lambda) \cdot \frac12 \mathbb{E}_{v, \theta,z,H} \,p_{v,\theta,z,H} \Tr{\ket{v}\bra{v} (U O^{H_{v,z}})^k \left(\rho^{v,\theta, z, H}_{\mathsf{R}} \otimes \tau_{\theta} \right) \left(U O^{H_{v,z}})^k \right)^{\dagger}} + \negl(\lambda)\nonumber \\
     &= \poly(\lambda)\cdot  M + \negl(\lambda)
     \,,
\end{align} 
where to get two equalities we used the definition of $\tau_{\theta}$ and $M$. This is the desired bound.
\end{proof}

In the rest of the section, we show that the quantity $M$ from Lemma \ref{lem: 30} is negligible. First of all, notice that $M$ is negligible if and only if the second term in $M$ is negligible, i.e. if and only if,
\begin{equation}\mathbb{E}_{H}\mathbb{E}_{v} \mathbb{E}_{\theta} \mathbb{E}_z\mathbb{E}_k p_{v,\theta} \Tr{\ket{v}\bra{v} (U O^{H})^k \left(\rho^{v,\theta, z, H}_{\mathsf{R}} \otimes \tau_{\theta} \right) \left(U O^{H})^k \right)^{\dagger}} = \negl(\lambda) \,. \label{eq: quantity to bound}
\end{equation}
where we are using the same notation as in Lemma \ref{lem: 30}. Thus, what we wish to show is equivalent to showing that, for any adversary $\mathcal A$ in $H_4$ who passes verification with probability at least $\frac12 - \negl(\lambda)$, the probability of querying the oracle at the encoded string $v$ at any point after a successful verification is negligible.

Thus, we will show that \eqref{eq: quantity to bound} is negligible. 
First, notice that an adversary $\mathcal{A}$ which passes verification in $H_4$ with probability at least $\frac12 -\negl(\lambda)$ and violates \eqref{eq: quantity to bound} immediately implies an adversary which succeeds at the following game $\widetilde{H}_0$ with non-negligible probability.
\vspace{2mm}

\noindent $\widetilde{H}_0$: This is identical to $H_4$ except we ask the adversary to return a guess

$v'$ for the encoded string $v$, instead of a bit. $\mathcal A$ wins if $v'=v$.
\vspace{2mm}

The reduction crucially uses the hypothesis that $\mathcal{A}$ passes verification with probability at least $\frac12 - \negl(\lambda)$.
We will show through another sequence of hybrids (which we denote using tildes) that the optimal winning probability in $\widetilde{H}_0$ is negligible. This will complete the proof that the quantity in \eqref{eq: quantity to bound}, and thus $M$ is negligible, for any adversary $\mathcal{A}$ who passes verification with probability at least $\frac12 - \negl(\lambda)$. Since the optimal winning probability in $H_3$ and $H_4$ is at least $\frac12$ (the honest strategy followed by random guessing achieves $\frac12$), this concludes the proof of Lemma \ref{lem: 29}, and hence that the optimal winning probability in $H_0$ is at most $\frac12 + \negl(\lambda)$. The following are the hybrids.
\vspace{2mm}

\noindent $\widetilde{H}_1$: Instead of sampling $v \leftarrow \{0,1\}^m$ and $\theta \leftarrow \{0,1\}^m$ at the beginning of the game, the lessor now prepares an $\EPR$ pair
on two registers $\mathsf{X}$ and $\mathsf{Y}$, and sends the registers $\mathsf{YZ}$ of the state $\ket{\phi^+ }_{\mathsf{XY}}\otimes \ket{z}_{\mathsf{Z}}$ to $\mathcal{A}$.
Rather than running $\SSL.\Verify$ for verification and measuring the register $\mathsf{Y}$, the lessor now measures both registers $\mathsf{X}$ and $\mathsf{Y}$ in the $H^\theta$ basis for a random $\theta \leftarrow \{0,1\}^m$, and checks if the outcomes result in the same string, which we denote by $v$.
\vspace{2mm}

\noindent $\widetilde{H}_2$: This game is identical to the one before, except that we change the verification procedure as follows. Instead of measuring each of the registers $\mathsf{X}$ and $\mathsf{Y}$ in the $H^\theta$ basis, the lessor now measures a bipartite projector $\Pieq$ in order to check if the registers $\mathsf{XY}$ yield the same outcome if measured in the $H^\theta$ basis. We define the projector as follows:
$${\Pieq} = \sum_{v \in \{0,1\}^{m}} H^\theta\ket{v} \bra{v}_\mathsf{X} H^\theta \otimes H^\theta\ket{v} \bra{v}_\mathsf{Y} H^\theta.$$
Afterwards, the lessor measures register $\mathsf{X}$ in the $H^\theta$ basis to determine $v$.
\vspace{2mm}

We will denote these hybrids using a tilde to distinguish them from the original sequence of hybrids.

\begin{lem}\label{lem: tildeH1 tildeH0}
$p(\widetilde{H}_1) = p(\widetilde{H}_0)$.
\end{lem}
\begin{proof}
The argument is fairly standard. 
We consider the following two statements:
\begin{itemize}
    \item sample $v \leftarrow \{0,1\}^m$, let $\theta \in \{0,1\}^m$, and output $\bigotimes_{i=1}^{m} \ket{v_i^{\theta_i}}_{\mathsf{Y}}$.
    \item create an $m$-qubit $\EPR$ pair $\ket{\phi^+}_{\mathsf{XY}}$,
measure $\mathsf{X}$ in the $H^\theta$ basis, and output register $\mathsf{Y}$.
\end{itemize}
It is evident that the equivalence of the two statements implies that $p(\widetilde{H}_1)$ and $p(\widetilde{H}_0)$ are identical. Note that we omit the register $\ket{z}$ in the proof, since it is independent of the $\EPR$ registers and thus does not affect the argument. Consider the following family of projectors given by $$ \{ \big(H^\theta \ket{v}\bra{v}H^\theta \otimes \Id_{\mathsf{Y}}\big) \}_{v \in \{0,1\}^m}.$$
We analyze the post-measurement state $\ket{\psi_v}/\sqrt{\Braket{\psi_v |\psi_v}}$ with respect to the state given by $\ket{\psi_v} = \big(H^\theta \ket{v}\bra{v}_{\mathsf{X}} H^\theta \otimes \Id_{\mathsf{Y}}\big) \ket{\phi^+}$:
\begin{align*}
\ket{\psi_v}_{\mathsf{XY}} &= \big(H^\theta \ket{v}\bra{v}H^\theta \otimes \Id  \big) \ket{\phi^+}_{\mathsf{XY}}\\
&= \Big(\big(H^\theta \otimes \Id\big) \big( \ket{v}\bra{v} \otimes \Id \big) \big(H^\theta  \otimes \Id \big) \Big) \ket{\phi^+}_{\mathsf{XY}} \\
&= \Big( \big(H^\theta \otimes \Id\big) \big( \ket{v}\bra{v} \otimes \Id \big) \big(\Id \otimes H^\theta \big) \Big) \ket{\phi^+}_{\mathsf{XY}} \quad\quad\quad \text{(Lemma \ref{lem:ricochet})}\\
&= 2^{-m/2} \sum_{v' \in \{0,1\}^m} \Big(  \big(H^\theta \otimes \Id\big) \big( \ket{v}\bra{v} \otimes \Id \big) \big(\Id \otimes H^\theta \big) \Big) \ket{v'}_{\mathsf{X}} \otimes \ket{v'}_{\mathsf{Y}}\\
&= 2^{-m/2} \sum_{v' \in \{0,1\}^m} H^\theta \ket{v}_{\mathsf{X}}\Braket{v|v'} \otimes H^\theta \ket{v'}_{\mathsf{Y}}\\
&= 2^{-m/2} H^\theta \ket{v}_{\mathsf{X}} \otimes H^\theta \ket{v}_{\mathsf{Y}}.
\end{align*}
This proves the claim, since the $\mathsf{Y}$ register of $\ket{\psi_v}/\sqrt{\Braket{\psi_v |\psi_v}}$ is identical to $\bigotimes_{i=1}^{m} \ket{v_i^{\theta_i}}$.
\end{proof}

\begin{lem}\label{lem: tildeH2 tildeH1}
$p(\widetilde{H}_2) = p(\widetilde{H}_1)$
\end{lem}
\begin{proof}
The lemma is immediate as the measurement in $\tilde H_2$ is  a coarse-graining of the measurement in $\tilde H_1$, with the acceptance condition remaining the same.
\end{proof}

In the remaining part of the proof, we will show that $p(\widetilde{H}_2)$ is negligible. 
The following is an important technical lemma, which is inspired by Lemma $16$ and Lemma $19$ in \cite{Unruh15}.

\begin{lem}
$p(\widetilde{H}_2) = \negl(\lambda)\,.$
\end{lem}

\begin{proof}
Let $\mathcal{A}$ be an adversary for $\widetilde{H}_2$. Denote by $v'$ the final guess returned by the adversary, and by $v$ the encoded string. Let $\mathsf{ok}$ be a random variable for whether the verification passes. Then, the winning probability of $\mathcal{A}$ in $\widetilde{H}_2$ is given by:
$$
\Pr \big[ v' = v \, \land \, \mathsf{ok}=1] \,.
$$
We show that, for any $t \in [m]$,
\begin{equation}
\Pr \big[ v' = v \, \land \, \mathsf{ok}=1]  \, \leq \, 2^{-m} (m+1)^{2t} + 2^{\frac{-t-1}{2}}.
\end{equation}
Picking $t \approx \sqrt{m}$ then gives the desired result, as the RHS becomes negligible in $\lambda$.\ \\

Fix a basis choice $\theta \in \{0,1\}^m$. Let $\rho_\theta$ be the state on registers $\mathsf{X},\mathsf{Y}$ and $\mathsf{R}$ in $\widetilde{H}_2$ after the verification, where $\mathsf{R}$ is the leftover register held onto by $\mathcal{A}$ that also includes the challenge $\tau_{\theta}$ (where $\tau_{\theta}$ was defined in Lemma \ref{lem: 30}) sent by the lessor after verification.

In the analysis that follows, it is convenient to approximate $\rho_\theta$ by an ideal state that is diagonal in a basis for the image of $\Pi_t^{\EPR} \otimes \Id_{\mathsf{R}}$, where $\Pi_t^{\EPR}$ is as defined in Lemma \ref{lem:monogamy-uncertainty}. Recall that $\Pi_t^{\EPR}$ projects onto the subspace spanned by $\EPR$ pairs with up to $t$ Pauli errors, i.e. onto the space spanned by the orthogonal basis states $\big\{\ket{\phi^+_{ab}} : a,b \in \{0,1\}^m\big\}$, where
\begin{equation}
\ket{\phi^+_{ab}} = \frac{1}{\sqrt{2^m}}\sum_{v \in \{0,1\}^m} \ket{v} \otimes X^a Z^b \ket{v}.
\end{equation}
We can use Lemma \ref{lem:closeness_ideal} to argue that there exists such an ideal state $\rho_\theta^\mathsf{id}$, and that the trace distance between the two states satisfies:
$$\mathsf{TD}(\rho_\theta,\rho_\theta^\mathsf{id}) \leq \sqrt{1 - \mathrm{Tr}\big[\big(\Pi_{t}^{\EPR} \otimes \Id_{\mathsf{R}}\big) \, \rho_\theta \big] }\,.$$ 

We can represent the adversary's strategy in guessing $v$, after verification, by a projective measurement $\{ \Pi_{v'} \}_{v'}$.

We are now ready to bound the probability $\Pr \big[ v' = v \, \land \, \mathsf{ok}=1]$. Let $\Theta$ be a random variable for the basis choice made by the lessor. Then, by marginalizing over $\Theta$, we get:
\begin{align}
\label{eq: 62}
\Pr \big[ v' = v \, \land \, \mathsf{ok}=1 \big] &= \sum_{\theta \in \{0,1\}^m} 2^{-m} \cdot \Pr \big[v' = v | \, \mathsf{ok}=1 \, \land \,\Theta=\theta \big] \cdot \Pr[\mathsf{ok}=1 | \Theta=\theta] \nonumber\\
&\leq \sum_{\theta \in \{0,1\}^m} 2^{-m} \cdot \Pr \big[v' = v | \, \mathsf{ok}=1 \, \land \,\Theta=\theta \big] \nonumber\\
&= \mathbb{E}_{\theta} \Pr[v'=v \, | \, \mathsf{ok}=1 \, \land \, \Theta = \theta].
\end{align}

Fix any $\theta$. Using Lemma \ref{lem:TD_inequalities} and Lemma \ref{lem:monogamy-uncertainty} we obtain:
\begin{align}
\Pr \big[v' = v \, | \, \mathsf{ok} =1 \land \Theta = \theta\big] 
&\leq 2^{-m} (m+1)^{2t} + \mathsf{TD}(\rho_\theta,\rho_\theta^\mathsf{id}) \nonumber\\
&\leq 2^{-m} (m+1)^{2t} + \sqrt{1 - \mathrm{Tr}\big[\big(\Pi_{t}^{\EPR} \otimes \Id_{\mathsf{R}}\big) \, \rho_\theta \big]}\,.
\label{eq:v_eq_v}
\end{align}

Now, averaging over $\theta$ in the above inequality gives:
\begin{align}
    \mathbb{E}_{\theta} \Pr[v'=v \,| \, \mathsf{ok}=1 \land \Theta = \theta] &\leq 2^{-m} (m+1)^{2t} + \mathbb{E}_{\theta}\sqrt{1 - \mathrm{Tr}\big[\big(\Pi_{t}^{\EPR} \otimes \Id_{\mathsf{R}}\big) \, \rho_\theta \big]} \nonumber\\
    &\leq 2^{-m} (m+1)^{2t} + \sqrt{\mathbb{E}_{\theta}\mathrm{Tr}\Big[\Big( \big(\Id - \Pi_{t}^{\EPR} \big) \otimes \Id_{\mathsf{R}}\Big) \, \rho_\theta \Big] } \,.
    \label{eq: 67}
\end{align}
where the last inequality follows from Jensen's inequality.
We will proceed to bound the above term $\mathbb{E}_{\theta}\mathrm{Tr}\big[\big( \big(\Id- \Pi_{t}^{\EPR}\big) \otimes \Id_{\mathsf{R}}\big) \, \rho_\theta \big] $ by $2^{-t-1}$.
Let us first show that for any $a,b \in \{0,1\}^m$:
\begin{align}
p_{ab} \overset{\text{def}}{=} \sum_{\theta \in \{0,1\}^m} 2^{-m} \mathrm{Tr} \big[\big(\Id - \Pi_t^{\EPR}\big) \Pieq \ket{\phi^+_{ab}} \bra{\phi^+_{ab}}_{\mathsf{XY}} \big] \,\, \leq \, 2^{-t-1}. \label{eq:pab_bound}
\end{align}

This follows from considering the following two cases:
\begin{itemize}
    \item $w(a),w(b) \leq t$: Using Lemma \ref{lem:equalityEPR} we find that one of the following is true. Depending on $\theta$, either $\Pieq \ket{\phi^+_{ab}} = 0$ or $\Pieq \ket{\phi^+_{ab}} = \ket{\phi^+_{ab}}$. We also get that $\big(\Id - \Pi_t^{\EPR} \big) \ket{\phi^+_{ab}} =0$, since $\Pi_t^{\EPR} \ket{\phi^+_{ab}} = \ket{\phi^+_{ab}}$, and thus it follows that $p_{ab}=0$.
    \item $\max\big(w(a),w(b)\big) \geq t+1$: Here, Lemma \ref{lem:equalityEPR} implies that there are at most $2^m / {2^{t+1}}$ many values of $\theta$ for which it holds that $\Pieq \ket{\phi^+_{ab}} \neq 0$, and thus $p_{ab} \leq 2^{-m} \cdot 2^m / {2^{t+1}} \, = \, 2^{-t-1}$.
\end{itemize}

Observe now that $\Pi_t^{\EPR}$ and $\ket{\phi^+_{ab}}\bra{\phi^+_{ab}}$ are diagonal in the Bell basis, hence they commute. Lemma \ref{lem:equalityEPR} implies that the same is also true for the projector $\Pieq$. 
For any fixed $\theta \in \{0,1\}^m$, we express $\rho_\theta$ as a generic density operator on registers $\mathsf{X}$, $\mathsf{Y}$ and $\mathsf{R}$ such that, for a finite index set $I^\theta$, coefficients $q_{ij}$
and an orthogonal basis $\{\ket{\Psi^{i,\theta}} \, : \,i \in I^\theta \}$ the registers $\mathsf{X}$ and $\mathsf{Y}$:
\begin{align}
{\rho_\theta} = \sum_{i,j \in I^\theta} q_{ij} \, \ket{\Psi^{i,\theta}} \bra{\Psi^{j,\theta}}_\mathsf{XY} \otimes {\sigma_{\mathsf{R}}^{i,j,\theta}}, \label{eq:rho} 
\end{align}
where $\sigma^{i,j,\theta}$ are matrices for indices $i,j \in I^\theta$.
Since we assumed that $\rho_\theta$ is the state conditioned on the verification being successful for some $\theta$, we have the property that
\begin{equation} \big(\Pieq \otimes \Id_\mathsf{R} \big) {\rho_{\theta}} \big(\Pieq \otimes \Id_\mathsf{R} \big) = {\rho_{\theta}}, \quad\quad \forall \theta \in \{0,1\}^m.\label{eq:invariant}
\end{equation}
In other words, $\rho_\theta$ on is invariant under the action of the projector $\Pieq \otimes \Id_\mathsf{R}$.
Then,
\begin{align}
&\mathbb{E}_{\theta}\mathrm{Tr}\Big[ \big(\Id - \Pi_{t}^{\EPR} \big) \otimes \Id_\mathsf{R} \, {\rho_\theta} \Big] \nonumber\\
&=\sum_{\theta \in \{0,1\}^m} 2^{-m} \mathrm{Tr} \big[\big(\Id - \Pi_t^{\EPR}\big)\otimes \Id_\mathsf{R} \, {\rho_\theta} \big]\nonumber\\
& = \sum_{\theta \in \{0,1\}^m} 2^{-m} \mathrm{Tr} \big[\big(\Id - \Pi_t^{\EPR}\big)  \otimes \Id_\mathsf{R} \, \big(\Pieq \otimes \Id_\mathsf{R} \big) {\rho_{\theta}} \big(\Pieq \otimes \Id_\mathsf{R} \big) \big] \quad\quad\quad\quad\quad (\text{Eq.}~\eqref{eq:invariant}) \nonumber\\
& = \sum_{\theta \in \{0,1\}^m} 2^{-m} \mathrm{Tr} \Big[\Big(\big(\Id - \Pi_t^{\EPR}\big) \Pieq  \otimes \Id_\mathsf{R}\Big) \, {\rho_\theta} \Big]\nonumber\\
&= \sum_{\theta \in \{0,1\}^m} 2^{-m} \mathrm{Tr} \Big[ \Big( \sum_{a,b \in \{0,1\}^m}  \ket{\phi^+_{ab}} \bra{\phi^+_{ab}} \Big) \big(\Id - \Pi_t^{\EPR}\big) \Pieq \otimes \Id_\mathsf{R} \, {\rho_\theta} \Big]\nonumber\\
&= \sum_{\theta \in \{0,1\}^m} \sum_{a,b \in \{0,1\}^m} 2^{-m} \, \mathrm{Tr} \Big[ \ket{\phi^+_{ab}} \bra{\phi^+_{ab}} \big(\Id - \Pi_t^{\EPR}\big) \Pieq \otimes \Id_\mathsf{R} \, {\rho_\theta} \Big]\nonumber\\
&=\sum_{\theta \in \{0,1\}^m} \sum_{a,b \in \{0,1\}^m} 2^{-m} \, \mathrm{Tr} \Big[\big(\Id - \Pi_t^{\EPR}\big) \Pieq \otimes \Id_\mathsf{R} \, \big( \ket{\phi^+_{ab}} \bra{\phi^+_{ab}} \otimes \Id_\mathsf{R} \big) {\rho_\theta} \big( \ket{\phi^+_{ab}} \bra{\phi^+_{ab}} \otimes \Id_\mathsf{R} \big)  \Big]\nonumber
\end{align}
In the third to last line, we inserted the complete set $\sum_{a,b}  \ket{\phi^+_{ab}} \bra{\phi^+_{ab}} = \Id$.
Then, using the definition of $\rho$ in Eq.\eqref{eq:rho}, we can continue to expand the expression above as follows:
\begin{align}
&\sum_{\theta \in \{0,1\}^m} \sum_{a,b \in \{0,1\}^m} 2^{-m} \, \mathrm{Tr} \Big[\big(\Id - \Pi_t^{\EPR}\big) \Pieq \otimes \Id_\mathsf{R} \, \big( \ket{\phi^+_{ab}} \bra{\phi^+_{ab}} \otimes \Id_\mathsf{R} \big) {\rho_\theta} \big( \ket{\phi^+_{ab}} \bra{\phi^+_{ab}} \otimes \Id_\mathsf{R} \big)  \Big]\nonumber\\
&= \sum_{\theta \in \{0,1\}^m} \sum_{a,b \in \{0,1\}^m} 2^{-m} \sum_{i,j \in I^\theta} q_{ij} \, \mathrm{Tr} \Big[\big(\Id - \Pi_t^{\EPR}\big) \Pieq  \, \ket{\phi^+_{ab}} \bra{\phi^+_{ab}} \big(\ket{\Psi^{i,\theta}} \bra{\Psi^{j,\theta}}_\mathsf{XY}\big) \ket{\phi^+_{ab}}  \bra{\phi^+_{ab}} \otimes {\sigma_{\mathsf{R}}^{i,j,\theta}} \Big]\nonumber\\
&= \sum_{a,b \in \{0,1\}^m} \sum_{i,j \in I^\theta} p_{ab} \,q_{ij} \bra{\phi^+_{ab}} \big(\ket{\Psi^{i,\theta}} \bra{\Psi^{j,\theta}}_\mathsf{XY}\big)  \ket{\phi^+_{ab}} \,\,\mathrm{Tr}\big[ {\sigma_{\mathsf{R}}^{i,j,\theta}}\big] 
\quad\quad\quad\quad\,\,\,\, (\text{by def.})
\nonumber \\
&\leq\, 2^{-t-1} \sum_{i,j \in I^\theta} q_{ij} \sum_{a,b \in \{0,1\}^m} \bra{\phi^+_{ab}} \big(\ket{\Psi^{i,\theta}} \bra{\Psi^{j,\theta}}_\mathsf{XY}\big) \ket{\phi^+_{ab}} \, \mathrm{Tr}\big[ {\sigma_{\mathsf{R}}^{i,j,\theta}}\big] \quad\quad\quad\,\,\, (\text{Eq.}~\eqref{eq:pab_bound}) \nonumber\\
&=\, 2^{-t-1} \sum_{i,j \in I^\theta} q_{ij} \mathrm{Tr}\big[\ket{\Psi^{i,\theta}} \bra{\Psi^{j,\theta}}_\mathsf{XY}\big] \, \mathrm{Tr}\big[ {\sigma_{\mathsf{R}}^{i,j,\theta}}\big] \quad = \quad 2^{-t-1} \, \mathrm{Tr}\big[{\rho_\theta}\big] \quad = \quad  2^{-t-1}.  \nonumber
\end{align}
In the last line, we used that $\big\{\ket{\phi^+_{ab}} : a,b \in \{0,1\}^m\big\}$ is an orthogonal basis for $\mathsf{XY}$. Thus, we get
$$
\mathbb{E}_{\theta}\mathrm{Tr}\Big[ \big(\Id - \Pi_{t}^{\EPR} \big) \otimes \Id_\mathsf{R} \, {\rho_\theta} \Big] \,\, \leq \,\, 2^{-t-1}.
$$
Plugging this bound in \eqref{eq: 67} and then into \eqref{eq: 62} gives
\begin{equation}
    \Pr \big[ v' = v \, \land \, \mathsf{ok}=1 \big] \leq 2^{-m} (m+1)^{2t} + 2^{\frac{-t-1}{2}} \,.
\end{equation}
Choosing $t \approx \sqrt{m}$ makes the RHS negligible.
\end{proof}

\begin{cor}
$p(\widetilde{H}_0) = \negl(\lambda)$.
\end{cor}

As we argued earlier, this concludes the proof of Lemma \eqref{lem: 29}, and thus of Theorem \ref{thm:ssl}.

\subsection{Extension to compute-and-compare programs}\label{sec:SSL-CnC}

In this section, we show that an $\SSL$ scheme for point functions, which is secure with respect to the appropriate program and challenge ensembles, implies an $\SSL$ scheme for compute-and-compare programs with the same level of security, with respect to apppropriate program and challenge ensembles.
The idea is the same as in Section \ref{sec:CnC}: to lease the compute-and-compare program $\mathsf{CC}[f,y]$, we first lease the point function $P_y$, and then hand out the function $f$ in the clear. 

Let $(\SSL\mbox{-}\mathsf{PF.Gen},\SSL\mbox{-}\mathsf{PF.Lease},\SSL\mbox{-}\mathsf{PF.Eval},\SSL\mbox{-}\mathsf{PF.Verify})$ be any $\SSL$ scheme for point functions. The compute-and-compare program scheme is defined as follows:

\begin{construction}[$\SSL$ scheme for compute-and-compare programs]\label{cons: pf to cc - SSL} The $\SSL$ scheme $(\SSL\mbox{-}\mathsf{CC.Gen},\SSL\mbox{-}\mathsf{CC.Lease},\SSL\mbox{-}\mathsf{CC.Eval},\SSL\mbox{-}\mathsf{CC.Verify})$ for compute-and-compare programs is defined by:
\begin{itemize}
\item $\SSL\mbox{-}\mathsf{CC.Gen}(1^{\lambda})$: Takes as input the security parameter $\lambda$. Then,
\begin{itemize}
    \item Let $\mathsf{sk} \leftarrow \SSL\mbox{-}\mathsf{PF.Gen}(1^{\lambda})$. Output $\mathsf{sk}$.
\end{itemize}
\item $\SSL\mbox{-}\mathsf{CC.Lease}(1^{\lambda},\mathsf{sk},(f,y))$: Takes as input a security parameter $\lambda$, a secret key $sk$, and a compute-and-compare program $\mathsf{CC}[f,y]$, specified succinctly by $f$ and $y$. Then,
\begin{itemize}
    \item Let $\rho = \SSL\mbox{-}\mathsf{PF.Lease}(1^\lambda,\mathsf{sk},y))$. Output $(f, \rho)$.
\end{itemize}
\item $\SSL\mbox{-}\mathsf{CC.Eval}(1^{\lambda}, (f,\rho); x)$: Takes as input a security parameter $\lambda$, an alleged program copy $(f, \rho)$, and a string $x \in \{0,1\}^n$ (where $n$ is the size of the inputs to $f$). Then,
\begin{itemize}
    \item Compute $y' = f(x)$.
    \item Let $b \leftarrow \SSL\mbox{-}\mathsf{PF.Eval}(\rho; y')$. Output $b$.
\end{itemize}
\item $\SSL\mbox{-}\mathsf{CC.Verify}(1^{\lambda},\mathsf{sk},(f,\rho); \sigma)$:
\begin{itemize}
    \item Let $b' \leftarrow \SSL\mbox{-}\mathsf{PF.Verify}(1^{\lambda},\mathsf{sk},y; \sigma)$. Output $b'$.
\end{itemize}
\end{itemize}
\end{construction}

Recall the definition of the class of distributions over compute-and-compare programs  $\mathcal{D}_{\mathsf{CC}\mbox{-}\mathsf{UNP}}$ in Section \ref{sec:CnC}. We recall it here for convenience.
\begin{itemize}
    \item $\mathcal{D}_{\mathsf{CC}\mbox{-}\mathsf{UNP}}$. We refer to this class as the class of \textit{unpredictable compute-and-compare programs}. This consists of ensembles $D = \{D_{\lambda}\}$ where $D_{\lambda}$ is a distribution over compute-and-compare programs such that $\mathsf{CC}[f,y] \leftarrow D_{\lambda}$ satisfies $\Hmin(y|f) \geq \lambda^\epsilon$ for some $\epsilon>0$, and where the input length of $f$ is $\lambda$ and the output length is bounded by some polynomial $t(\lambda)$.
\end{itemize}

We also define the following class of distributions over input challenges:
\begin{itemize}
\item $\mathcal{D}_{\mathsf{CC}\mbox{-}\mathsf{Chall}\mbox{-}\mathsf{SSL}}$. An ensemble $D = \{D_{f,y}\}$, where each $D_{f,y}$ is a distribution over the domain of $f$, belongs to the class $\mathcal{D}_{\mathsf{CC}\mbox{-}\mathsf{Chall}\mbox{-}\mathsf{SSL}}$ if there exists an efficiently sampleable family $\{X_{\lambda}\}$  of distributions over $\{0,1\}^{\lambda}$ with $\Hmin(X_{\lambda}) \geq \lambda^\epsilon$, for some $\epsilon >0$, and an efficiently sampleable family $\{Z_{f,y}\}$, where $Z_{f,y}$ is a distribution over the set $f^{-1}(y)$, such that $D_{f,y}$ is the following distribution (where $\lambda$ is the size of inputs to $f$):
\begin{itemize}
    \item with probability $1/2$, sample $z \leftarrow Z_{f,y}$ and output $z$. 
    \item with probability $1/2$, sample $x \leftarrow X_{\lambda}$, and output $x$. 
\end{itemize}
We say the ensemble $D$ is \emph{specified} by the families $\{X_{\lambda}\}$ and $\{Z_{f,y}\}$.
\end{itemize}

Similarly to Section \ref{sec:CnC}, we also define two classes of distributions over pairs of programs and challenges for compute-and-compare programs. 
\begin{itemize}
\item $\mathcal{D}_{\mathsf{CC}\mbox{-}\mathsf{pairs}\mbox{-}\mathsf{stat}\mbox{-}\mathsf{SSL}}.$ This consists of pairs of ensembles $\left(D = \{D_{\lambda}\}, D' = \{D'_{f,y}\} \right)$ where $D \in \mathcal{D}_{\mathsf{CC}\mbox{-}\mathsf{UNP}}$ and $D' \in \mathcal{D}_{\mathsf{CC}\mbox{-}\mathsf{Chall}\mbox{-}\mathsf{SSL}}$ satisfying the following. Let $D'$ be specified by the families  $\{X_{\lambda}\}$ and $\{Z_{f,y}\}$, and denote by  $\mathsf{MarkedInput}\left(D_{\lambda}, \{Z_{f,y}\} \right)$ the distribution over $\{0,1\}^{\lambda}$ induced by $D_{\lambda}$ and $\{Z_{f,y}\}$, i.e.:
\begin{itemize}
    \item Sample $(f,y) \leftarrow D_{\lambda}$, then output $z \leftarrow Z_{f,y}$.
\end{itemize}
For any fixed $f_*$ with domain $\{0,1\}^{\lambda}$ such that $(f_*,y_*)$ is in the support of $D_{\lambda}$ for some $y_*$, denote by $\mathsf{MarkedInput}(D_{\lambda}, \{Z_{f,y}\})|_{f_*}$, the distribution $\mathsf{MarkedInput}(D_{\lambda}, \{Z_{f,y}\})$ conditioned on $D_{\lambda}$ sampling $f_*$.
Then, we require that, for any sequence $\{f_*^{(\lambda)}\}$ (where, for all $\lambda$, $(f_*^{(\lambda)}, y_*)$ is in the support of $D_{\lambda}$ for some $y_*$) the families $\{X_{\lambda}\}$ and $\{\mathsf{MarkedInput}(D_{\lambda}, \{Z_{f,y}\})|_{f_*^{(\lambda)}}\}$ are statistically indistinguishable. 
\item $\mathcal{D}_{\mathsf{CC}\mbox{-}\mathsf{pairs}\mbox{-}\mathsf{comp}\mbox{-}\mathsf{SSL}}.$ This is defined in the same way as $\mathcal{D}_{\mathsf{CC}\mbox{-}\mathsf{pairs}\mbox{-}\mathsf{stat}\mbox{-}\mathsf{SSL}}$, except that we only require $\{X_{\lambda}\}$ and $\{\mathsf{MarkedInput}(D_{\lambda}, \{Z_{f,y}\})|_{f_*^{(\lambda)}}\}$ to be \emph{computationally} indistinguishable. 
\end{itemize}

\begin{theorem}
\label{thm: from pf to cc - SSL}
Let $(\SSL\mbox{-}\mathsf{PF.Gen},\SSL\mbox{-}\mathsf{PF.Lease},\SSL\mbox{-}\mathsf{PF.Eval},\SSL\mbox{-}\mathsf{PF.Verify})$ be an $\SSL$ scheme for point functions that is $\delta$-secure with respect to all pairs $(D,D') \in \mathcal{D}_{\mathsf{PF}\mbox{-}\mathsf{pairs}\mbox{-}\mathsf{stat}\mbox{-}\mathsf{SSL}}$ ($\in \mathcal{D}_{\mathsf{PF}\mbox{-}\mathsf{pairs}\mbox{-}\mathsf{comp}\mbox{-}\mathsf{SSL}}$). Then, the scheme of Construction \ref{cons: pf to cc - SSL} is a $\delta$-secure $\SSL$ scheme for compute-and-compare programs with respect to all pairs $(D,D') \in \mathcal{D}_{\mathsf{CC}\mbox{-}\mathsf{pairs}\mbox{-}\mathsf{stat}\mbox{-}\mathsf{SSL}}$ ($\in \mathcal{D}_{\mathsf{CC}\mbox{-}\mathsf{pairs}\mbox{-}\mathsf{comp}\mbox{-}\mathsf{SSL}}$). The same conclusion holds relative to any oracle, i.e. when all algorithms have access to the same oracle, with respect to query-bounded (computationally bounded) adversaries.
\end{theorem}

The proof of Theorem \ref{thm: from pf to cc - SSL} uses a similar reduction to the point function security game as in the copy-protection variant in Theorem \ref{thm: from pf to cc}. The main difference is that the reduction between the $\SSL$ games now involves a verification step.
We add the proof for completeness.

\begin{proof}[Proof of Theorem \ref{thm: from pf to cc - SSL}]
We prove the claim for $\left(\{D_{\lambda}\}, \{D_{f,y}\} \right)\in \mathcal{D}_{\mathsf{CC}\mbox{-}\mathsf{pairs}\mbox{-}\mathsf{stat}\mbox{-}\mathsf{SSL}}$ only, since the case of $\left(\{D_{\lambda}\}, \{D_{f,y}\} \right)\in \mathcal{D}_{\mathsf{CC}\mbox{-}\mathsf{pairs}\mbox{-}\mathsf{comp}\mbox{-}\mathsf{SSL}}$ is virtually identical. Let $t(\lambda)$ be the length of strings in the range of $f$'s sampled from $D_{\lambda}$ and let the ensemble $\{D_{f,y}\}$ be specified by $\{X_{\lambda}\}$ and $\{Z_{f,y}\}$ (using the notation introduced above for ensembles in $\mathcal{D}_{\mathsf{CC}\mbox{-}\mathsf{Chall}\mbox{-}\mathsf{SSL}}$). 

Let $\mathcal{A}$ be an adversary for the compute-and-compare $\SSL$ scheme of Construction \ref{cons: pf to cc - SSL} with respect to ensembles $\{D_{\lambda}\}$ and $\{D_{f,y}\}$ who wins at the $\SSL$ security game with probability $p(\lambda)>0$.
It then follows that for each $\lambda$ there exists $f_*^{(\lambda)}$ such that $(f_*^{(\lambda)},y)$ is in the support of $D_{\lambda}$ for some $y$, and such that the probability that $\mathcal{A}$ wins is at least $p(\lambda)$, conditioned on $f_*^{(\lambda)}$ being sampled. 
We will construct an adversary $\mathcal{A}'$ that wins with probability $p(\lambda) - \negl(\lambda)$ in the point function security game with respect to the distributions $\{D'_{t(\lambda)}\}$ and $\{D'_y\}$, defined as follows:
\begin{itemize}
    \item $D'_{t(\lambda)}$: sample $x \leftarrow X_{\lambda}$ and output the point function $P_{f^{(\lambda)}_*(x)}$.
    \item $D'_y$: sample $x \leftarrow D_{f_*^{(\lambda)},y}$ and output $f_*^{(\lambda)}(x)$.
\end{itemize}
The adversary $\mathcal{A}'$ against the point function $\SSL$ game acts as follows:
\begin{itemize}
    \item $\mathcal{A}'$ receives a state $\rho$ from the lessor, and then forwards $(f_*^{(\lambda)}, \rho)$ to adversary $\mathcal{A}$.
    \item $\mathcal{A}$ returns a supposed program copy $\sigma$ for the point function to $\mathcal{A}'$ who then sends it back to the lessor for verification.
    \item Conditioning on the verification being successful, the lessor replies with a challenge input $x \leftarrow D'_y$. $\mathcal{A}'$ then
    samples $x' \leftarrow Z_{f,x}$, and runs $\mathcal{A}$ with input challenge $x'$.
    \item Let $b$ be the bit returned by $\mathcal{A}$. The adversary $\mathcal{A}'$ replies with the same $b$ to the lessor.
\end{itemize}

It is straightforward to check that the game ``simulated'' by $\mathcal{A}'$ for $\mathcal{A}$ is statistically indistinguishable from a security game with respect to $\{D_{\lambda}\}$ and $\{D_{f,y}\}$, conditioned on $f^{(\lambda)}_*$. 
Thus, we deduce, by hypothesis, that $\mathcal{A}$
passes verification and returns the correct bit with probability at least $p(\lambda)-\negl(\lambda)$, and thus $\mathcal{A}'$ wins with probability at least $p(\lambda)-\negl(\lambda)$.
Crucially, note that $\left(\{D'_{t(\lambda)}\}, \{D'_y\}\right) \in \mathcal{D}_{\mathsf{PF}\mbox{-}\mathsf{pairs}\mbox{-}\mathsf{stat}\mbox{-}\mathsf{SSL}}$. It follows that if the $\SSL$ is $\delta$-secure, then the compute-and-compare scheme must also be $\delta$-secure. 

The proof of the theorem statement relative to any oracle is analogous.
\end{proof}

\bibliographystyle{plainnat}
\bibliography{references}

\appendix
\section{Appendix}
\subsection{Proof of Lemma \ref{lem: qrom technical step}}
\label{apx: lem 7 proof}
\begin{proof}
For any $x \in \{0,1\}^{\lambda}$, define $V_x^H =\left( U O^H (I - \ket{x}\bra{x}) \right)^q$ and define $W_x^H = UO^H - V_x^H$. Then,
\begin{align}
&\frac12 \mathbb{E}_{H}\mathbb{E}_{x \leftarrow X}\| \Pi^0  (U O^H)^q\left( \ket{H(x)} \otimes \ket{\psi_x} \right)\|^2 + \frac12 \mathbb{E}_{H}\mathbb{E}_{z \leftarrow \{0,1\}^m}\| \Pi^1  (U O^H)^q\left( \ket{z} \otimes \ket{\psi_x} \right)\|^2 \nonumber \\
= &\frac12 \mathbb{E}_{H}\mathbb{E}_{x \leftarrow X} \mathbb{E}_{z \leftarrow \{0,1\}^m} \| \Pi^0  (U O^{H_{x,z}})^q\left( \ket{z} \otimes \ket{\psi_x} \right)\|^2 + \frac12 \mathbb{E}_{H}\mathbb{E}_{z \leftarrow \{0,1\}^m}\| \Pi^1  (U O^H)^q\left( \ket{z} \otimes \ket{\psi_x} \right)\|^2 \nonumber \\
=& \frac12 \mathbb{E}_{H}\mathbb{E}_{x \leftarrow X} \mathbb{E}_{z \leftarrow \{0,1\}^m} \| \Pi^0  (V_x^{H_{x,z}} + W_x^{H_{x,z}})\left( \ket{z} \otimes \ket{\psi_x} \right)\|^2 \nonumber \\
& \quad +\frac12 \mathbb{E}_{H}\mathbb{E}_{x \leftarrow X}\mathbb{E}_{z \leftarrow \{0,1\}^m}\| \Pi^1  (V_x^H + W_x^H)\left( \ket{z} \otimes \ket{\psi_x} \right)\|^2 \nonumber \\
\leq& \frac12 \mathbb{E}_{H}\mathbb{E}_{x \leftarrow X} \mathbb{E}_{z \leftarrow \{0,1\}^m} \| \Pi^0 V_x^{{H_{x,z}}} \left( \ket{z} \otimes \ket{\psi_x} \right)\|^2 + \frac12 \mathbb{E}_{H}\mathbb{E}_{x \leftarrow X}\mathbb{E}_{z \leftarrow \{0,1\}^m}\| \Pi^1 V_x^H \left( \ket{z} \otimes \ket{\psi_x} \right)\|^2 \nonumber \\
+& \frac12 (3q+2)q \,\mathbb{E}_{H}\mathbb{E}_{x \leftarrow X} \mathbb{E}_{z \leftarrow \{0,1\}^m} \mathbb{E}_k \|  \ket{x}\bra{x} (U O^{{H_{x,z}}})^k \ket{z} \otimes \ket{\psi_x}\| \nonumber \\
+& \frac12 (3q+2)q \,\mathbb{E}_{H} \mathbb{E}_{x \leftarrow X}\mathbb{E}_{z \leftarrow \{0,1\}^m} \mathbb{E}_k \|  \ket{x}\bra{x} (U O^H)^k \ket{z} \otimes \ket{\psi_x}\|  \nonumber \\
=& \frac12 \mathbb{E}_{H}\mathbb{E}_{x \leftarrow X} \mathbb{E}_{z \leftarrow \{0,1\}^m} \| \Pi^0 V_x^{{H_{x,z}}} \left( \ket{z} \otimes \ket{\psi_x} \right)\|^2 + \frac12 \mathbb{E}_{H}\mathbb{E}_{x \leftarrow X}\mathbb{E}_{z \leftarrow \{0,1\}^m}\| \Pi^1 V_x^H \left( \ket{z} \otimes \ket{\psi_x} \right)\|^2 \nonumber \\
+& (3q+2)q\,M \, \label{eq: first inequality}
\end{align}
where the first equality uses Lemma \ref{lem: basic}, and the inequality uses Lemma 18 in \cite{broadbent2019uncloneable}.

In order to prove the desired inequality, it is sufficient to show that 
\begin{equation}
    \frac12 \mathbb{E}_{H}\mathbb{E}_{x \leftarrow X} \mathbb{E}_{z \leftarrow \{0,1\}^m} \| \Pi^0 V_x^{{H_{x,z}}} \left( \ket{z} \otimes \ket{\psi_x} \right)\|^2 + \frac12 \mathbb{E}_{H}\mathbb{E}_{x \leftarrow X}\mathbb{E}_{z \leftarrow \{0,1\}^m}\| \Pi^1 V_x^H \left( \ket{z} \otimes \ket{\psi_x} \right)\|^2 \leq \frac12 \,.
\end{equation}
Notice that $V_x^{H_{x,z}} = V_x^{H}$, since $V_x^H$ projects onto the subspace orthogonal to $x$ before every query to $H$. This implies that the LHS simplifies as
\begin{align}
    &\frac12 \mathbb{E}_{H}\mathbb{E}_{x \leftarrow X} \mathbb{E}_{z \leftarrow \{0,1\}^m} \| \Pi^0 V_x^{{H_{x,z}}} \left( \ket{z} \otimes \ket{\psi_x} \right)\|^2 + \frac12 \mathbb{E}_{H}\mathbb{E}_{x \leftarrow X}\mathbb{E}_{z \leftarrow \{0,1\}^m}\| \Pi^1 V_x^H \left( \ket{z} \otimes \ket{\psi_x} \right)\|^2    \nonumber \\
    =&\frac12 \mathbb{E}_{H}\mathbb{E}_{x\leftarrow X} \mathbb{E}_{z \leftarrow \{0,1\}^m} \| \Pi^0 V_x^H \left( \ket{z} \otimes \ket{\psi_x} \right)\|^2  + \frac12 \mathbb{E}_{H}\mathbb{E}_{x \in X} \mathbb{E}_{z \leftarrow \{0,1\}^m} \| \Pi^1 V_x^H \left( \ket{z} \otimes \ket{\psi_x} \right)\|^2 \nonumber\\
    =& \frac12 \mathbb{E}_{H}\mathbb{E}_{x \in X} \mathbb{E}_{z \leftarrow \{0,1\}^m} \| V_x^H \left( \ket{z} \otimes \ket{\psi_x} \right)\|^2  \nonumber \\
    \leq & \frac12 \,\,, \label{eq: one half}
\end{align}
where to get the third line, we used the fact that $\Pi^0, \Pi^1$ are a complete pair of orthogonal projectors, and to get the last line we exploited properties of the Euclidean norm. 

Combining \eqref{eq: first inequality} and \eqref{eq: one half} gives the desired inequality. 

With a little extra work, one can show that $M$ is negligible if and only if 
$$\frac12 \mathbb{E}_{H}\mathbb{E}_{x \leftarrow X}\mathbb{E}_{z \leftarrow \{0,1\}^m}\| \Pi^1 V_x^H \left( \ket{z} \otimes \ket{\psi_x} \right)\|^2 $$
is negligible. We refer the reader to the proof of Theorem 3 in \cite{ambainis2019quantum} for the full details.
\end{proof}

\subsection{Proof of Lemma \ref{lem: new hybrids 1}}
\label{app: 1}
We continue the proof of Lemma \ref{lem: new hybrids 1}, using the notation we introduced in the main text. In what follows, $\mathcal{P}$, $\mathcal{F}_1$, $\mathcal{F}_2$ always have access to a uniformly random oracle $H$, but we omit writing this. We have
\begin{align}
    &\Pr[(\mathcal{P}', \mathcal{F}_1', \mathcal{F}_2') \text{ win } H_1] \nonumber \\
    &= \frac13 \Pr\Big[\mathcal{F}_1^{\hat{G}_{x_1',w_1}}(\textsf{A}, x_1')= 0 \,\,\,\land \,\,\, \mathcal{F}_2^{\hat{G}_{x_2', w_2}}(\textsf{B}, x_2') = 0 \nonumber \\
    &\,\,\,\,\,\,\,\,\,\,\,\,\,\,\,\,\,: \textsf{AB} \leftarrow \rho, \,\,\, \rho \leftarrow \mathcal{P}^{\hat{G}}\left(\ket{v^{G(y)}}, H(v) \right), P_y \leftarrow D_y, v \leftarrow \{0,1\}^{m(\lambda)}, \nonumber\\
    & \,\,\,\,\,\,\,\,\,\,\,\,\,\,\,\,\,w_1, w_2 \leftarrow \{0,1\}^{m(\lambda)}, x_1',x_2' \leftarrow X_{\lambda}\,,\, G \leftarrow \text{Bool}(n,m(\lambda))\, ,\, \hat{G} \leftarrow \text{Bool}(n,m(\lambda))\Big] \nonumber\\   
    &\,\,\,\,+ \frac13 \Pr\Big[\mathcal{F}_1^{\hat{G}_{x_1',w_1}}(\textsf{A}, x_1')= 1 \,\,\,\land \,\,\, \mathcal{F}_2^{\hat{G}_{x_2', w_2}}(\textsf{B}, x_2') = 0 \nonumber\\
    &\,\,\,\,\,\,\,\,\,\,\,\,\,\,\,\,\,: \textsf{AB} \leftarrow \rho, \,\,\, \rho \leftarrow \mathcal{P}^{\hat{G}}\left(\ket{v^{G(y)}}, H(v) \right), P_y \leftarrow D_y, v \leftarrow \{0,1\}^{m(\lambda)}, \nonumber\\
    &\,\,\,\,\,\,\,\,\,\,\,\,\,\,\,\,\,w_1 \leftarrow G(y), w_2 \leftarrow \{0,1\}^{m(\lambda)}, x_1',x_2' \leftarrow X_{\lambda}\,,\, G \leftarrow \text{Bool}(n,m(\lambda))\, ,\, \hat{G} \leftarrow \text{Bool}(n,m(\lambda))\Big] \nonumber\\   
    &\,\,\,\,+ \frac13 \Pr\Big[\mathcal{F}_1^{\hat{G}_{x_1',w_1}}(\textsf{A}, x_1')= 0 \,\,\,\land \,\,\, \mathcal{F}_2^{\hat{G}_{x_2', w_2}}(\textsf{B}, x_2') = 1 \nonumber\\
    &\,\,\,\,\,\,\,\,\,\,\,\,\,\,\,\,\,: \textsf{AB} \leftarrow \rho, \,\,\, \rho \leftarrow \mathcal{P}^{\hat{G}}\left(\ket{v^{G(y)}}, H(v) \right), P_y \leftarrow D_y, v \leftarrow \{0,1\}^{m(\lambda)}, \nonumber\\
    & \,\,\,\,\,\,\,\,\,\,\,\,\,\,\,\,\,w_1 \leftarrow \{0,1\}^{m(\lambda)}, w_2 \leftarrow G(y), x_1',x_2' \leftarrow X_{\lambda}\,,\, G \leftarrow \text{Bool}(n,m(\lambda))\, ,\, \hat{G} \leftarrow \text{Bool}(n,m(\lambda))\Big] \,. \label{eq: 23}
\end{align}
For the next step, the key observation is that $\mathcal{P}$ only queries the oracle at $x_1'$ and $x_2'$ with negligible weight. Likewise, $\mathcal{F}_1$ only queries $x_2'$ with negligible weight, and $\mathcal{F}_2$ only queries $x_1'$ with negligible weight. The reason why this is true in this case is that, if it were true, $P$ could be used to construct an adversary that guesses a string sampled from $X_{\lambda}$ with non-negligible probability. But $X_{\lambda}$ has polynomial min-entropy. Thus, by an application of the one-way-to-hiding lemma, one can replace the current oracle accesses of $\mathcal{P}$, $\mathcal{F}_1$ and $\mathcal{F}_2$ with oracle access to the function $\hat{G}_{(x_1', w_1), (x_2',w_2)}$. Hence, we have, up to negligible quantities,
\begin{align}
    &\eqref{eq: 23}=\frac13 \Pr\Big[\mathcal{F}_1^{\hat{G}_{(x_1', w_1), (x_2',w_2)}}(\textsf{A}, x_1')= 0 \,\,\,\land \,\,\, \mathcal{F}_2^{\hat{G}_{(x_1', w_1), (x_2',w_2)}}(\textsf{B}, x_2') = 0 \nonumber\\
    &\,\,\,\,\,\,\,\,\,\,\,\,\,\,\,\,\,: \textsf{AB} \leftarrow \rho, \,\,\, \rho \leftarrow \mathcal{P}^{\hat{G}_{(x_1', w_1), (x_2',w_2)}}\left(\ket{v^{G(y)}}, H(v) \right), P_y \leftarrow D_y, v \leftarrow \{0,1\}^{m(\lambda)}, \nonumber\\
    &\,\,\,\,\,\,\,\,\,\,\,\,\,\,\,\,\,w_1, w_2 \leftarrow \{0,1\}^{m(\lambda)}, x_1',x_2' \leftarrow X_{\lambda}\,,\, G \leftarrow \text{Bool}(n,m(\lambda))\, ,\, \hat{G} \leftarrow \text{Bool}(n,m(\lambda))\Big] \nonumber\\   
    &\,\,\,\,+ \frac13 \Pr\Big[\mathcal{F}_1^{\hat{G}_{(x_1', w_1), (x_2',w_2)}}(\textsf{A}, x_1')= 1 \,\,\,\land \,\,\, \mathcal{F}_2^{\hat{G}_{(x_1', w_1), (x_2',w_2)}}(\textsf{B}, x_2') = 0 \nonumber\\
    &\,\,\,\,\,\,\,\,\,\,\,\,\,\,\,\,\,: \textsf{AB} \leftarrow \rho, \,\,\, \rho \leftarrow \mathcal{P}^{\hat{G}_{(x_1', w_1), (x_2',w_2)}}\Big(\ket{v^{w_1}}, H(v) \Big), y \leftarrow \{0,1\}^{\lambda}, v \leftarrow \{0,1\}^{m(\lambda)}, \nonumber\\
    & \,\,\,\,\,\,\,\,\,\,\,\,\,\,\,\,\,w_1 \leftarrow G(y), w_2 \leftarrow \{0,1\}^{m(\lambda)}, x_1',x_2' \leftarrow X_{\lambda}\,,\, G \leftarrow \text{Bool}(n,m(\lambda))\, ,\, \hat{G} \leftarrow \text{Bool}(n,m(\lambda))\Big] \nonumber\\   
    &\,\,\,\,+ \frac13 \Pr\Big[\mathcal{F}_1^{\hat{G}_{(x_1', w_1), (x_2',w_2)}}(\textsf{A}, x_1')= 0 \,\,\,\land \,\,\, \mathcal{F}_2^{\hat{G}_{(x_1', w_1), (x_2',w_2)}}(\textsf{B}, x_2') = 1 \nonumber\\
    &\,\,\,\,\,\,\,\,\,\,\,\,\,\,\,\,\,: \textsf{AB} \leftarrow \rho, \,\,\, \rho \leftarrow \mathcal{P}^{\hat{G}_{(x_1', w_1), (x_2',w_2)}}\Big(\ket{v^{w_2}}, H(v) \Big), P_y \leftarrow D_y, v \leftarrow \{0,1\}^{m(\lambda)}, \nonumber\\
    &\,\,\,\,\,\,\,\,\,\,\,\,\,\,\,\,\,w_1 \leftarrow \{0,1\}^{m(\lambda)}, w_2 \leftarrow G(y), x_1',x_2' \leftarrow X_{\lambda}\,,\, G \leftarrow \text{Bool}(n,m(\lambda))\, ,\, \hat{G} \leftarrow \text{Bool}(n,m(\lambda))\Big] \label{eq: 199}
\end{align}

For the second and third terms, it is convenient to further reprogram the oracle $\hat{G}$ at $y$, assigning as output a fresh uniformly random value in $\{0,1\}^{m(\lambda)}$. We can do this since none of the algorithms queries at $y$ with non-negligible weight. In this way, we are able to rewrite \eqref{eq: 199} more concisely, up to negligible terms, as:
\begin{align}
&\frac13 \Pr\Big[\mathcal{F}_1^{G}(\textsf{A}, x_1')= 0 \,\,\,\land \,\,\, \mathcal{F}_2^{G}(\textsf{B}, x_2') = 0 \nonumber\\
    &\,\,\,\,\,\,\,\,\,\,\,\,\,\,\,\,\,\,\,\,\,\,\,\,\,\,\,\,\,\,\,\,\,\,\,\,\,: \textsf{AB} \leftarrow \rho, \,\,\, \rho \leftarrow \mathcal{P}^{G}\left(\ket{v^{G(y)}}, H(v) \right), P_y \leftarrow D_y, v \leftarrow \{0,1\}^{m(\lambda)}, \nonumber\\
    & \,\,\,\,\,\,\,\,\,\,\,\,\,\,\,\,\,\,\,\,\,\,\,\,\,\,\,\,\,\,\,\,\,\,\,\,\,\,\,\,\,\,\,\,\, x_1',x_2' \leftarrow X_{\lambda}\,,\, G \leftarrow \text{Bool}(n,m(\lambda))\Big] \nonumber\\   
    &\,\,\,\,+ \frac13 \Pr\Big[\mathcal{F}_1^{G}(\textsf{A}, x_1')= 1 \,\,\,\land \,\,\, \mathcal{F}_2^{G}(\textsf{B}, x_2') = 0 \nonumber\\
    &\,\,\,\,\,\,\,\,\,\,\,\,\,\,\,\,\,\,\,\,\,\,\,\,\,\,\,\,\,\,\,\,\,\,\,\,\,: \textsf{AB} \leftarrow \rho, \,\,\, \rho \leftarrow \mathcal{P}^{G}\left(\ket{v^{G(x'_1)}}, H(v) \right), v \leftarrow \{0,1\}^{m(\lambda)}, \nonumber\\
    & \,\,\,\,\,\,\,\,\,\,\,\,\,\,\,\,\,\,\,\,\,\,\,\,\,\,\,\,\,\,\,\,\,\,\,\,\,\,\,\,\,\,\,\,\, x_1',x_2' \leftarrow X_{\lambda}\,,\, G \leftarrow \text{Bool}(n,m(\lambda))\Big] \nonumber\\   
    &\,\,\,\,+ \frac13 \Pr\Big[\mathcal{F}_1^{G}(\textsf{A}, x_1')= 1 \,\,\,\land \,\,\, \mathcal{F}_2^{G}(\textsf{B}, x_2') = 0 \nonumber\\
    &\,\,\,\,\,\,\,\,\,\,\,\,\,\,\,\,\,\,\,\,\,\,\,\,\,\,\,\,\,\,\,\,\,\,\,\,\,: \textsf{AB} \leftarrow \rho, \,\,\, \rho \leftarrow \mathcal{P}^{G}\left(\ket{v^{G(x'_2)}}, H(v) \right), v \leftarrow \{0,1\}^{m(\lambda)}, \nonumber\\
    & \,\,\,\,\,\,\,\,\,\,\,\,\,\,\,\,\,\,\,\,\,\,\,\,\,\,\,\,\,\,\,\,\,\,\,\,\,\,\,\,\,\,\,\,\, x_1',x_2' \leftarrow X_{\lambda}\,,\, G \leftarrow \text{Bool}(n,m(\lambda))\Big] \label{eq: 200}
\end{align}
Finally, notice that
\begin{align}
    \eqref{eq: 200} =&\,\frac13 \Pr (\mathcal{P},\mathcal{F}_1,\mathcal{F}_2) \text{ win } H_0 | \,\, \text{0-input, 0-input}]  \nonumber\\
    + &\,\frac13 \Pr (\mathcal{P},\mathcal{F}_1,\mathcal{F}_2) \text{ win } H_0 | \,\, \text{1-input, 0-input}] \nonumber\\
   +&\,\frac13 \Pr (\mathcal{P},\mathcal{F}_1,\mathcal{F}_2) \text{ win } H_0 | \,\, \text{0-input, 1-input}] \nonumber \\
   =&\, p \,,
\end{align}
which yields the desired result.

\end{document}